\documentclass[12pt,oneside,reqno]{amsart}
\PassOptionsToPackage{natbib=true}{biblatex}
\usepackage[T1]{fontenc}
\usepackage[utf8]{inputenc}
\usepackage{colortbl}
\usepackage{mathtools}
\usepackage{amstext}
\usepackage{amsthm}
\usepackage{amssymb}
\usepackage{geometry}
\usepackage{setspace}
\usepackage[pdfusetitle,
 bookmarks=true,bookmarksnumbered=false,bookmarksopen=false,
 breaklinks=false,pdfborder={0 0 1},backref=false,colorlinks=false]
 {hyperref}

\makeatletter

\providecommand{\tabularnewline}{\\}

\usepackage{colortbl}
\usepackage{float}
\usepackage{amstext}
\usepackage{amsthm}
\usepackage{amsfonts}
\usepackage{amstext}
\usepackage{amsthm}

\usepackage{tikz}
\usepackage{pgfplots}
\usepgfplotslibrary{groupplots}
\pgfplotsset{compat=1.18}

\theoremstyle{definition}

\usepackage{csquotes}   

\AddToHook{package/biblatex/after}{
  
  \renewbibmacro*{name:andothers}{
    \ifboolexpr{
      test {\ifnumequal{\value{listcount}}{\value{liststop}}}
      and
      test \ifmorenames
    }
      {\ifnumgreater{\value{liststop}}{1}
         {\finalandcomma}
         {}%
       \andothersdelim\bibstring[\emph]{andothers}}
      {}}
  
  \DeclareFieldFormat{pages}{#1}                                          
  \renewbibmacro{in:}{\ifentrytype{article}{}{\printtext{\bibstring{in}\intitlepunct}}}  
}

\hypersetup{colorlinks,
            bookmarksdepth=3,
            citecolor=[RGB]{46,126,42},
            linkcolor=[RGB]{128,0,6},
            urlcolor=[RGB]{138,0,135}}
\usepackage{bookmark}
\usepackage[hyphenbreaks]{breakurl}

\makeatother

\theoremstyle{plain}
\newtheorem{thm}{\protect\theoremname}
\newtheorem{assumption}{\protect\assumptionname}
\newtheorem{lem}{\protect\lemmaname}
\theoremstyle{remark}
\newtheorem{rem}{\protect\remarkname}
\theoremstyle{plain}
\newtheorem{cor}{\protect\corollaryname}
\newtheorem{prop}{\protect\propositionname}
\usepackage[style=authoryear,sorting=nyt,dashed=false,maxbibnames=999,maxcitenames=2,uniquename=init,uniquelist=false,giveninits=true,date=year]{biblatex}
\providecommand{\assumptionname}{Assumption}
\providecommand{\corollaryname}{Corollary}
\providecommand{\lemmaname}{Lemma}
\providecommand{\propositionname}{Proposition}
\providecommand{\remarkname}{Remark}
\providecommand{\theoremname}{Theorem}

\begin{document}
\title{Power enhancement via cross-fit variance estimation: Applications to specification, overidentification, and many-restriction testing}
\author{Keita Sunada}
\address{Aarhus Center for Econometrics, Aarhus University, Universitetsbyen 51, 8000 Aarhus C, Denmark.}
\email{\href{mailto:ksunada@econ.au.dk}{ksunada@econ.au.dk}}
\author{Yukitoshi Matsushita}
\address{Graduate School of Economics, Hitotsubashi University, 2-1 Naka, Kunitachi, Tokyo 186-8601, Japan.}
\email{\href{mailto:matsushita.y@r.hit-u.ac.jp}{matsushita.y@r.hit-u.ac.jp}}
\author{Taisuke Otsu}
\address{Department of Economics, London School of Economics, Houghton Street, London, WC2A 2AE, UK.}
\email{\href{mailto:t.otsu@lse.ac.uk}{t.otsu@lse.ac.uk}}
\thanks{We are grateful to Mikkel Sølvsten and Luke Taylor for helpful comments. Matsushita gratefully acknowledges financial support from JSPS KAKENHI (23K01331); Sunada gratefully acknowledges research support from the Aarhus Center for Econometrics (ACE), funded by the Danish National Research Foundation grant number DNRF186.}
\begin{abstract}
Quadratic-form test statistics are widely used in econometrics, and their performance depends on accurate variance estimation. Conventional plug-in estimators are consistent under the null hypothesis, but under alternatives the drift in the residuals inflates them and the test loses power. We develop a general framework for variance estimation in such statistics, replacing one of the two squared-residual factors by an auxiliary linear combination of the residuals (``cross-fitting'') chosen to annihilate the drift. We characterize the conditional bias of each estimator exactly. The drift enters the plug-in estimator squared, multiplied by quantities bounded away from zero, so its bias is positive whenever the drift is non-degenerate. It reaches the cross-fit estimator only through the part that survives the cross-fitting, and then only through off-diagonal entries of a residual-maker matrix. From this calculation we obtain conditions under which the cross-fit estimator remains consistent under alternatives while the plug-in estimator does not. At a common critical value the cross-fit test therefore rejects whenever the plug-in test does, and against distant alternatives the plug-in statistic converges to a finite limit, small when few observations carry the departure: its power can tend to zero where the cross-fit test's tends to one. We verify the conditions under primitive assumptions in nonparametric specification testing, overidentification testing, and testing many linear restrictions, and illustrate the procedure on the Oregon Health Insurance Experiment.
\end{abstract}

\maketitle

\section{Introduction} \label{sec:intro}

Quadratic-form test statistics play a central role in econometrics and arise in a wide range of applications, including nonparametric specification testing \citep[e.g.][]{HongWhite1995,DonaldImbensNewey2003,SunLi2006}, overidentification testing in instrumental variables models \citep[e.g.][]{AnatolyevGospodinov2011,lee2012hahn,ChaoHausmanNeweySwansonWoutersen2014}, and tests of many linear restrictions \citep[e.g.][]{CattaneoJanssonNewey2018,AnatolyevSolvsten2023}. In these procedures, statistical significance is determined by studentized test statistics, making accurate variance estimation essential for reliable inference. Variance estimators are commonly constructed by plugging residual-based estimates into the population variance formula, and are justified by their consistency under the null hypothesis. Each of these literatures has accordingly been built around valid inference under the null and the asymptotic null distribution.

What such an estimator does under alternatives is considerably less understood, and it can have important consequences for power. Under fixed or local alternatives, the drift component of the residuals enters the variance estimator and can generate substantial variance inflation. As a result, the studentized test statistic may suffer a considerable loss of power even when the underlying quadratic-form statistic itself exhibits a strong signal. This issue becomes particularly pronounced in modern settings where the signal accumulates over many observations or restrictions.

We construct variance estimators for quadratic-form statistics that remain valid under the null while being substantially less sensitive to deterministic drift under alternatives. The construction replaces one of the two squared-residual factors by an auxiliary linear combination of the residuals (``cross-fitting'') chosen to annihilate the drift. Because the drift then survives in only one of the two factors, it can reach the variance estimate only through a second-order covariance, whereas it enters the plug-in estimator directly. This is the construction that \citet{MikushevaSun2022} introduced for the jackknife Anderson--Rubin statistic in a linear instrumental-variables model with many weak instruments; the framework below shows that neither the construction nor the power gain it delivers is specific to that setting.

We first compute the conditional bias of each estimator under the alternative exactly. The two formulas differ in where the systematic component of the residual can appear. The drift reaches the plug-in estimator through products of diagonal quantities, which are bounded away from zero. In contrast, it reaches the cross-fit estimator only through the part of the systematic component that survives the cross-fitting, and, when that part vanishes, only through off-diagonal entries of a residual-maker matrix. The plug-in bias is in addition nonnegative, so it inflates the estimated variance, while the cross-fit bias need not be. From this calculation we obtain conditions on the weights and on the residual decomposition under which the cross-fit estimator is consistent for the asymptotic variance under both hypotheses and the plug-in estimator is inconsistent under the alternative. The condition behind the cross-fit consistency asks that the weights of the quadratic form not be concentrated on a few pairs; the condition behind the plug-in inconsistency asks that the systematic component be non-degenerate
relative to the weights.

Referred to a common critical value, the two studentized statistics differ asymptotically by the factor $\sqrt{1+B_n/V}$, where $B_n$ is the plug-in bias and $V$ the asymptotic variance, so the cross-fit test rejects whenever the plug-in test does with probability approaching one. Also, along local alternatives $B_{n}/V\to0$ and the two tests have the same limiting power. However, against a distant alternative, the plug-in bias grows as fast as the square of the numerator, so the plug-in statistic converges to a finite limit. That limit is invariant to the scale of the systematic component, and it is bounded by the number of observations on which that component is nonzero. When that bound falls below the critical value, the plug-in statistic never reaches it however large the departure, and the plug-in test has power tending to zero where the cross-fit test has power tending to one.

We verify the conditions of the general theory under primitive assumptions in three important econometric applications: nonparametric specification testing, overidentification testing in instrumental variables models, and testing many linear restrictions in linear regression. In each, the asymptotic results follow as direct consequences of the general theory. Although these testing problems arise in different econometric settings, in each of them the drift enters the plug-in estimator squared; what cross-fitting removes differs across the three, as Remark \ref{rem:appbias} sets out. Monte Carlo experiments demonstrate that the proposed procedure keeps size close to nominal over most of the range, turning mildly liberal only in the smallest-sample designs, and delivers substantial power gains relative to conventional plug-in variance estimators. The gains can be substantial when the number of restrictions or instruments is large relative to the sample size, and are largest for many restrictions, where the plug-in procedure also grows more conservative under the null as the number of restrictions increases. We further illustrate the practical relevance of the proposed approach through an empirical analysis of treatment effect heterogeneity in the Oregon Health Insurance Experiment.

Our work is related to the literature on cross-fitting, sample splitting, and leave-out estimation. Cross-fitting and sample splitting have become standard tools for eliminating overfitting bias and obtaining valid inference with flexible first-stage estimators \citep{NeweyRobins2018,Chernozhukov2018}. Their use for variance estimation is a currently active area: leave-out and bias-corrected variance estimators have been developed for high-dimensional linear regression and quadratic-form inference \citep{KlineSaggioSolvsten2020,Jochmans2022}, for testing many linear restrictions under heteroskedasticity \citep{AnatolyevSolvsten2023}, and more recently for quadratic forms under clustered sampling \citep{KolesarMinWangZhang2026}.

Closest to this paper is \citet{MikushevaSun2022}, who establish for the jackknife Anderson--Rubin statistic that the cross-fit estimator of the scale parameter remains consistent under local alternatives, and document the power loss that the plug-in estimator suffers. The present paper generalizes that result in four directions: general weighting arrays and auxiliary vectors in place of a single statistic, an estimated nuisance parameter in the residual, alternatives beyond their local-drift regime, and an exact conditional bias with a nonnegative lower bound for the plug-in estimator in place of an approximation. See Remark \ref{rem:ms22} for the details.

The obstacle the leave-out literature confronts is that an unbiased estimator of $\sigma^{2}_{i}\sigma^{2}_{j}$ built from leave-out residuals is not invariant to the parameterization of the regression, and the literature has taken three routes around it. First, \citet[Remark 3]{AnatolyevSolvsten2023} accept the non-invariance, and pay for it with a leave-three-out construction. Second, \citet{AnatolyevKorobka2026} restore invariance by splitting the sample, at the cost of a bound on how many restrictions may be tested --- $q/n\le1/4$ for the pairwise products.\footnote{\citet[Section 3]{Jochmans2022} explains why the first two routes diverge: an invariant cross-fit estimator needs two disjoint subsamples that each identify the parameter, which forces a bound on the number of restrictions.} We take a third route: impose the restriction under test, which removes the part of the contamination that the restriction identifies, and characterize what remains.\footnote{See also \citet{Boot2023} for a different route to power under many restrictions: joint confidence regions centered at James--Stein averaging estimators, which gain power against sparse alternatives when the restricted estimator points at the right coefficients and lose it when it points at the wrong ones. Ours asks for no such choice, since we change only the variance estimator and leave the statistic alone, and Boot's heteroskedasticity-robust variant uses the same leave-out variance estimator that the present paper replaces.} \citet[Remark 4]{AnatolyevSolvsten2023} point in this direction, and we investigate it in the setting where imposing the restriction leaves only a fixed-dimensional parameter to estimate. Under this narrower setup the resulting test delivers moderately higher power than theirs in our simulations and avoids their leave-three-out construction, which makes it far cheaper to compute.

The paper is organized as follows. Section \ref{sec:gen} develops the general theory, Section \ref{sec:app} the three applications, and Sections \ref{sec:sim} and \ref{sec:emp} the simulation and empirical evidence; the Appendix contains all proofs.

\section{General result} \label{sec:gen}

Let $\{\hat{u}_{i}\}^{n}_{i=1}$ denote estimated residuals or moment conditions. We consider statistics of the form 
\begin{equation}
Q_{n}=\sum^{n}_{i=1}\sum_{j\neq i}a_{ij}\hat{u}_{i}\hat{u}_{j},\label{eq:T}
\end{equation}
where $\{a_{ij}\}$ is an array of weights that may depend on underlying variables. Such statistics arise in a variety of econometric applications, including specification testing, overidentification testing in instrumental variable models, and tests involving many restrictions. We suppose that $\hat u_i$ admits the decomposition 
\begin{equation}
\hat{u}_{i}=\Delta_{i}+e_{i}+\rho_{i},\label{eq:uhat}
\end{equation}
where $\Delta_{i}$ denotes the leading drift component under the alternative hypothesis, $e_{i}$ is a mean-zero stochastic error, and $\rho_{i}$ is an asymptotically negligible remainder. Under the null hypothesis, $\Delta_{i}=0$. Under the alternative hypothesis, the drift component may be non-negligible and can materially affect the behavior of variance estimators.

Let $\sigma^2_i=\mathbb{E}[e^{2}_{i}]$ and
\begin{equation}
V=2\sum^{n}_{i=1}\sum_{j\neq i}a^{2}_{ij}\sigma^2_i \sigma^2_j\label{eq:V}
\end{equation}
denote the leading asymptotic variance component of $Q_{n}$. A conventional plug-in estimator of $V$ is 
\begin{equation}
\hat{V}=2\sum^{n}_{i=1}\sum_{j\neq i}a^{2}_{ij}\hat{u}^{2}_{i}\hat{u}^{2}_{j}.\label{eq:Vhat}
\end{equation}
Substituting \eqref{eq:uhat} into \eqref{eq:Vhat} introduces drift-dependent terms under the alternative hypothesis. In particular, terms of the form $\Delta^{2}_{i}\sigma^{2}_{j}$ and $\Delta^{2}_{i}\Delta^{2}_{j}$ inflate the estimated variance and attenuate the resulting studentized statistic.

To reduce this contamination, we consider a cross-fit variance estimator. Let $\Delta=(\Delta_{1},\ldots,\Delta_{n})^{\prime}$, $e=(e_1,\ldots,e_n)'$, $\rho=(\rho_1,\ldots,\rho_n)'$, and $\ell_{i}$ be an auxiliary vector constructed so that $\ell_i'\Delta=0$. We propose the following cross-fit variance estimator: 
\begin{equation}
\hat{V}_{\mathrm{cf}}=\sum^{n}_{i=1}\sum_{j\neq i}c_{ij}\left\{ \hat{u}_{i}(\ell^{\prime}_{i}\hat{u})\right\} \left\{ \hat{u}_{j}(\ell^{\prime}_{j}\hat{u})\right\}, \label{eq:Vcf}
\end{equation}
where the weights $\{c_{ij}\}$ are chosen to remove the bias. Let $L$ be the matrix with rows $\ell^{\prime}_{i}$. If the $e_{i}$ are independent with mean zero, then for $i\neq j$,
\[
\mathbb{E}\bigl[e_{i}(\ell^{\prime}_{i}e)e_{j}(\ell^{\prime}_{j}e)\bigr]=(L_{ii}L_{jj}+L_{ij}L_{ji})\,\sigma^2_i \sigma^2_j,
\]
so that taking
\[
c_{ij}=\frac{2a^{2}_{ij}}{L_{ii}L_{jj}+L_{ij}L_{ji}}
\]
makes $\mathbb{E}[\hat{V}_{\mathrm{cf}}]=V$ exactly when $\Delta=\rho=0$. 

The two estimators differ in where the drift can appear. By squaring the residuals, $\hat{V}$ carries $\Delta^{2}_{i}$, which is nonnegative, accumulates over $i$, and cannot cancel. In contrast, pairing $\hat{u}_{i}$ with $\ell^{\prime}_{i}\hat{u}$ removes that square, since $\ell^{\prime}_{i}\Delta=0$; products $\Delta_{i}\Delta_{j}$ survive for $i\neq j$, each carrying an off-diagonal weight built from $L$. In the applications, $L$ is a residual-maker matrix $I-P$, so these weights are bounded by the maximal leverage $\max_{i}P_{ii}$. Lemma \ref{lem:bias} makes the comparison exact.

In applications, we decompose $\rho_i$ as
\begin{equation}
\rho_{i}=\omega_{i}+H^{\prime}_{i}\hat{\delta},\label{eq:split}
\end{equation}
where $\omega_{i}$ is an approximation error that does not involve the estimation error, $H=(H_{1},\ldots,H_{n})^{\prime}$ is an $n\times p$ matrix, and $\hat{\delta}$ is an estimation error with $\|\hat{\delta}\|=O_{p}(r_{n})$ for a sequence $r_{n}\to0$. How large $\omega_{i}$ is depends on the quality of the approximation, and how large $H^{\prime}_{i}\hat{\delta}$ is depends on the convergence rate $r_{n}$ of the estimator. 

\textbf{Conditioning}. The arrays $\{a_{ij}\}$ and $\{c_{ij}\}$ and the vectors $\{\ell_{i}\}$ may depend on underlying variables, typically covariates. Let $\mathcal{F}$ denote the $\sigma\text{-field}$ generated by these variables together with $\Delta$ and $\omega$. We assume that the errors are independent across $i$ conditionally on $\mathcal{F}$ with $\mathbb{E}[e_{i}\mid\mathcal{F}]=0$, and every expectation, order of magnitude and limit below is conditional on $\mathcal{F}$. In each application $\mathcal{F}$ is generated by the covariates and instruments, which makes the approximation error $\omega$ measurable. When the design is random, as it is in two of the three applications, conditions stated below as $o(\cdot)$ or $O(\cdot)$ are to be understood as holding in probability, and the conclusions then hold in probability as well. The matrix $H$ need not be $\mathcal{F}$-measurable. For example, in the overidentification application below, $H_{i}=x_{i}$ contains the first-stage error, which is correlated with $e_{i}$ by construction. It enters only through the stochastic-order conditions on the estimation error.

\textbf{Notation}. We take the weights to be symmetric, $a_{ij}=a_{ji}$, as they are in all three applications. This is a normalization rather than a restriction because replacing $a_{ij}$ by $(a_{ij}+a_{ji})/2$ leaves $Q_{n}$ unchanged. Throughout, $\|\cdot\|$ is the Euclidean norm for vectors and the spectral norm for matrices.

\medskip{}
The following assumption is maintained throughout this section. Let $\omega=(\omega_{1},\ldots,\omega_{n})^{\prime}$, $A_{n}=\sum^{n}_{i=1}\sum_{j\neq i}a^{2}_{ij}$, and $C_{n}=\sum^{n}_{i=1}\sum_{j\neq i}|c_{ij}|$.

\begin{assumption} \label{as:gen}
The variables $\{e_{i}\}^{n}_{i=1}$ are, conditionally on $\mathcal{F}$, independent mean-zero random variables satisfying $c\le\sigma^2_i \le C$ and $\max_{1\le i\le n}\mathbb{E}[e^{4}_{i}]\le C$ for some constants $0<c<C<\infty$. Moreover, the leverage condition $L_{ii}L_{jj}+L_{ij}L_{ji}\ge c_{L}>0$ holds for all $i\neq j$, so that $C_{n}\le(2/c_{L})A_{n}$, and there exists a deterministic sequence $b_{n}\ge1$ such that 
\begin{align*}
 & \max_{1\le i\le n}|e_{i}|=O_{p}(b_{n}),\qquad\max_{1\le i\le n}|\ell^{\prime}_{i}e|=O_{p}(b_{n}),\\
 & \max_{1\le i\le n}|\omega_{i}|+\max_{1\le i\le n}|\ell^{\prime}_{i}\omega|=o_{p}(b^{-3}_{n}),\qquad\|\hat{\delta}\|=O_{p}(r_{n}).
\end{align*}
\end{assumption}

The assumption requires $\omega$ to be small uniformly in $i$, but requires only a rate for $\hat{\delta}$. By \eqref{eq:split} the estimation error enters into both factors of $\hat{u}_{i}(\ell^{\prime}_{i}\hat{u})$: it enters $\hat{u}_{i}$ through $H^{\prime}_{i}\hat{\delta}$ and $\ell^{\prime}_{i}\hat{u}$ through $\ell^{\prime}_{i}H\hat{\delta}$. Each factor is therefore of first degree in $\hat{\delta}$ and  their product is of second degree, so $\hat{V}_{\mathrm{cf}}$, which pairs two such products, is a polynomial of degree four whose coefficients are sums over $i\neq j$. The next lemma writes that polynomial explicitly.
\begin{lem} \label{lem:exp}
Let $\check{u}_{i}=\Delta_{i}+e_{i}+\omega_{i}$ be the residual stripped of estimation error, $\check{u}=(\check{u}_{1},\ldots,\check{u}_{n})^{\prime}$, and $h_{i}=\ell^{\prime}_{i}H$, and define
\[
g_{0i}=\check{u}_{i}(\ell^{\prime}_{i}\check{u}),\qquad g_{1i}=\check{u}_{i}h^{\prime}_{i}+(\ell^{\prime}_{i}\check{u})H_{i},\qquad g_{2i}=\tfrac{1}{2}(H_{i}h_{i}+h^{\prime}_{i}H^{\prime}_{i}).
\]
Then $\hat{u}_{i}(\ell^{\prime}_{i}\hat{u})=g_{0i}+g^{\prime}_{1i}\hat{\delta}+\hat{\delta}^{\prime}g_{2i}\hat{\delta}$ and
\begin{equation}
\hat{V}_{\mathrm{cf}}=\sum^{n}_{i=1}\sum_{j\neq i}c_{ij}g_{0i}g_{0j}+\mathcal{G}^{\prime}_{1}\hat{\delta}+\hat{\delta}^{\prime}\mathcal{G}_{2}\hat{\delta}+\mathcal{R}_{n},\label{eq:expansion}
\end{equation}
 where
\[
\mathcal{G}_{1}=\sum^{n}_{i=1}\sum_{j\neq i}c_{ij}(g_{0i}g_{1j}+g_{0j}g_{1i}),\qquad\mathcal{G}_{2}=\sum^{n}_{i=1}\sum_{j\neq i}c_{ij}(g_{1i}g^{\prime}_{1j}+g_{0i}g_{2j}+g_{0j}g_{2i}),
\]
and the remainder $\mathcal{R}_{n}$, which collects the terms of degree three and four in $\hat{\delta}$, satisfies $|\mathcal{R}_{n}|\le\|\hat{\delta}\|^{3}\mathcal{G}_{3}+\|\hat{\delta}\|^{4}\mathcal{G}_{4}$ for
\[
\mathcal{G}_{3}=\sup_{\|\delta\|=1}\Bigl|\sum^{n}_{i=1}\sum_{j\neq i}c_{ij}\bigl\{(g^{\prime}_{1i}\delta)(\delta^{\prime}g_{2j}\delta)+(g^{\prime}_{1j}\delta)(\delta^{\prime}g_{2i}\delta)\bigr\}\Bigr|,\qquad\mathcal{G}_{4}=\sup_{\|\delta\|=1}\Bigl|\sum^{n}_{i=1}\sum_{j\neq i}c_{ij}(\delta^{\prime}g_{2i}\delta)(\delta^{\prime}g_{2j}\delta)\Bigr|.
\]
If $\ell^{\prime}_{i}H=0$ for every $i$, then $g_{2i}=0$, $g_{1i}=(\ell^{\prime}_{i}\check{u})H_{i}$, and $\mathcal{R}_{n}=0$.
\end{lem}

Under the law of large numbers \eqref{eq:cond} for the errors and the negligibility conditions \eqref{eq:condest} for the estimation error, the cross-fit estimator is consistent under the null. Both are verified for each application considered in Section \ref{sec:app}.
\begin{thm} \label{thm:null}
Suppose Assumption \ref{as:gen} holds and, under $H_{0}$, (\ref{eq:uhat}) holds with $\Delta_{i}=0$. Assume further that 
\begin{equation}
\sum^{n}_{i=1}\sum_{j\neq i}\left\{ c_{ij}e_{i}(\ell_{i}'e)e_{j}(\ell_{j}'e)-2a^{2}_{ij}\sigma^2_i \sigma^2_j\right\} =o_{p}(V).\label{eq:cond}
\end{equation}
In addition, the estimation error is negligible for the variance:
\begin{equation}
\|\hat{\delta}\|\,\|\mathcal{G}_{1}\|=o_{p}(V),\quad\|\hat{\delta}\|^{2}\,\|\mathcal{G}_{2}\|=o_{p}(V),\quad\|\hat{\delta}\|^{3}\mathcal{G}_{3}=o_{p}(V),\quad\|\hat{\delta}\|^{4}\mathcal{G}_{4}=o_{p}(V),\label{eq:condest}
\end{equation}
where $\mathcal{G}_{1}$ and $\mathcal{G}_{2}$ are defined in Lemma \ref{lem:exp}. Then
\[
\hat{V}_{\mathrm{cf}}=V+o_{p}(V).
\]
\end{thm}

Under the alternative the residual has a systematic component. Denote $\pi_{i}=\Delta_{i}+\omega_{i}$ for that systematic component and $\lambda_{i}=\ell^{\prime}_{i}\omega$ for the part of it that survives the auxiliary factor $\ell^{\prime}_{i}$. Then the two factors of $\hat{u}_{i}(\ell^{\prime}_{i}\hat{u})$ can be written as
\[
\hat{u}_{i}=\underbrace{\Delta_{i}+\omega_{i}}_{\pi_{i}}+e_{i}+H^{\prime}_{i}\hat{\delta},\qquad\ell^{\prime}_{i}\hat{u}=\underbrace{\ell^{\prime}_{i}\Delta}_{0}+\underbrace{\ell^{\prime}_{i}\omega}_{\lambda_{i}}+\ell^{\prime}_{i}e+\ell^{\prime}_{i}H\hat{\delta},
\]
the drift dropping out of the second because $\ell^{\prime}_{i}\Delta=0$. Let $\Sigma=\operatorname{diag}(\sigma^{2}_{1},\ldots,\sigma^{2}_{n})$ and $\mu_{3i}=\mathbb{E}[e^{3}_{i}]$. In this notation the residual of Lemma \ref{lem:exp} is $\check{u}_{i}=\pi_{i}+e_{i}$, and we write 
\[
\check{V} =2\sum^{n}_{i=1}\sum_{j\neq i}a^{2}_{ij}\check{u}^{2}_{i}\check{u}^{2}_{j},\qquad
\check{V}_{\mathrm{cf}} =\sum^{n}_{i=1}\sum_{j\neq i}c_{ij}\left\{ \check{u}_{i}(\ell^{\prime}_{i}\check{u})\right\} \left\{ \check{u}_{j}(\ell^{\prime}_{j}\check{u})\right\} ,
\]
which are \eqref{eq:Vhat} and \eqref{eq:Vcf} with $\hat u$ replaced by $\check u$ (i.e., $\hat\delta=0$) throughout.
\begin{lem} \label{lem:bias}
Suppose \eqref{eq:uhat} holds under $H_{1}$ with $\ell^{\prime}_{i}\Delta=0$, and let $\{e_{i}\}$ be independent with mean zero and bounded third moments. Then
\begin{equation}
B_n := \mathbb{E}[\check{V}]-V=2\sum^{n}_{i=1}\sum_{j\neq i}a^{2}_{ij}\bigl\{\pi^{2}_{i}\sigma^{2}_{j}+\sigma^{2}_{i}\pi^{2}_{j}+\pi^{2}_{i}\pi^{2}_{j}\bigr\},\label{eq:biasplug}
\end{equation}
and
\begin{align}
B^{\mathrm{cf}}_{n} := \mathbb{E}[\check{V}_{\mathrm{cf}}]-V & =\sum^{n}_{i=1}\sum_{j\neq i}c_{ij}\Bigl\{\pi_{i}\lambda_{i}\pi_{j}\lambda_{j}+\pi_{i}\lambda_{i}L_{jj}\sigma^{2}_{j}+L_{ii}\sigma^{2}_{i}\pi_{j}\lambda_{j}\label{eq:biascf}\\
 & \qquad+\pi_{i}\pi_{j}(L\Sigma L^{\prime})_{ij}+\pi_{i}\lambda_{j}L_{ij}\sigma^{2}_{j}+\lambda_{i}\pi_{j}L_{ji}\sigma^{2}_{i}\nonumber \\
 & \qquad+\pi_{i}L_{ij}L_{jj}\mu_{3j}+\pi_{j}L_{ji}L_{ii}\mu_{3i}\Bigr\}.\nonumber 
\end{align}
\end{lem}

The plug-in squares $\hat{u}_{i}$, so the systematic component multiplies itself and enters the bias as $\pi^{2}_{i}$. The cross-fit multiplies $\hat{u}_{i}$ by $\ell^{\prime}_{i}\hat{u}$, from which $\ell^{\prime}_{i}\Delta=0$ has already removed the drift; the systematic component of that factor is $\lambda_{i}$, and the bias carries $\pi_{i}\lambda_{i}$.

Every term of \eqref{eq:biasplug} is nonnegative, so with $\sigma^{2}_{i}\ge c$ the plug-in bias is bounded below by $4c\sum^{n}_{i=1}\sum_{j\neq i}a^{2}_{ij}\pi^{2}_{i}$ and is strictly positive whenever $\sum^{n}_{i=1}\sum_{j\neq i}a^{2}_{ij}\pi^{2}_{i}>0$. The cross-fit bias need not be positive, since its terms carry $\pi_{i}\lambda_{i}$ and $\mu_{3i}$, whose signs are unrestricted, so no analogous lower bound is available.

The first two terms of \eqref{eq:biasplug} are quadratic in $\pi$ and the last is quartic. The quadratic pair dominates for small $\pi$: with $\pi_{i}\equiv\pi$ and $\sigma_{i}\equiv\sigma$ they contribute $2\pi^{2}\sigma^{2}$ against $\pi^{4}$, so the two are comparable at $\pi^{2}=2\sigma^{2}$. Under local alternatives, along which $\pi_{i}\to0$, the quartic term is the first to vanish and the quadratic one governs.

\begin{rem} \label{rem:lambda0}
When $\ell_i'\pi=0$ for every $i$ so that $\lambda\equiv0$, the bias reduces to $\sum^{n}_{i=1}\sum_{j\neq i}c_{ij}\{\pi_{i}\pi_{j}(L\Sigma L^{\prime})_{ij}+\pi_{i}L_{ij}L_{jj}\mu_{3j}+\pi_{j}L_{ji}L_{ii}\mu_{3i}\}$. Every surviving term carries an off-diagonal factor, $(L\Sigma L^{\prime})_{ij}$ or $L_{ij}$, and the two third-moment terms, being linear rather than quadratic in $\pi$, dominate when the drift is small. The maintained assumption bounds $L_{ii}L_{jj}+L_{ij}L_{ji}$ away from zero but places no upper bound on the off-diagonal entries, so this reduction by itself does not make the bias negligible. Below, we impose a condition on those entries explicitly. $\blacksquare$
\end{rem}

The next theorem is the main theoretical result of the paper. Under the alternative the cross-fit variance estimator remains consistent for $V$ while the plug-in estimator converges to $V+B_{n}$, and $B_{n}$ is bounded away from zero relative to $V$ under a non-degeneracy condition on the systematic component.
\begin{thm} \label{thm:alt}
Suppose Assumption \ref{as:gen} holds and, under $H_{1}$, \eqref{eq:uhat} holds with $\ell^{\prime}_{i}\Delta=0$. Let $\check{A}_{i}=\check{u}_{i}(\ell^{\prime}_{i}\check{u})$ and assume that
\begin{align}
 & \sum^{n}_{i=1}\sum_{j\neq i}c_{ij}\bigl\{\check{A}_{i}\check{A}_{j}-\mathbb{E}[\check{A}_{i}\check{A}_{j}]\bigr\}=o_{p}(V),\label{eq:condalt}\\
 & 2\sum^{n}_{i=1}\sum_{j\neq i}a^{2}_{ij}\bigl\{\check{u}^{2}_{i}\check{u}^{2}_{j}-\mathbb{E}[\check{u}^{2}_{i}]\mathbb{E}[\check{u}^{2}_{j}]\bigr\}=o_{p}(V+B_{n}),\nonumber \\
 & \hat{V}_{\mathrm{cf}}-\check{V}_{\mathrm{cf}}=o_{p}(V),\qquad\hat{V}-\check{V}=o_{p}(V+B_{n}),\qquad B^{\mathrm{cf}}_{n}=o(V).\nonumber 
\end{align}
Then 
\[
\frac{\hat{V}_{\mathrm{cf}}}{V}\overset{p}{\to}1,\qquad\frac{\hat{V}}{V+B_{n}}\overset{p}{\to}1.
\]
If in addition $\liminf_{n}A^{-1}_{n}\sum^{n}_{i=1}\sum_{j\neq i}a^{2}_{ij}\pi^{2}_{i}>0$, then $\liminf_{n}B_{n}/V>0$, so the plug-in estimator is not consistent for $V$.
\end{thm}

The first two conditions of \eqref{eq:condalt} are laws of large numbers for $\check{V}_{\mathrm{cf}}$ and $\check{V}$; the next two control the estimation error, and are the alternative-hypothesis counterparts of \eqref{eq:condest}; the last asks that the cross-fit bias be negligible. By Lemma \ref{lem:exp}, $\hat{V}_{\mathrm{cf}}-\check{V}_{\mathrm{cf}}=\mathcal{G}^{\prime}_{1}\hat{\delta}+\hat{\delta}^{\prime}\mathcal{G}_{2}\hat{\delta}+\mathcal{R}_{n}$, so that the difference is $o_{p}(V)$ whenever
$\|\hat{\delta}\|\|\mathcal{G}_{1}\|=o_{p}(V)$,
$\|\hat{\delta}\|^{2}\|\mathcal{G}_{2}\|=o_{p}(V)$,
$\|\hat{\delta}\|^{3}\mathcal{G}_{3}=o_{p}(V)$, and
$\|\hat{\delta}\|^{4}\mathcal{G}_{4}=o_{p}(V)$. Its plug-in counterpart $\hat{V}-\check{V}=o_{p}(V+B_{n})$ is the same statement for $\hat{V}$. When $\ell^{\prime}_{i}H=0$ the cross-fit expansion has $\mathcal{R}_{n}=0$, so only $\mathcal{G}_{1}$ and $\mathcal{G}_{2}$ need bounding. However, the plug-in expansion keeps its degree-three and degree-four terms, which must be bounded as well. Lemma \ref{lem:lln} gives aggregate bounds that imply the first two conditions of \eqref{eq:condalt} and $B^{\mathrm{cf}}_{n}=o(V)$, which the three applications verify below.

\begin{rem} \label{rem:ms22}
\citet{MikushevaSun2022} propose the jackknife Anderson--Rubin statistic in a linear instrumental-variables model with many weak instruments; their Theorem 3 shows that the cross-fit estimator of the scale parameter remains consistent under local alternatives, and their Section 4.2 documents the resulting power loss of the plug-in estimator. Theorem \ref{thm:alt} extends that result in four directions. First, it covers general weighting arrays $\{a_{ij}\}$ and auxiliary vectors $\{\ell_{i}\}$ rather than a single statistic. Second, it allows the residual to depend on an estimated nuisance parameter through $H^{\prime}_{i}\hat{\delta}$, which their implied error does not. The terms $\mathcal{G}_{1},\ldots,\mathcal{G}_{4}$ of Lemma \ref{lem:exp} control that dependence. Third, it asks only that the cross-fit bias be negligible, in place of the upper bound on the drift relative to $A_{n}$ that their local-alternative condition imposes. Finally, it establishes the inconsistency of the plug-in estimator formally; their Section 4.2 argues the point informally. The comparison extends to the bias expressions. Their leading term for the plug-in estimator is quartic in the drift alone, where \eqref{eq:biasplug} adds a quadratic one that dominates it at moderate drift. For the cross-fit estimator they establish consistency but no bias formula, where \eqref{eq:biascf} is exact and has a third-moment term that leads when the drift is small. $\blacksquare$
\end{rem}

Write $T_{n}=Q_{n}/\sqrt{\hat{V}}$ and $T^{\mathrm{cf}}_{n}$ for the same quantity with $\hat{V}_{\mathrm{cf}}$ in the denominator. The next result shows that, compared against a common critical value, the two-sided test based on $T^{\mathrm{cf}}_{n}$ has a rejection region containing that based on $T_{n}$, and its rejection probability is at least as large.
\begin{cor} \label{cor:pow}
Under the conditions of Theorem \ref{thm:alt},
\[
\frac{T^{\mathrm{cf}}_{n}}{T_{n}}=\sqrt{\frac{\hat{V}}{\hat{V}_{\mathrm{cf}}}}=\sqrt{1+\frac{B_{n}}{V}}\bigl\{1+o_{p}(1)\bigr\}.
\]
If in addition $\liminf_{n}B_{n}/V\ge\gamma>0$, then $P(\hat{V}\ge\hat{V}_{\mathrm{cf}})\to1$, so that $|T^{\mathrm{cf}}_{n}|\ge|T_{n}|$ with probability approaching one. 
\end{cor}

Corollary \ref{cor:pow} compares the two procedures at a common critical value. Under size correction the comparison need not favor either procedure. The simulations favor cross-fitting either way: at a common critical value in Section \ref{sec:sim}, and after calibration in Section \ref{sec:calibrated} and Appendix \ref{subsec:Size-corrected-power}, where most of the gap remains and it closes only against the smoother alternative at the smallest samples.

\begin{rem} \label{rem:gain}
By Corollary \ref{cor:pow} the gain from cross-fitting is governed by $B_{n}/V$, and this ratio depends on how the systematic component is distributed across observations, not only on its size. Dividing \eqref{eq:biasplug} by \eqref{eq:V} gives
\[
\frac{B_{n}}{V}=\frac{2\sum^{n}_{i=1}\pi^{2}_{i}\sum_{j\neq i}a^{2}_{ij}\sigma^{2}_{j}}{\sum^{n}_{i=1}\sigma^{2}_{i}\sum_{j\neq i}a^{2}_{ij}\sigma^{2}_{j}}+\frac{\sum^{n}_{i=1}\sum_{j\neq i}a^{2}_{ij}\pi^{2}_{i}\pi^{2}_{j}}{\sum^{n}_{i=1}\sum_{j\neq i}a^{2}_{ij}\sigma^{2}_{i}\sigma^{2}_{j}}.
\]
In the first term $\pi^{2}_{i}$ and $\sigma^{2}_{i}$ carry the same weight $\sum_{j\neq i}a^{2}_{ij}\sigma^{2}_{j}$: two systematic components with the same $\sum_{i}\pi^{2}_{i}$ produce different gains, with a larger gain when the large values of $\pi^{2}_{i}$ belong to observations whose weight is large. $\blacksquare$
\end{rem}

Along local sequences, those with $\pi_i\to0$ sufficiently fast that $B_n/V\to0$, the two statistics agree to first order and the two tests have the same limiting power. The comparison sharpens as the alternative moves away from the null. Write
\[
N_{n}=\sum^{n}_{i=1}\sum_{j\neq i}a_{ij}\pi_{i}\pi_{j},\qquad S_{n}=2\sum^{n}_{i=1}\sum_{j\neq i}a^{2}_{ij}\pi^{2}_{i}\pi^{2}_{j},\qquad\tau_{n}=\frac{N_{n}}{S^{1/2}_{n}},
\]
for the systematic part of the numerator of $Q_{n}$, the quartic part of the plug-in bias \eqref{eq:biasplug}, and their ratio. Two properties of $\tau_{n}$ follow from the definitions. Both $N_{n}$ and $S^{1/2}_{n}$ are of second degree in $\pi$, so $\tau_{n}$ does not change when $\pi$ is replaced by $c\pi$ for any $c\neq0$. And if at most $k\ge2$ of the $\pi_{i}$ are nonzero, only the $k(k-1)$ pairs inside their support contribute, so Cauchy--Schwarz gives $|\tau_{n}|\le\{k(k-1)/2\}^{1/2}$. The next result shows that the plug-in statistic converges to $\tau_{n}$.

\begin{cor} \label{cor:sat}
Suppose the conditions of Theorem \ref{thm:alt} hold under $H_{1}$, and in addition that (i) $Q_{n}=N_{n}+o_{p}(S^{1/2}_{n})$ and (ii) the quartic part of the plug-in bias dominates, $V+B_{n}=S_{n}\{1+o(1)\}$. Then
\[
T_{n}=\bigl\{\tau_{n}+o_{p}(1)\bigr\}\bigl\{1+o_{p}(1)\bigr\},\qquad\frac{T^{\mathrm{cf}}_{n}}{T_{n}}\overset{p}{\to}\infty.
\]
Let the test reject when $|T_{n}|$ exceeds a fixed critical value $c>0$. If $\limsup_{n}|\tau_{n}|<c$ then $P(|T_{n}|>c)\to0$; if in addition $\liminf_{n}|\tau_{n}|>0$ then $P(|T^{\mathrm{cf}}_{n}|>c)\to1$.
\end{cor}

Conditions (i) and (ii) hold along sequences of alternatives that move away from the null in a fixed direction. Write $\pi=t\pi^{0}$ for a fixed vector $\pi^{0}$ and a scalar $t\to\infty$. Then $S_{n}$ is of order $t^{4}$, while $V$ does not depend on $t$, the quadratic part of \eqref{eq:biasplug} is of order $t^{2}$, and the conditional variance of $Q_{n}-N_{n}$, which is $V+4\sum^{n}_{j=1}\sigma^{2}_{j}(\sum_{i\neq j}a_{ij}\pi_{i})^{2}$ up to the estimation error, is of order $t^{2}$ as well.

Section \ref{subsec:Testing-many-linear} varies how widely the systematic component is spread, and the size-corrected gain rises as it concentrates. The same contrast arises for the jackknife Anderson--Rubin statistic in \citet[Section 4.2]{MikushevaSun2022}, where the bias of the analogous plug-in variance is quartic in the departure and the shift it leaves saturates in the same way.

\section{Applications} \label{sec:app}

This section applies the general framework of Section~\ref{sec:gen} to three representative testing problems: specification testing, overidentification testing in instrumental variable models, and testing many linear restrictions in linear regression. Table~\ref{table:apps} collects the correspondence.
\begin{table}[t]
\caption{The three applications in the notation of Section~\ref{sec:gen}. In each case $\ell_{i}=m_{i}$ is the $i$th column of $M$ and $c_{ij}=2a^{2}_{ij}/(M_{ii}M_{jj}+M^{2}_{ij})$. The upper block maps each application into the general framework; the lower block highlights the application-specific components. Wherever $\ell^{\prime}_{i}H=0$, $g_{2i}=0$ and $\mathcal{R}_{n}=0$ in Lemma \ref{lem:exp}.}\label{table:apps}

\centering
\small
\begin{tabular}{lccc}
\hline
 & Specification & Overidentification & Many restrictions \\
 & (Section~\ref{sub:spec}) & (Section~\ref{sub:over}) & (Section~\ref{sub:F}) \\
\hline
$a_{ij}$ & $K^{-1/2}P_{ij}$ & $P_{ij}$ & $Q_{ij}$ \\
$M$ & $I-B(B'B)^{-1}B'$ & $I-Z(Z'Z)^{-1}Z'$ & $I-X(X'X)^{-1}X'$ \\
$\Delta$ & $B\alpha_{*}-X\gamma_{*}$ & $P\bar{u}$ & $M_{W}Z(\beta_{2}-\beta_{20})$ \\
$H_{i}$ & $x_{i}$ & $x_{i}$ & $w_{i}$ \\
\hline
$\omega_{i}$ & $\theta_{0}(x_{i})-\alpha^{\prime}_{*}b(x_{i})$ & $m^{\prime}_{i}\bar{u}$ & $0$ \\
$\lambda_{i}=m^{\prime}_{i}\omega$ & $o_{p}(b_{n}^{-3})$ & $\omega_{i}$ & $0$ \\
$\ell^{\prime}_{i}H$ & $0$ & $m^{\prime}_{i}v$ & $0$ \\
\hline
\end{tabular}
\end{table}

\subsection{Specification testing} \label{sub:spec}

We first consider the series-based specification test of \citet{SunLi2006}, which is a modified version of \citet{HongWhite1995} to avoid a non-zero centering term. Let $\{(y_{i},x_{i})\}^{n}_{i=1}$ be an i.i.d. sample with $x_{i}\in\mathbb{R}^{G}$. We test 
\[
H_{0}:\;\mathbb{E}[y\mid x]=\gamma_{0}'x\qquad\text{a.s.}
\]
against general alternatives. Let $b(x)$ be a $K$-dimensional vector of basis functions, and define 
\[
B=(b(x_{1}),\ldots,b(x_{n}))',\qquad P=B(B'B)^{-1}B',\qquad M=I-P,\qquad\hat{u}_{i}=y_{i}-\hat{\gamma}'x_{i},
\]
where $\hat{\gamma}$ is the OLS estimator. The test statistic of \citet{SunLi2006} is 
\[
T=\frac{K^{-1/2}\sum^{n}_{i=1}\sum_{j\neq i}P_{ij}\hat{u}_{i}\hat{u}_{j}}{\sqrt{\hat{V}}},\qquad\hat{V}=\frac{2}{K}\sum^{n}_{i=1}\sum_{j\neq i}P^{2}_{ij}\hat{u}^{2}_{i}\hat{u}^{2}_{j}.
\]

Replacing the plug-in variance estimator by its cross-fit analogue yields 
\[
\hat{V}_{\mathrm{cf}}=\frac{2}{K}\sum^{n}_{i=1}\sum_{j\neq i}\frac{P^{2}_{ij}}{M_{ii}M_{jj}+M^{2}_{ij}}\{\hat{u}_{i}(m_{i}'\hat{u})\}\{\hat{u}_{j}(m_{j}'\hat{u})\},
\]
where $m_{i}$ denotes the $i$th column of $M$. The specification-testing problem fits the framework of Section \ref{sec:gen} with 
\[
a_{ij}=K^{-1/2}P_{ij},\qquad c_{ij}=\frac{2}{K}\frac{P^{2}_{ij}}{M_{ii}M_{jj}+M^{2}_{ij}},\qquad\ell_{i}=m_{i}.
\]

Decompose the residual to locate the drift. Let $\theta_{0}(x)=\mathbb{E}[y\mid x]$, let $\alpha_{*}=\mathbb{E}[b(x)b(x)^{\prime}]^{-1}\mathbb{E}[b(x)y]$ be the coefficient of the best sieve approximation to it, and let $\gamma_{*}$ be the probability limit of $\hat{\gamma}$. Then
\[
\hat{u}_{i}=\underbrace{\alpha^{\prime}_{*}b(x_{i})-\gamma^{\prime}_{*}x_{i}}_{\Delta_{i}}+\underbrace{y_{i}-\theta_{0}(x_{i})}_{e_{i}}+\underbrace{\theta_{0}(x_{i})-\alpha^{\prime}_{*}b(x_{i})}_{\omega_{i}}+x^{\prime}_{i}\hat{\delta},\qquad\hat{\delta}=\gamma_{*}-\hat{\gamma},
\]
which telescopes to $\hat{u}_{i}=y_{i}-\hat{\gamma}^{\prime}x_{i}$ and is \eqref{eq:split} with $H_{i}=x_{i}$.

\begin{prop} \label{prop:spec}
Suppose Assumption \ref{as:spec} holds. Then the specification-testing problem satisfies the conditions of Theorems \ref{thm:null} and \ref{thm:alt}. Hence the cross-fit variance estimator is consistent for the asymptotic variance under the null hypothesis, and, under Assumption \ref{as:spec}(iv), under the alternative as well, while the plug-in estimator is not.
\end{prop}

\subsubsection{Decomposition of $m^{\prime}_{i}\hat{u}$}

Applying $m^{\prime}_{i}$ to $\hat{u}_{i}$ gives the second factor of $\hat{u}_{i}(m^{\prime}_{i}\hat{u})$,
\[
m^{\prime}_{i}\hat{u}=m^{\prime}_{i}\Delta+m^{\prime}_{i}\omega+m^{\prime}_{i}e+m^{\prime}_{i}H\hat{\delta},
\]
where $m^{\prime}_{i}e$ is mean zero. We take the other three terms in turn.

The first term is the drift, and it is annihilated exactly: $m^{\prime}_{i}\Delta=0$ for every $i$. Because the basis spans the linear terms, both $B\alpha_{*}$ and $X\gamma_{*}$ lie in $\operatorname{col}(B)$, so $\Delta\in\operatorname{col}(B)$ and $M\Delta=(I-P)\Delta=0$.

The second term is $m^{\prime}_{i}\omega$, which is what the auxiliary factor leaves of the sieve approximation error. The systematic component is $\pi_{i}=\Delta_{i}+\omega_{i}=\theta_{0}(x_{i})-\gamma^{\prime}_{*}x_{i}$, the nonlinearity that the linear specification cannot represent; the sieve coefficient cancels, so $\pi$ does not depend on the basis. Applying the auxiliary factor to it leaves $m^{\prime}_{i}\pi=\lambda_{i}$, since $m^{\prime}_{i}\Delta=0$, and Assumption \ref{as:spec}(iii) makes $\max_{i}|\omega_{i}|$ and $\max_{i}|\lambda_{i}|$, with $\lambda_{i}=m^{\prime}_{i}\omega$, of order $o_{p}(b^{-3}_{n})$. The plug-in therefore squares $\pi_{i}$, which is of order one under the alternative, whereas the second factor retains only $\lambda_{i}$ of it.

The last term vanishes as well: the auxiliary vectors annihilate the estimation directions. Because the basis spans the linear terms, $X=B\alpha_{K}$ for some $\alpha_{K}$, and hence $\ell^{\prime}_{i}H=m^{\prime}_{i}X=0$. By Lemma \ref{lem:exp} this makes $g_{2i}=0$ and $\mathcal{R}_{n}=0$, so only $\mathcal{G}_{1}$ and $\mathcal{G}_{2}$ survive in \eqref{eq:condest}.

\subsection{Overidentification testing} \label{sub:over}

We next consider the overidentification test for linear instrumental variable models studied by \citet{ChaoHausmanNeweySwansonWoutersen2014}. Let
\[
y_{i}=\beta^{\prime}x_{i}+u_{i},\qquad x_{i}=\Upsilon_{i}+v_{i},
\]
for $i=1,\ldots,n$, where $x_{i}\in\mathbb{R}^{G}$ is endogenous and $z_{i}\in\mathbb{R}^{K}$ is a vector of instruments with $K>G$. Following \citet{ChaoHausmanNeweySwansonWoutersen2014}, the instruments $Z=(z_{1},\ldots,z_{n})^{\prime}$ and the reduced form $\Upsilon=(\Upsilon_{1},\ldots,\Upsilon_{n})^{\prime}$ are treated as nonrandom, with $\mathbb{E}[u_{i}]=0$ and $\mathbb{E}[v_{i}]=0$, and $\Upsilon$ is spanned by the instruments, $\Upsilon=Z\Psi$.\footnote{\label{fn:series}\citet{ChaoHausmanNeweySwansonWoutersen2014} also allow the reduced form to be approximated rather than spanned, $\Upsilon_{i}=f_{0}(w_{i})$ with $n^{-1}\sum_{i}\|\Upsilon_{i}-\Psi^{\prime}z_{i}\|^{2}\to0$ for $z_{i}=(p_{1K}(w_{i}),\ldots,p_{KK}(w_{i}))^{\prime}$. We impose the exact case in Assumption \ref{as:over}(ii), because the mean-square condition does not bound $\max_{i}\|m^{\prime}_{i}\Upsilon\|$, which the remainder term requires here.} Let
\[
P=Z(Z^{\prime}Z)^{-1}Z^{\prime},\qquad M=I-P,\qquad\hat{u}_{i}=y_{i}-\hat{\beta}^{\prime}x_{i},
\]
where $\hat{\beta}$ has probability limit $\beta_{*}$, equal to $\beta$ under the null. The choice of $\hat{\beta}$ can be the JIV estimators analyzed by \citet{ChaoSwansonHausmanNeweyWoutersen2012} or the HLIM and HFUL estimators of \citet{HausmanNeweyWoutersenChaoSwanson2012}, for example. \citet{ChaoHausmanNeweySwansonWoutersen2014} use the last one. Assumption \ref{as:over}(iii) asks no more of it than $\hat{\beta}\overset{p}{\to}\beta_{*}$.

Because $Z$ is nonrandom and $\mathbb{E}[u_{i}]=0$, the moment condition $\mathbb{E}[Z^{\prime}u]=0$ holds by construction, so it is not what the test can check. Write
\[
\bar{u}_{i}=\mathbb{E}[y_{i}-\beta^{\prime}_{*}x_{i}],\qquad\bar{u}=(\bar{u}_{1},\ldots,\bar{u}_{n})^{\prime},
\]
for the expected residual at the pseudo-true value. The null of correct specification is $H_{0}:\bar{u}=0$, tested against $H_{1}:\bar{u}\neq0$, which arises when the instruments affect the outcome directly or the reduced form is not explained by the structural equation. With $K>G$ the instruments impose $K-G$ overidentifying restrictions. The test statistic of \citet{ChaoHausmanNeweySwansonWoutersen2014} is
\[
J=\frac{\sum^{n}_{i=1}\sum_{j\neq i}P_{ij}\hat{u}_{i}\hat{u}_{j}}{\sqrt{\hat{\Phi}}}+K,\qquad\hat{\Phi}=\frac{1}{K}\sum^{n}_{i=1}\sum_{j\neq i}P^{2}_{ij}\hat{u}^{2}_{i}\hat{u}^{2}_{j}.
\]
The null hypothesis is rejected when $J$ exceeds the $(1-\alpha)$ quantile of the $\chi^{2}_{K-G}$ distribution.

Replacing the plug-in variance estimator by its cross-fit analogue yields 
\[
\hat{\Phi}_{\mathrm{cf}}=\frac{1}{K}\sum^{n}_{i=1}\sum_{j\neq i}\frac{P^{2}_{ij}}{M_{ii}M_{jj}+M^{2}_{ij}}\{\hat{u}_{i}(m_{i}'\hat{u})\}\{\hat{u}_{j}(m_{j}'\hat{u})\},
\]
where $m_{i}$ denotes the $i$th column of $M$. The overidentification-testing problem fits the framework of Section~\ref{sec:gen} with 
\[
a_{ij}=P_{ij},\qquad c_{ij}=\frac{2P^{2}_{ij}}{M_{ii}M_{jj}+M^{2}_{ij}},\qquad\ell_{i}=m_{i}.
\]
With these weights the variance estimators of Section~\ref{sec:gen} are exactly $2K$ times those in the statistic, $\hat{V}=2K\hat{\Phi}$ and $\hat{V}_{\mathrm{cf}}=2K\hat{\Phi}_{\mathrm{cf}}$. Because the conclusions of Section~\ref{sec:gen} are statements about ratios, they apply to $\hat{\Phi}$ and $\hat{\Phi}_{\mathrm{cf}}$ unchanged.

Decompose the residual to locate the drift. Since $\hat{u}_{i}=(y_{i}-\beta^{\prime}_{*}x_{i})+x^{\prime}_{i}(\beta_{*}-\hat{\beta})$ and $y_{i}-\beta^{\prime}_{*}x_{i}=\bar{u}_{i}+\{(y_{i}-\beta^{\prime}_{*}x_{i})-\bar{u}_{i}\}$, splitting $\bar{u}=P\bar{u}+M\bar{u}$ gives
\[
\hat{u}_{i}=\Delta_{i}+\underbrace{(y_{i}-\beta^{\prime}_{*}x_{i})-\bar{u}_{i}}_{e_{i}}+\underbrace{m^{\prime}_{i}\bar{u}+x^{\prime}_{i}(\beta_{*}-\hat{\beta})}_{\rho_{i}},\qquad\Delta=P\bar{u}.
\]
This is \eqref{eq:uhat} and \eqref{eq:split} with $\omega_{i}=m^{\prime}_{i}\bar{u}$, $H_{i}=x_{i}$ and $\hat{\delta}=\beta_{*}-\hat{\beta}$. The error $e_{i}$ has mean zero by construction and equals $u_{i}+v^{\prime}_{i}(\beta-\beta_{*})$ whenever the structural equation $y_{i}=\beta^{\prime}x_{i}+u_{i}$ holds. Under $H_{0}$, we have $e_{i}=u_{i}$.

\begin{prop} \label{prop:over}
Suppose Assumption \ref{as:over} holds. Then the overidentification-testing problem satisfies the conditions of Theorems \ref{thm:null} and \ref{thm:alt}. Hence the cross-fit variance estimator is consistent for the asymptotic variance under the null hypothesis, and, under Assumption \ref{as:over}(iv), under the alternative as well, while the plug-in estimator is not.
\end{prop}

The decomposition of $m^{\prime}_{i}\hat{u}$ is as in Section \ref{sub:spec}. The drift is annihilated because $\Delta=P\bar{u}$ lies in $\operatorname{col}(Z)$ and $M=I-P$ projects onto its orthogonal complement. The other two terms behave differently.

Here $\omega=M\bar{u}$ is the part of the systematic residual not spanned by the instruments, the overidentification analogue of the sieve error in the specification test. Because $M$ is idempotent, $m^{\prime}_{i}\omega=m^{\prime}_{i}\bar{u}=\omega_{i}$: the auxiliary factor removes nothing from $\omega$. In the notation of Section~\ref{sec:gen} this is $\lambda_{i}=\omega_{i}$, so the two maximal bounds on $\omega$ in Assumption \ref{as:gen} coincide. The systematic component is $\pi_{i}=\Delta_{i}+\omega_{i}=\bar{u}_{i}$, the expected residual itself, and since the drift is annihilated, the auxiliary factor leaves $m^{\prime}_{i}\pi=\omega_{i}$. The plug-in therefore squares $\bar{u}_{i}$, whereas the second factor retains only $\omega_{i}$.

The estimation directions are not annihilated, unlike in the specification test. Since $\Upsilon$ is spanned by the instruments but $v$ is not,
\[
\ell^{\prime}_{i}H=m^{\prime}_{i}x=m^{\prime}_{i}v\neq0,
\]
so $g_{2i}\neq0$ in Lemma \ref{lem:exp} and the remainder $\mathcal{R}_{n}$
has to be bounded. 

\begin{rem}
Since $m^{\prime}_{i}\omega$ does not vanish, Assumption \ref{as:over}(iii) has to bound $\max_{i}|\omega_{i}|$ directly, and we describe which departures from the null satisfy that bound. Write the departure as a term entering the outcome equation directly, $y_{i}=\beta^{\prime}x_{i}+h_{i}+u_{i}$. Then $\bar{u}=\Upsilon(\beta-\beta_{*})+h$, and $M\Upsilon=0$ because $\Upsilon=Z\Psi$, so $\omega=Mh$ whatever the pseudo-true value. Thus, Assumption \ref{as:over}(iii) asks that the part of the departure the instruments do not span be uniformly small, $\max_{i}|(Mh)_{i}|=o_{p}(b^{-3}_{n})$. Two situations justify the condition. First, $h$ may lie in $\operatorname{col}(Z)$, so that the instruments enter the outcome equation linearly and $Mh=0$. Second, $h$ may be a smooth function of the variables from which the instruments are built, as in the series construction of footnote \ref{fn:series}, in which case $Mh$ is an approximation error of the same kind as the sieve error of Section \ref{sub:spec}. The alternative against which \citet{ChaoHausmanNeweySwansonWoutersen2014} report the power of their test is of the first kind, a direct effect of one instrument on the outcome with the coefficient on it moved away from zero, and the design of Section \ref{subsec:sim-overid} follows that design. In contrast, Assumption \ref{as:over}(iii) rules out a case where $h_{i}$ is driven by a variable the instruments do not span; in this case, the systematic component $\pi_{i}=\Delta_{i}+\omega_{i}$ enters the second factor of $\hat{u}_{i}(m^{\prime}_{i}\hat{u})$ as well as the first. $\blacksquare$
\end{rem}

\subsection{Testing many linear restrictions} \label{sub:F}

We finally consider testing many linear restrictions in linear regression models, where variance inflation under the alternative becomes particularly important as the number of restrictions increases. Let $\{(y_{i},x_{i})\}^{n}_{i=1}$ be an i.i.d. sample with $x_{i}=(w_{i}',z_{i}')'$, where $w_{i}\in\mathbb{R}^{d-r}$ is fixed-dimensional, with $d=\dim(x_{i})$ the total number of regressors, and $z_{i}\in\mathbb{R}^{r}$, with the number of restrictions $r$ potentially large. Consider the heteroskedastic linear regression model 
\[
y_{i}=w_{i}'\beta_{1}+z_{i}'\beta_{2}+\varepsilon_{i},\qquad\mathbb{E}[\varepsilon_{i}\mid x_{i}]=0,
\]
and test $H_{0}:\beta_{2}=\beta_{20}$. Let $X=[W\;Z]$ and 
\[
M_{W}=I-W(W'W)^{-1}W',\qquad Q=M_{W}Z(Z'M_{W}Z)^{-1}Z'M_{W},\qquad M=I-X(X'X)^{-1}X'.
\]
Write $Q_{ij}$ and $M_{ij}$ for the $(i,j)$th entries of $Q$ and $M$, respectively, and let $m_{i}$ denote the $i$th column of $M$. Let $\hat{\sigma}^{2}_{\varepsilon}=\frac{1}{n-d}\sum^{n}_{i=1}(y_{i}-x_{i}'\hat{\beta})^{2}$ denote the unrestricted residual variance estimator, where $\hat{\beta}$ is the unrestricted OLS estimator. The conventional $F$ statistic is 
\[
F=\frac{(y-Z\beta_{20})^{\prime}Q(y-Z\beta_{20})}{r\hat{\sigma}^{2}_{\varepsilon}}.
\]
This is the statistic \citet{AnatolyevSolvsten2023} use for $H_{0}:R\beta=q$, taken at $R=[0\;I_{r}]$ and $q=\beta_{20}$; a general full-row-rank $R$ can be brought to $[0\;I_{r}]$ by a nonsingular transformation of the regressors, which leaves $F$ unchanged.

Following \citet{AnatolyevSolvsten2023}, consider critical values of the form 
\[
c_{\alpha}(E,V)=\frac{1}{r\hat{\sigma}^{2}_{\varepsilon}}(E+V^{1/2}\kappa_{\alpha}),\qquad\kappa_{\alpha}=\frac{q_{1-\alpha}(F_{r,n-d})-1}{\sqrt{2/r+2/(n-d)}},
\]
and $q_{1-\alpha}(F_{r,n-d})$ denotes the $(1-\alpha)$ quantile of the $F_{r,n-d}$ distribution. The null hypothesis is rejected whenever $F>c_{\alpha}(E,V)$. \citet{AnatolyevSolvsten2023} estimate $E$ using a leave-one-out estimator and $V$ using a leave-three-out estimator. 

\citet[Remark 4]{AnatolyevSolvsten2023} suggest estimating the error variances under the null, and write out the resulting leave-out estimates in their Appendix A.4. We follow that line in a narrower setting, where imposing the restriction leaves only a fixed-dimensional parameter to estimate, and we use the restricted residual itself in place of their leave-out construction. Consider the plug-in estimators 
\[
\hat{E}=\sum^{n}_{i=1}Q_{ii}\tilde{e}^{2}_{i},\qquad\hat{V}^{\mathrm{plug}}=2\sum^{n}_{i=1}\sum_{j\neq i}Q^{2}_{ij}\tilde{e}^{2}_{i}\tilde{e}^{2}_{j},
\]
where $\tilde{e}=y-W\tilde{\beta}_{1}-Z\beta_{20}$ is the restricted residual and 
\[
\tilde{\beta}_{1}=(W'W)^{-1}W'(y-Z\beta_{20})
\]
is the restricted OLS estimator under $H_{0}$. Keeping the centering estimator unchanged, we replace the plug-in variance estimator by
\[
\hat{V}^{\mathrm{cf}}=2\sum^{n}_{i=1}\sum_{j\neq i}\frac{Q^{2}_{ij}}{M_{ii}M_{jj}+M^{2}_{ij}}(\tilde{e}_{i}m^{\prime}_{i}\tilde{e})(\tilde{e}_{j}m^{\prime}_{j}\tilde{e}).
\]
The switch costs nothing in the statistic itself: $QW=0$, so $(y-Z\beta_{20})^{\prime}Q(y-Z\beta_{20})=\tilde{e}^{\prime}Q\tilde{e}$ and the numerator of $F$ is unchanged.

Subtracting the centering estimator removes the diagonal terms, $\tilde{e}^{\prime}Q\tilde{e}-\hat{E}=\sum^{n}_{i=1}\sum_{j\neq i}Q_{ij}\tilde{e}_{i}\tilde{e}_{j}$, so the centered statistic is exactly \eqref{eq:T} with $\hat{u}_{i}=\tilde{e}_{i}$. The many-restrictions testing problem therefore fits the framework of Section \ref{sec:gen} with 
\[
a_{ij}=Q_{ij},\qquad c_{ij}=2\frac{Q^{2}_{ij}}{M_{ii}M_{jj}+M^{2}_{ij}},\qquad\ell_{i}=m_{i}.
\]

Decompose the residual to locate the drift. By the definition of $\tilde{\beta}_{1}$, the restricted residual is $\tilde{e}=M_{W}(y-Z\beta_{20})$. Substituting $y_{i}=w^{\prime}_{i}\beta_{1}+z^{\prime}_{i}\beta_{2}+\varepsilon_{i}$ gives $y-Z\beta_{20}=W\beta_{1}+\Delta_{0}+\varepsilon$ with $\Delta_{0}=Z(\beta_{2}-\beta_{20})$, and $M_{W}W=0$ leaves the exact identity $\tilde{e}=M_{W}(\Delta_{0}+\varepsilon)$. Splitting it as
\[
\hat{u}_{i}=\tilde{e}_{i}=\underbrace{(M_{W}\Delta_{0})_{i}}_{\Delta_{i}}+\underbrace{\varepsilon_{i}}_{e_{i}}+\underbrace{-w^{\prime}_{i}(W^{\prime}W)^{-1}W^{\prime}\varepsilon}_{\rho_{i}}
\]
gives \eqref{eq:uhat} and \eqref{eq:split} with $\omega_{i}=0$, $H_{i}=w_{i}$ and $\hat{\delta}=-(W^{\prime}W)^{-1}W^{\prime}\varepsilon$. Here $\hat{\delta}=\mathbb{E}[\tilde{\beta}_{1}\mid X]-\tilde{\beta}_{1}$ is the sampling noise in the restricted estimator, free of the drift. The bound $\|\hat{\delta}\|=O_{p}(n^{-1/2})$, which Appendix \ref{app:F} derives from Assumption \ref{as:F}(i), therefore holds under the alternative as well as the null, with the rate coming from the fixed dimension of $w_{i}$. 

\begin{prop} \label{prop:F}
Suppose Assumption \ref{as:F} holds. Then the many-restrictions testing problem satisfies the conditions of Theorems \ref{thm:null} and \ref{thm:alt}. Hence the cross-fit variance estimator is consistent for the asymptotic variance under the null hypothesis, and, under Assumption \ref{as:F}(iv), under the alternative as well, while the plug-in estimator is not.
\end{prop}

The decomposition of $m^{\prime}_{i}\hat{u}$ is as in Section \ref{sub:spec}, with $\omega\equiv0$ here, so that $\pi_{i}=\Delta_{i}$ and $\lambda_{i}=0$. The drift is annihilated because $MW=0$ gives $MP_{W}=0$, and $\Delta_{0}\in\operatorname{col}(Z)\subseteq\operatorname{col}(X)$ gives $M\Delta_{0}=0$; hence $M\Delta=MM_{W}\Delta_{0}=0$. The auxiliary vectors annihilate the estimation directions as well, $\ell^{\prime}_{i}H=m^{\prime}_{i}W=0$, so $g_{2i}=0$ and $\mathcal{R}_{n}=0$ in Lemma \ref{lem:exp}, as in the specification test.

\begin{rem} \label{rem:appbias}
We close Section \ref{sec:app} by comparing the two biases across the three applications. Lemma \ref{lem:bias} gives them in each case on substituting $a_{ij}$ and $\ell_{i}=m_{i}$ from Table \ref{table:apps}. The plug-in bias \eqref{eq:biasplug} behaves the same way in all three. It is positive, and under condition (iv) of the corresponding application-specific assumption it does not vanish relative to the variance: $\liminf_{n}B_{n}/V>0$. The cross-fit bias differs across the three, according to whether $\lambda$ vanishes.

In the specification test $\lambda$ is negligible, and in the many-restrictions test $\lambda\equiv0$, so in both the cross-fit bias reduces to the three terms of Remark \ref{rem:lambda0}. Each carries an off-diagonal entry of $M$ or $M\Sigma M$. In the specification test that factor is bounded by the maximal leverage $\max_{i}P_{ii}$, which Assumption \ref{as:spec}(ii) makes $o_{p}(1)$; in the many-restrictions test it is small once the off-diagonal entries of $Q$ are not concentrated on a few pairs, which Assumption \ref{as:F}(iv) requires. Because the number of contaminating terms grows with the number of restrictions while the drift per observation stays bounded, the reduction matters most there.

In the overidentification test $\lambda_{i}=\omega_{i}$ does not vanish and the reduction fails. The terms of \eqref{eq:biascf} carrying $\pi_{i}\lambda_{i}=\bar{u}_{i}\omega_{i}$ survive, and the leading three,
\[
\sum^{n}_{i=1}\sum_{j\neq i}c_{ij}\bigl\{\bar{u}_{i}\omega_{i}\bar{u}_{j}\omega_{j}+\bar{u}_{i}\omega_{i}M_{jj}\sigma^{2}_{j}+M_{ii}\sigma^{2}_{i}\bar{u}_{j}\omega_{j}\bigr\},
\]
carry no off-diagonal factor: at most a diagonal entry of $M$, which is bounded away from zero. Accordingly $\max_{i}|\omega_{i}|$ has to be bounded directly there, and the gain in that application depends on how much of the misspecification the instruments span. $\blacksquare$
\end{rem}

\section{Simulation results} \label{sec:sim}

This section reports the finite-sample behavior of the cross-fit variance estimator in the three problems of Section \ref{sec:app}. In every design the test statistic is held fixed and the procedures differ only in the variance estimator, or, in Section \ref{subsec:Testing-many-linear}, in the critical value built from it, so the differences reported below are attributable to variance estimation alone. Rejection frequencies are at the nominal $5\%$ level, and the null rejection rates are reported with each design. They are not size-corrected except in two places, Table \ref{table:directions} and Appendix \ref{subsec:Size-corrected-power}, where each procedure is instead referred to the $95$th percentile of its own simulated null distribution.

\subsection{Specification testing} \label{subsec:Specification-testing-sim}

We first consider the specification-testing problem of Section \ref{sub:spec}. Following \citet{TripathiKitamura2003},
\[
y_{i}=\beta_{0}+\beta_{1}x_{i}+\frac{c}{\sqrt{2\pi}}\exp\!\left(-(x_{i}-2.5)^{2}\right)+u_{i},\qquad u_{i}=\varepsilon_{i}\sqrt{0.1+0.2x_{i}+0.3x^{2}_{i}},
\]
with $\beta_{0}=\beta_{1}=1$, $n=400$, $\log x_{i}\sim N(0,1)$ truncated at the upper $5\%$ tail, and $\varepsilon_{i}\sim N(0,1)$ independent of $x_{i}$. The parameter $c$ indexes the departure from linearity: $c=0$ is the null, under which the linear specification is correct, and $c>0$ the alternative, under which it is not.

The sieve is a cubic spline,
\[
b(x)=\bigl(1,\;x,\;x^{2},\;x^{3},\;(x-\tau_{1})^{3}_{+},\;\ldots,\;(x-\tau_{K-4})^{3}_{+}\bigr)',
\]
with the knots $\tau_{k}$ at the $k/(K-3)$ sample quantiles of $x$. The constant and the linear term belong to it, so it spans the linear specification, as Section \ref{sub:spec} assumes. At a given noncentrality the inflation of the plug-in variance is of order $\sqrt{K}/n$ relative to the variance. We therefore vary $K$, reporting $K\in\{8,20,40\}$, which spans $K/n$ from $0.02$ to $0.10$.

Table \ref{table:spec_test_power} reports the frequencies. Under the null the conventional test is mildly conservative at every $K$, while the cross-fit test is close to nominal at $K=8$ and mildly liberal at $K=40$, reaching $6.8\%$; the two effects act in the same direction, so part of the difference reported next is a difference in null rejection rates rather than in power. Under the alternative the cross-fit procedure dominates at every $c$ and every $K$, and the margin widens with $K/n$: the largest gap is $2.4$ percentage points at $K=8$, $3.5$ at $K=20$ and $4.6$ at $K=40$. Within each $K$ it is largest at intermediate $c$, where the contamination is big enough to cost power while both procedures still have room to gain.

\begin{table}[t]
\caption{Empirical rejection frequencies for the specification test}\label{table:spec_test_power}

\centering
\begin{tabular}{ccccccc}
\hline
& \multicolumn{2}{c}{$K=8$} & \multicolumn{2}{c}{$K=20$} & \multicolumn{2}{c}{$K=40$} \\
\cline{2-3} \cline{4-5} \cline{6-7}
$c$ & $T$ & $T_{\mathrm{cf}}$ & $T$ & $T_{\mathrm{cf}}$ & $T$ & $T_{\mathrm{cf}}$ \\
\hline
0.0 & 3.5 & 5.2 & 3.6 & 6.1 & 4.5 & 6.8 \\
1.5 & 14.2 & 16.6 & 12.2 & 15.2 & 9.9 & 13.7 \\
2.0 & 25.0 & 26.6 & 20.0 & 23.5 & 16.4 & 20.7 \\
2.5 & 39.0 & 41.2 & 32.2 & 35.7 & 25.7 & 30.3 \\
3.0 & 55.5 & 57.1 & 46.8 & 49.9 & 37.8 & 42.3 \\
3.5 & 71.1 & 72.0 & 62.0 & 64.8 & 51.7 & 55.4 \\
4.0 & 82.9 & 83.1 & 76.7 & 77.8 & 65.9 & 69.0 \\
\hline
\end{tabular}

\smallskip

{\footnotesize Notes: Rejection frequencies in percent at the nominal $5\%$ level, $n=400$, 5,000 Monte Carlo replications. $c=0$ is the null; $c>0$ the alternative. $T$ is the test of \citet{SunLi2006}; $T_{\mathrm{cf}}$ replaces $\hat{V}$ by $\hat{V}_{\mathrm{cf}}$. The cross-fit variance estimate is negative in at most $0.8\%$ of replications; those replications are counted as non-rejections.}{\footnotesize\par}
\end{table}

\subsection{Overidentification testing} \label{subsec:sim-overid}

We next consider the overidentification test of Section \ref{sub:over}. The model is
\[
y_{i}=\beta^{\prime}x_{i}+h_{i}+\varepsilon_{i},\qquad x_{i}=(1,\Upsilon_{i}+v_{i})^{\prime},\qquad\Upsilon_{i}=z^{\prime}_{i}\psi,\qquad\varepsilon_{i}=\rho v_{i}+\sqrt{1-\rho^{2}}\,s_{i}e_{i},
\]
with $\beta=0$, $G=2$, $\rho=0.6$, $s^{2}_{i}=(1+z^{2}_{i1})/2$, $v_{i},e_{i}\sim N(0,1)$ and $n=400$, so that the conditional variance is driven by the first instrument. The instruments are $z_{i}=(1,z_{i1},\ldots,z_{i,K-1})^{\prime}$, where $z_{i2}$ is a dummy equal to one with probability $0.05$ and the rest are standard normal; $\psi$ is spread evenly over the non-constant instruments with the concentration parameter held at $n\psi^{\prime}\psi=400$, so identification does not weaken as $K$ grows. We use the HFUL estimator of \citet{HausmanNeweyWoutersenChaoSwanson2012} as \citet{ChaoHausmanNeweySwansonWoutersen2014} do.

The misspecification is $h_{i}=\eta z_{i3}$, a direct effect of one instrument on the outcome. It lies in $\operatorname{col}(Z)$, so $\omega=M\bar{u}=0$ and the whole of $\bar{u}$ enters the statistic as drift, $\Delta=P\bar{u}=\bar{u}$.

Table \ref{table:overid_power} reports the frequencies. Both statistics are close to nominal under the null, with the cross-fit version at most one percentage point above the conventional one. The cross-fit statistic has rejection frequencies at least as large at every $\eta$, and the margin grows with $K/n$: $0.5$ percentage points at $0.025$, $1.9$ at $0.10$ and $3.9$ at $0.30$.

\begin{table}[t]
\caption{Empirical rejection frequencies for the overidentification test}\label{table:overid_power}

\centering
\begin{tabular}{ccccccc}
\hline
& \multicolumn{2}{c}{$K=10$} & \multicolumn{2}{c}{$K=40$} & \multicolumn{2}{c}{$K=120$} \\
& \multicolumn{2}{c}{$K/n=0.025$} & \multicolumn{2}{c}{$K/n=0.10$} & \multicolumn{2}{c}{$K/n=0.30$} \\
\cline{2-3} \cline{4-5} \cline{6-7}
$\eta$ & $J$ & $J_{\mathrm{cf}}$ & $J$ & $J_{\mathrm{cf}}$ & $J$ & $J_{\mathrm{cf}}$ \\
\hline
0.00 & 5.1 & 5.4 & 5.6 & 6.0 & 4.9 & 5.9 \\
0.05 & 8.9 & 9.1 & 7.0 & 7.6 & 5.5 & 6.5 \\
0.10 & 21.3 & 21.7 & 11.4 & 12.8 & 7.2 & 8.7 \\
0.15 & 46.2 & 46.7 & 23.8 & 25.4 & 11.4 & 13.6 \\
0.20 & 75.2 & 75.3 & 44.2 & 46.1 & 19.9 & 23.4 \\
0.30 & 99.4 & 99.4 & 90.0 & 90.8 & 51.7 & 55.6 \\
\hline
\end{tabular}

\smallskip

{\footnotesize Notes: Rejection frequencies in percent at the nominal $5\%$ level, $n=400$, 2,000 Monte Carlo replications, first step estimated by HFUL. $\eta=0$ is the null. $J$ is the statistic of \citet{ChaoHausmanNeweySwansonWoutersen2014}; $J_{\mathrm{cf}}$ replaces $\hat{\Phi}$ by $\hat{\Phi}_{\mathrm{cf}}$. The cross-fit variance estimate is positive in every replication.}{\footnotesize\par}
\end{table}

\subsection{Testing many linear restrictions} \label{subsec:Testing-many-linear}

We finally consider the many-restrictions testing problem of Section \ref{sub:F}. We draw $n=300$ independent observations from
\[
y_{i}=w^{\prime}_{i}\beta_{1}+z^{\prime}_{i}\beta_{2}+\varepsilon_{i},\qquad\mathbb{E}[\varepsilon_{i}\mid x_{i}]=0,\qquad\mathbb{E}[\varepsilon^{2}_{i}\mid x_{i}]=(1+5|z_{i,1}|)^{2},
\]
without an intercept, where $x_{i}=(w^{\prime}_{i},z^{\prime}_{i})^{\prime}\in\mathbb{R}^{1+r}$ is standard normal, $w_{i}$ is scalar and $\beta_{1}=1$. The conditional variance is driven by the first restricted regressor.

We test $H_{0}:\beta_{2}=0$ against $\beta_{2}=(\eta/\sqrt{r})\mathbf{1}_{r}$, where $\eta\geq0$ indexes the signal strength and $\eta=0$ is the null. The alternative is composite, any $\beta_{2}\neq0$, and the departure enters the statistic through the quadratic form $\sum^{n}_{i=1}\sum_{j\neq i}Q_{ij}\Delta_{i}\Delta_{j}$, which is unchanged when $\beta_{2}$ is replaced by $-\beta_{2}$: the rejection frequency depends on the magnitude of $\beta_{2}$ and on its direction in $\mathbb{R}^{r}$, not on its sign. Taking $\beta_{2}$ proportional to $\mathbf{1}_{r}$ fixes that direction at the one spreading the departure evenly over the $r$ coefficients, and $\eta$ varies only the magnitude. The normalization by $\sqrt{r}$ holds the magnitude of the coefficient vector fixed as $r$ varies, so the three values $r\in\{60,120,180\}$ are comparable and span $r/n$ from $0.20$ to $0.60$. With standard normal regressors this design has $\max_{i}|\Delta_{i}|$ growing like $\sqrt{\log n}$, so it meets the boundedness condition of Assumption \ref{as:F}(iv) only up to a logarithmic factor; the calibrated design of Section \ref{sec:calibrated}, where the covariate is truncated, meets it exactly. Four procedures share the same $F$ statistic and differ only in the critical value: the conventional $F$ value, the procedure of \citet{AnatolyevSolvsten2023}, the plug-in procedure based on restricted residuals, and the proposed cross-fit procedure.

Figure \ref{fig:matsushita} reports the frequencies. Under the null the cross-fit and Anatolyev--S{ø}lvsten procedures are close to the nominal level at all three $r$, the plug-in procedure is conservative and becomes more so as $r$ grows, from $3.8\%$ at $r=60$ to $2.8\%$ at $r=180$, and the conventional $F$ test is mildly oversized. Under the alternative the cross-fit procedure rejects more often than the plug-in one at every $\eta$, and the margin widens sharply with $r/n$: the largest gap is $4.1$ percentage points at $r=60$, $7.9$ at $r=120$ and $13.6$ at $r=180$. The cross-fit procedure also tracks the Anatolyev--S{ø}lvsten procedure closely throughout, exceeding it by at most $2.5$ percentage points, at $r=180$.

Remark \ref{rem:gain} and Corollary \ref{cor:sat} say where the gain should appear. Holding $\|\beta_{2}\|=\eta$ fixed, we place the departure in three ways: spread over the $r$ coefficients, as above; on the first restricted coefficient, whose regressor drives the error variance; and on a coefficient whose regressor is a standardized Bernoulli$(p)$ dummy, which concentrates the drift on a $p$ fraction of the observations and leaves it independent of the error variance. Table \ref{table:directions} reports the largest gap over the $\eta$ grid.

The variance driver makes almost no difference. At $r=60$ the corrected gain is 0.2 percentage points when the departure is spread over the $r$ coefficients and 0.3 when it falls on the driver. As mentioned in Remark \ref{rem:gain}, the gain depends on where the departure falls relative to the weights, and moving it onto the driver changes only its position relative to the error variances. Concentration makes a large difference. The corrected gain is 1.6 points when the dummy fires on $5\%$ of the observations and 27.0 when it fires on $1\%$, as the kurtosis of the drift rises from 3 to 19 to 118. Using the notation of Corollary \ref{cor:sat}, the saturation level $\tau_{n}$ falls as the departure concentrates, so the plug-in statistic stops rising sooner. In the first three rows the raw gaps are much larger than the corrected ones, because the two procedures reject at different rates under the null; in the last row most of the gap survives the correction.

\begin{table}[t]
\caption{Largest rejection-frequency gap between the cross-fit and plug-in procedures, by placement of the departure}\label{table:directions}

\centering
\small
\begin{tabular}{lccccc}
\hline
 &  & \multicolumn{2}{c}{$r=60$} & \multicolumn{2}{c}{$r=180$}\\
\cline{3-4} \cline{5-6}
Departure & kurtosis & raw & corrected & raw & corrected\\
\hline
spread over the $r$ coefficients & 3.0 & 4.6 & 0.2 & 14.1 & 0.7\\
on the variance driver & 3.0 & 3.7 & 0.3 & 12.4 & 0.8\\
on a rare dummy, $p=0.05$ & 19.1 & 5.6 & 1.6 & 13.8 & 1.1\\
on a rare dummy, $p=0.01$ & 118.2 & 32.3 & 27.0 & 26.0 & 14.4\\
\hline
\end{tabular}
\par\smallskip

Notes: Largest gap over the $\eta$ grid, in percentage points, at $n=300$ with 2,000 replications and $\|\beta_{2}\|$ held fixed across rows. Corrected refers each procedure to the 95th percentile of its own simulated null distribution. Kurtosis is that of the drift $\Delta_{i}$ across observations.
\end{table}

\begin{figure}[t]
\centering \begin{tikzpicture}
\begin{groupplot}[
    group style={
        group size=3 by 2,
        horizontal sep=0.7cm,
        vertical sep=0.6cm,
        xlabels at=edge bottom,
        xticklabels at=edge bottom,
    },
    width=0.29\textwidth,
    xmin=-0.05,
    xmax=3.05,
    grid=major,
    grid style={gray!20},
    legend style={
        font=\scriptsize,
        draw=gray!50,
        fill=white,
        fill opacity=0.95,
        text opacity=1
    },
    every axis plot/.append style={very thick},
    tick label style={font=\scriptsize},
    label style={font=\small},
    title style={font=\small},
]


\nextgroupplot[
    height=5.0cm,
    title={$r=60$ ($d=61$)},
    ylabel={rejection frequency},
    ylabel style={font=\scriptsize},
    ymin=0,
    ymax=1.02,
    ytick={0,0.2,0.4,0.6,0.8,1.0},
    legend pos=north west,
]
\addplot[blue!70!black, mark=triangle*, mark size=1.3pt]
coordinates {
(0.000,0.050) (0.200,0.052) (0.400,0.064) (0.600,0.084)
(0.800,0.114) (1.000,0.164) (1.200,0.239) (1.400,0.347)
(1.600,0.471) (1.800,0.606) (2.000,0.727) (2.200,0.834)
(2.400,0.910) (2.600,0.956) (2.800,0.986) (3.000,0.995)
};
\addlegendentry{AS23}

\addplot[orange!90!black, mark=square*, mark size=1.1pt]
coordinates {
(0.000,0.037) (0.200,0.043) (0.400,0.053) (0.600,0.071)
(0.800,0.100) (1.000,0.147) (1.200,0.215) (1.400,0.320)
(1.600,0.444) (1.800,0.570) (2.000,0.699) (2.200,0.815)
(2.400,0.896) (2.600,0.948) (2.800,0.982) (3.000,0.994)
};
\addlegendentry{plug-in}

\addplot[red!85!black, mark=*, mark size=1.1pt]
coordinates {
(0.000,0.050) (0.200,0.054) (0.400,0.066) (0.600,0.086)
(0.800,0.118) (1.000,0.169) (1.200,0.251) (1.400,0.358)
(1.600,0.484) (1.800,0.611) (2.000,0.737) (2.200,0.841)
(2.400,0.914) (2.600,0.959) (2.800,0.988) (3.000,0.995)
};
\addlegendentry{cross-fit}

\draw[dotted, gray!70]
(axis cs:-0.05,0.05) -- (axis cs:3.05,0.05);

\nextgroupplot[
    height=5.0cm,
    title={$r=120$ ($d=121$)},
    ymin=0,
    ymax=1.02,
    ytick={0,0.2,0.4,0.6,0.8,1.0},
]
\addplot[blue!70!black, mark=triangle*, mark size=1.3pt]
coordinates {
(0.000,0.050) (0.200,0.054) (0.400,0.060) (0.600,0.069)
(0.800,0.090) (1.000,0.117) (1.200,0.151) (1.400,0.209)
(1.600,0.272) (1.800,0.371) (2.000,0.474) (2.200,0.584)
(2.400,0.681) (2.600,0.782) (2.800,0.857) (3.000,0.915)
};

\addplot[orange!90!black, mark=square*, mark size=1.1pt]
coordinates {
(0.000,0.035) (0.200,0.036) (0.400,0.042) (0.600,0.050)
(0.800,0.062) (1.000,0.089) (1.200,0.113) (1.400,0.164)
(1.600,0.228) (1.800,0.312) (2.000,0.412) (2.200,0.521)
(2.400,0.635) (2.600,0.733) (2.800,0.822) (3.000,0.892)
};

\addplot[red!85!black, mark=*, mark size=1.1pt]
coordinates {
(0.000,0.056) (0.200,0.059) (0.400,0.066) (0.600,0.074)
(0.800,0.096) (1.000,0.119) (1.200,0.163) (1.400,0.220)
(1.600,0.288) (1.800,0.383) (2.000,0.490) (2.200,0.597)
(2.400,0.696) (2.600,0.789) (2.800,0.868) (3.000,0.920)
};

\draw[dotted, gray!70]
(axis cs:-0.05,0.05) -- (axis cs:3.05,0.05);

\nextgroupplot[
    height=5.0cm,
    title={$r=180$ ($d=181$)},
    ymin=0,
    ymax=1.02,
    ytick={0,0.2,0.4,0.6,0.8,1.0},
]
\addplot[blue!70!black, mark=triangle*, mark size=1.3pt]
coordinates {
(0.000,0.048) (0.200,0.050) (0.400,0.053) (0.600,0.059)
(0.800,0.069) (1.000,0.087) (1.200,0.107) (1.400,0.136)
(1.600,0.179) (1.800,0.225) (2.000,0.290) (2.200,0.359)
(2.400,0.435) (2.600,0.522) (2.800,0.608) (3.000,0.692)
};

\addplot[orange!90!black, mark=square*, mark size=1.1pt]
coordinates {
(0.000,0.028) (0.200,0.030) (0.400,0.033) (0.600,0.036)
(0.800,0.042) (1.000,0.052) (1.200,0.066) (1.400,0.085)
(1.600,0.113) (1.800,0.147) (2.000,0.197) (2.200,0.265)
(2.400,0.336) (2.600,0.411) (2.800,0.506) (3.000,0.595)
};

\addplot[red!85!black, mark=*, mark size=1.1pt]
coordinates {
(0.000,0.053) (0.200,0.056) (0.400,0.059) (0.600,0.069)
(0.800,0.081) (1.000,0.097) (1.200,0.120) (1.400,0.148)
(1.600,0.192) (1.800,0.243) (2.000,0.306) (2.200,0.374)
(2.400,0.460) (2.600,0.546) (2.800,0.627) (3.000,0.716)
};

\draw[dotted, gray!70]
(axis cs:-0.05,0.05) -- (axis cs:3.05,0.05);


\nextgroupplot[
    height=4.0cm,
    xlabel={$\eta$ (signal strength)},
    ylabel={rejection-frequency gap (pp)},
    ylabel style={font=\scriptsize},
    ymin=-0.04,
    ymax=0.16,
    ytick={0,0.05,0.10,0.15},
    yticklabels={0,5,10,15},
    legend pos=north west,
]
\addplot[red!85!black, mark=*, mark size=1.0pt]
coordinates {
(0.000,0.013) (0.200,0.012) (0.400,0.012) (0.600,0.015)
(0.800,0.019) (1.000,0.022) (1.200,0.036) (1.400,0.038)
(1.600,0.041) (1.800,0.041) (2.000,0.038) (2.200,0.026)
(2.400,0.018) (2.600,0.011) (2.800,0.006) (3.000,0.001)
};
\addlegendentry{cross-fit $-$ plug-in}

\addplot[blue!70!black, mark=triangle*, mark size=1.2pt, dashed]
coordinates {
(0.000,0.001) (0.200,0.003) (0.400,0.002) (0.600,0.003)
(0.800,0.004) (1.000,0.005) (1.200,0.012) (1.400,0.011)
(1.600,0.013) (1.800,0.005) (2.000,0.010) (2.200,0.007)
(2.400,0.004) (2.600,0.003) (2.800,0.002) (3.000,0.000)
};
\addlegendentry{cross-fit $-$ AS23}

\draw[dotted, gray!70]
(axis cs:-0.05,0) -- (axis cs:3.05,0);

\nextgroupplot[
    height=4.0cm,
    xlabel={$\eta$ (signal strength)},
    ymin=-0.04,
    ymax=0.16,
    ytick={0,0.05,0.10,0.15},
    yticklabels={0,5,10,15},
]
\addplot[red!85!black, mark=*, mark size=1.0pt]
coordinates {
(0.000,0.020) (0.200,0.022) (0.400,0.024) (0.600,0.024)
(0.800,0.034) (1.000,0.030) (1.200,0.050) (1.400,0.056)
(1.600,0.060) (1.800,0.071) (2.000,0.079) (2.200,0.076)
(2.400,0.061) (2.600,0.056) (2.800,0.046) (3.000,0.028)
};

\addplot[blue!70!black, mark=triangle*, mark size=1.2pt, dashed]
coordinates {
(0.000,0.006) (0.200,0.005) (0.400,0.005) (0.600,0.005)
(0.800,0.006) (1.000,0.002) (1.200,0.012) (1.400,0.010)
(1.600,0.016) (1.800,0.012) (2.000,0.016) (2.200,0.013)
(2.400,0.015) (2.600,0.007) (2.800,0.011) (3.000,0.005)
};

\draw[dotted, gray!70]
(axis cs:-0.05,0) -- (axis cs:3.05,0);

\nextgroupplot[
    height=4.0cm,
    xlabel={$\eta$ (signal strength)},
    ymin=-0.04,
    ymax=0.16,
    ytick={0,0.05,0.10,0.15},
    yticklabels={0,5,10,15},
]
\addplot[red!85!black, mark=*, mark size=1.0pt]
coordinates {
(0.000,0.025) (0.200,0.026) (0.400,0.026) (0.600,0.033)
(0.800,0.038) (1.000,0.044) (1.200,0.054) (1.400,0.063)
(1.600,0.079) (1.800,0.096) (2.000,0.109) (2.200,0.109)
(2.400,0.124) (2.600,0.136) (2.800,0.121) (3.000,0.121)
};

\addplot[blue!70!black, mark=triangle*, mark size=1.2pt, dashed]
coordinates {
(0.000,0.005) (0.200,0.006) (0.400,0.006) (0.600,0.010)
(0.800,0.011) (1.000,0.010) (1.200,0.013) (1.400,0.013)
(1.600,0.013) (1.800,0.018) (2.000,0.016) (2.200,0.016)
(2.400,0.025) (2.600,0.024) (2.800,0.019) (3.000,0.024)
};

\draw[dotted, gray!70]
(axis cs:-0.05,0) -- (axis cs:3.05,0);

\end{groupplot}
\end{tikzpicture}

\caption{Empirical rejection frequencies under $\beta_{2}=(\eta/\sqrt{r})\mathbf{1}_{r}$ for $r\in\{60,120,180\}$, based on $n=300$ and 2,000 Monte Carlo replications. The regressors are independent standard normal and $\sigma^{2}_{i}=(1+5|z_{i,1}|)^{2}$. The top row reports the rejection frequencies of AS23, the plug-in procedure, and the proposed cross-fit procedure. The bottom row reports the corresponding rejection-frequency gaps between the cross-fit procedure and the other two procedures. At $\eta=0$, the rejection frequencies of the conventional $F$/AS23/plug-in/cross-fit procedures are, respectively, 6.1/5.0/3.8/5.0\% for $r=60$, 5.9/5.0/3.5/5.5\% for $r=120$, and 5.3/4.8/2.8/5.3\% for $r=180$.}
\label{fig:matsushita} 
\end{figure}
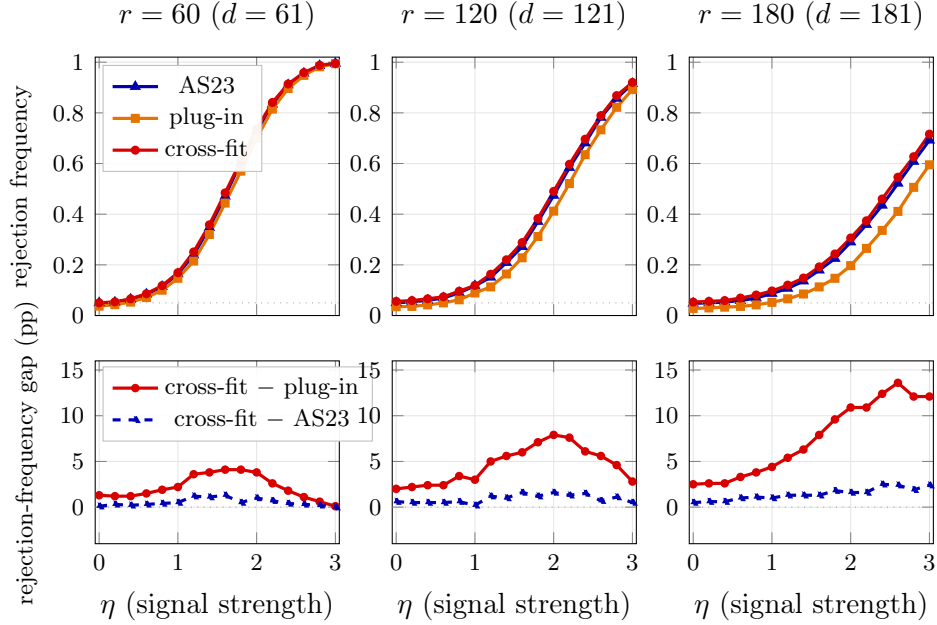

\section{Empirical application: OHIE joint HTE test} \label{sec:emp}

We illustrate the many-restrictions procedures on data from the Oregon Health Insurance Experiment (OHIE), a lottery-based expansion of access to Medicaid in Oregon in 2008. \citet{Finkelstein2012}, hereafter F12, estimate intent-to-treat (ITT) effects of lottery selection on a broad range of health-care utilization and health outcomes, and report treatment-effect estimates for several prespecified subgroups while emphasizing the limited precision of those comparisons. We complement their analysis with a formal joint test of heterogeneity in the ITT effect: we interact lottery selection with a sieve expansion of the baseline covariates and test whether all interaction coefficients are jointly zero. Such a test differs from inspecting a collection of individual subgroup coefficients because it accounts explicitly for the dependence among a potentially large number of restrictions.

Rejection would indicate that the data are inconsistent with a constant conditional ITT effect relative to the specified sieve alternative. Non-rejection means that the data do not detect such heterogeneity; it does not establish exact homogeneity outside the alternative class, nor rule out economically meaningful heterogeneity that is estimated imprecisely.

\subsection{Setup}

For each observation $i$, we estimate the linear probability model
\begin{equation}
y_{i}=w_{i}'\beta_{1}+z_{i}'\beta_{2}+\varepsilon_{i},\qquad H_{0}:\beta_{2}=0,\label{eq:ohie}
\end{equation}
where $y_{i}$ indicates whether individual $i$ had at least one non-childbirth overnight hospital admission during the preceding six months. The vector $w_{i}$ contains a constant, the lottery-selection indicator $D_{i}$, the baseline covariates and corresponding missing-value indicators, and survey-wave-by-household-size fixed effects. The vector $z_{i}=D_{i}b(X_{i})$ contains interactions between lottery selection and a sieve expansion $b(X_{i})$ of the baseline covariates.

We use lottery-list participants who completed both the baseline survey and the 12-month survey and for whom the hospital-admission outcome is available, yielding $n=16{,}455$ observations. Missing continuous baseline covariates are replaced by their sample medians and missing binary covariates by their sample modes. Corresponding missing-value indicators are included among the controls.\footnote{This imputation scheme is used only to retain observations in the joint regression. The resulting test should therefore be interpreted relative to the baseline variables and missing-value controls included in the specified model.}

The sieve contains a rich polynomial expansion of the baseline covariates. It includes main effects, polynomials in age and household income up to degree five, all pairwise interactions among the binary covariates, and interactions between the binary covariates and the polynomial terms in the continuous covariates. The full-sample specification contains $r=136$ restrictions.

Following F12, we also conduct separate analyses for individuals reporting fair or poor baseline health and for those reporting good baseline health. Within each subgroup, sieve terms that become collinear with the controls are removed, leaving $r=112$ restrictions.\footnote{The subgroup samples exclude observations that do not belong to the corresponding self-rated-health category. Their sample sizes therefore need not sum to the full-sample size.}

We compare the conventional $F$ critical value, the plug-in critical value based on restricted residuals, and the proposed cross-fit critical value. The leave-three-out procedure of \citet{AnatolyevSolvsten2023} is not reported because its computational cost is prohibitive for the full OHIE specifications considered here.

\subsection{Results}

Table~\ref{tab:ohie} reports the results. None of the three procedures rejects the null hypothesis at the $5\%$ level, in the full sample or in either self-rated-health subgroup.

The plug-in and cross-fit critical values are nearly identical in this application. This is consistent with a setting in which the realized drift contamination is small at the estimated specification. Because the empirical result is a single realization, however, the similarity of the two critical values does not by itself reveal how the procedures behave under alternatives. 

\begin{table}[t]
\caption{OHIE joint tests of heterogeneity in the ITT effect}
\label{tab:ohie} \centering {\small{}%
\begin{tabular}{lrrrrrrrc}
\hline 
{\small Configuration } & {\small$n$ } & {\small$r$ } & {\small$r/n$ } & {\small$F$ } & {\small$c^{F}_{\alpha}$ } & {\small$c^{\mathrm{plug}}_{\alpha}$ } & {\small$c^{\mathrm{cf}}_{\alpha}$ } & {\small Rejection }\tabularnewline
\hline 
{\small Full sample } & {\small 16,455 } & {\small 136 } & {\small 0.83\% } & {\small 1.009 } & {\small 1.209 } & {\small 1.091 } & {\small 1.090 } & {\small None }\tabularnewline
{\small Fair/poor baseline health } & {\small 6,365 } & {\small 112 } & {\small 1.76\% } & {\small 0.839 } & {\small 1.588 } & {\small 1.032 } & {\small 1.031 } & {\small None }\tabularnewline
{\small Good baseline health } & {\small 9,610 } & {\small 112 } & {\small 1.17\% } & {\small 0.880 } & {\small 1.588 } & {\small 1.105 } & {\small 1.104 } & {\small None }\tabularnewline
\hline 
\end{tabular}}{\small\par}

{\small\medskip{}
}{\small{}%
\begin{minipage}[c]{0.94\textwidth}%
 {\footnotesize Notes: $c^{F}_{\alpha}$, $c^{\mathrm{plug}}_{\alpha}$,
and $c^{\mathrm{cf}}_{\alpha}$ denote the 5\% critical values based
on the conventional $F$ approximation, the restricted-residual plug-in
variance estimator, and the proposed cross-fit variance estimator,
respectively. The procedure of \citet{AnatolyevSolvsten2023} is not
reported because of its computational cost in these specifications. }%
\end{minipage}}{\small\par}
\end{table}

\subsection{Calibrated simulation} \label{sec:calibrated}

We next consider a simulation design motivated by the empirical application. The purpose is not to reproduce the OHIE data exactly, but to study the relative performance of the competing procedures in a design with a binary outcome, treatment interactions, a moderately large sieve, and heterogeneous treatment effects concentrated in the tails of a baseline covariate. We simulate from \eqref{eq:ohie} with $\mathbb{E}[\varepsilon_{i}\mid x_{i}]=0$ and test the same null, $H_{0}:\beta_{2}=0$.

Each observation contains a binary treatment indicator $D_{i}$ and baseline covariates 
\[
X_{i}=(\mathrm{age}_{i},\mathrm{income}_{i},\mathrm{female}_{i},\mathrm{white}_{i},\mathrm{urban}_{i})'.
\]
The continuous covariates $\mathrm{age}_{i}$ and $\mathrm{income}_{i}$ are generated independently from $N(0,0.25)$ truncated to $[-2,2]$. The binary covariates are generated independently from Bernoulli distributions with success probabilities $0.55$, $0.70$, and $0.75$, respectively. We set $w_{i}=(1,D_{i},X_{i}')'$ and $z_{i}=D_{i}b(X_{i})$, where $b(X_{i})$ is a polynomial sieve containing main effects, higher-order polynomial terms, and interactions. The coefficients in $\beta_{1}$ are calibrated to approximately match the magnitudes observed in the empirical application.

Under the alternative, the conditional mean contains the heterogeneity term
\[
\Delta_{0i}=z_{i}'\beta_{2}=\eta D_{i}\phi(X_{i}),
\]
so that the drift of Section \ref{sub:F} is $\Delta=M_{W}\Delta_{0}$, where $\eta\geq0$ controls the magnitude of treatment-effect heterogeneity. We consider $\phi(X_{i})=\mathrm{age}^{3}_{i}$ and $\phi(X_{i})=\mathrm{age}^{5}_{i}$. The null hypothesis corresponds to $\eta=0$. These alternatives concentrate heterogeneous treatment effects in the tails of the age distribution, with the fifth-order alternative generating a more uneven drift pattern.

The baseline simulation uses $n=1{,}500$, with $r=27$ restrictions under the cubic alternative and $r=28$ under the fifth-order one.

We report rejection frequencies at the nominal $5\%$ level. At the baseline sample size the Anatolyev--Sølvsten, plug-in and cross-fit procedures reject between $3.6\%$ and $4.4\%$ under the null, while the conventional $F$ procedure rejects between $8.0\%$ and $8.8\%$. Holding the sieve dimension fixed and reducing $n$ to $1{,}000$, $750$ and $500$ raises $r/n$ from $1.8\%$ to $5.6\%$, and Table~\ref{tab:rawsize} shows that the pattern persists: the $F$ test is oversized throughout, from $8.0\%$ to $12.4\%$, while the other three stay close to nominal, the cross-fit procedure becoming mildly liberal at the smaller sample sizes and reaching $6.2\%$ under the fifth-order basis at $n=500$.

\begin{table}[t]
\caption{Null rejection frequencies as $r/n$ increases}
\label{tab:rawsize} \centering {\small{}%
\begin{tabular}{lrrrrr}
\hline 
{\small Design } & {\small$r/n$ } & {\small Conventional $F$ } & {\small AS23 } & {\small Plug-in } & {\small Cross-fit }\tabularnewline
\hline 
{\small$\mathrm{age}^{3},\ n=1{,}500$ } & {\small 1.8\% } & {\small 0.080 } & {\small 0.042 } & {\small 0.040 } & {\small 0.042 }\tabularnewline
{\small$\mathrm{age}^{3},\ n=1{,}000$ } & {\small 2.7\% } & {\small 0.080 } & {\small 0.050 } & {\small 0.048 } & {\small 0.052 }\tabularnewline
{\small$\mathrm{age}^{3},\ n=750$ } & {\small 3.6\% } & {\small 0.110 } & {\small 0.046 } & {\small 0.044 } & {\small 0.056 }\tabularnewline
{\small$\mathrm{age}^{3},\ n=500$ } & {\small 5.4\% } & {\small 0.110 } & {\small 0.038 } & {\small 0.040 } & {\small 0.050 }\tabularnewline
\hline 
{\small$\mathrm{age}^{5},\ n=1{,}500$ } & {\small 1.8\% } & {\small 0.088 } & {\small 0.036 } & {\small 0.038 } & {\small 0.040 }\tabularnewline
{\small$\mathrm{age}^{5},\ n=1{,}000$ } & {\small 2.8\% } & {\small 0.092 } & {\small 0.032 } & {\small 0.032 } & {\small 0.038 }\tabularnewline
{\small$\mathrm{age}^{5},\ n=750$ } & {\small 3.7\% } & {\small 0.100 } & {\small 0.036 } & {\small 0.038 } & {\small 0.038 }\tabularnewline
{\small$\mathrm{age}^{5},\ n=500$ } & {\small 5.6\% } & {\small 0.124 } & {\small 0.048 } & {\small 0.046 } & {\small 0.062 }\tabularnewline
\hline 
\end{tabular}}{\small\par}

{\small\medskip{}
}{\small{}%
\begin{minipage}[c]{0.92\textwidth}%
{\footnotesize Notes: Entries are rejection frequencies under $\eta=0$ at the nominal 5\% significance level, based on 500 Monte Carlo replications. }%
\end{minipage}}{\small\par}
\end{table}

Figure~\ref{fig:phase1b} reports raw rejection frequencies under the alternatives. The cross-fit procedure rejects more often than the plug-in procedure at every point of the grid. The largest gap is $5.6$ percentage points under the cubic alternative and $7.6$ under the fifth-order one, where it stays between $4$ and $8$ points over the intermediate range $\eta\in[0.075,0.2]$; against the Anatolyev--Sølvsten procedure the largest gap is about $4$ points. The gaps close as $\eta$ grows and all three procedures approach rejection probability one.

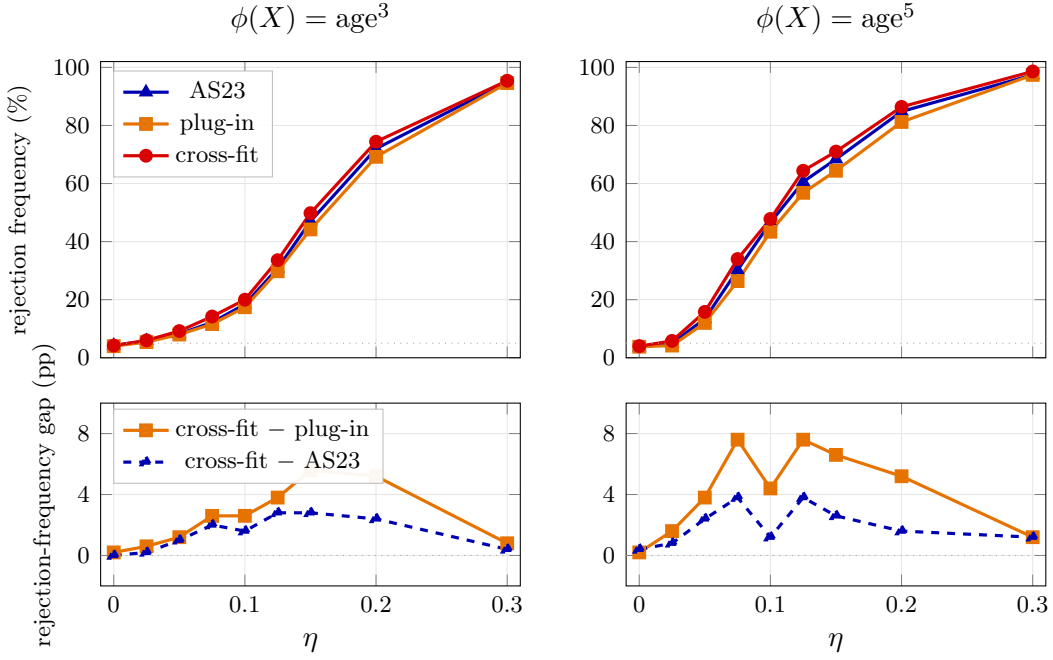
\begin{figure}[t]
\centering \begin{tikzpicture}
\begin{groupplot}[
    group style={
        group size=2 by 2,
        horizontal sep=1.4cm,
        vertical sep=0.6cm,
        xlabels at=edge bottom,
        xticklabels at=edge bottom
    },
    width=0.43\textwidth,
    xmin=-0.01,
    xmax=0.31,
    grid=major,
    grid style={gray!20},
    legend style={
        font=\scriptsize,
        draw=gray!50,
        fill=white,
        fill opacity=0.95,
        text opacity=1
    },
    every axis plot/.append style={very thick},
    tick label style={font=\scriptsize},
    label style={font=\small},
    title style={font=\small}
]

\nextgroupplot[
    height=5.5cm,
    title={\(\phi(X)=\mathrm{age}^{3}\)},
    ylabel={rejection frequency (\%)},
    ylabel style={font=\scriptsize},
    ymin=0,
    ymax=1.02,
    ytick={0,0.2,0.4,0.6,0.8,1.0},
    yticklabels={0,20,40,60,80,100},
    legend pos=north west
]
\addplot[blue!70!black,mark=triangle*]
coordinates {
(0.000,0.042)
(0.025,0.058)
(0.050,0.082)
(0.075,0.122)
(0.100,0.184)
(0.125,0.308)
(0.150,0.470)
(0.200,0.720)
(0.300,0.950)
};
\addlegendentry{AS23}

\addplot[orange!90!black,mark=square*]
coordinates {
(0.000,0.040)
(0.025,0.054)
(0.050,0.080)
(0.075,0.116)
(0.100,0.174)
(0.125,0.298)
(0.150,0.442)
(0.200,0.692)
(0.300,0.946)
};
\addlegendentry{plug-in}

\addplot[red!85!black,mark=*]
coordinates {
(0.000,0.042)
(0.025,0.060)
(0.050,0.092)
(0.075,0.142)
(0.100,0.200)
(0.125,0.336)
(0.150,0.498)
(0.200,0.744)
(0.300,0.954)
};
\addlegendentry{cross-fit}

\draw[dotted,gray!70]
(axis cs:-0.01,0.05) -- (axis cs:0.31,0.05);

\nextgroupplot[
    height=5.5cm,
    title={\(\phi(X)=\mathrm{age}^{5}\)},
    ymin=0,
    ymax=1.02,
    ytick={0,0.2,0.4,0.6,0.8,1.0},
    yticklabels={0,20,40,60,80,100}
]
\addplot[blue!70!black,mark=triangle*]
coordinates {
(0.000,0.036)
(0.025,0.050)
(0.050,0.134)
(0.075,0.302)
(0.100,0.466)
(0.125,0.606)
(0.150,0.684)
(0.200,0.848)
(0.300,0.974)
};

\addplot[orange!90!black,mark=square*]
coordinates {
(0.000,0.038)
(0.025,0.042)
(0.050,0.120)
(0.075,0.264)
(0.100,0.434)
(0.125,0.568)
(0.150,0.644)
(0.200,0.812)
(0.300,0.974)
};

\addplot[red!85!black,mark=*]
coordinates {
(0.000,0.040)
(0.025,0.058)
(0.050,0.158)
(0.075,0.340)
(0.100,0.478)
(0.125,0.644)
(0.150,0.710)
(0.200,0.864)
(0.300,0.986)
};

\draw[dotted,gray!70]
(axis cs:-0.01,0.05) -- (axis cs:0.31,0.05);

\nextgroupplot[
    height=4cm,
    xlabel={\(\eta\)},
    ylabel={rejection-frequency gap (pp)},
    ylabel style={font=\scriptsize},
    ymin=-0.02,
    ymax=0.10,
    ytick={0.00,0.04,0.08},
    yticklabels={0,4,8},
    legend pos=north west
]
\addplot[orange!90!black,mark=square*]
coordinates {
(0.000,0.002)
(0.025,0.006)
(0.050,0.012)
(0.075,0.026)
(0.100,0.026)
(0.125,0.038)
(0.150,0.056)
(0.200,0.052)
(0.300,0.008)
};
\addlegendentry{cross-fit \(-\) plug-in}

\addplot[blue!70!black,mark=triangle*,dashed]
coordinates {
(0.000,0.000)
(0.025,0.002)
(0.050,0.010)
(0.075,0.020)
(0.100,0.016)
(0.125,0.028)
(0.150,0.028)
(0.200,0.024)
(0.300,0.004)
};
\addlegendentry{cross-fit \(-\) AS23}

\draw[dotted,gray!70]
(axis cs:-0.01,0) -- (axis cs:0.31,0);

\nextgroupplot[
    height=4cm,
    xlabel={\(\eta\)},
    ymin=-0.02,
    ymax=0.10,
    ytick={0.00,0.04,0.08},
    yticklabels={0,4,8}
]
\addplot[orange!90!black,mark=square*]
coordinates {
(0.000,0.002)
(0.025,0.016)
(0.050,0.038)
(0.075,0.076)
(0.100,0.044)
(0.125,0.076)
(0.150,0.066)
(0.200,0.052)
(0.300,0.012)
};

\addplot[blue!70!black,mark=triangle*,dashed]
coordinates {
(0.000,0.004)
(0.025,0.008)
(0.050,0.024)
(0.075,0.038)
(0.100,0.012)
(0.125,0.038)
(0.150,0.026)
(0.200,0.016)
(0.300,0.012)
};

\draw[dotted,gray!70]
(axis cs:-0.01,0) -- (axis cs:0.31,0);

\end{groupplot}
\end{tikzpicture}

\caption{Raw rejection frequencies in the OHIE-calibrated simulation. The columns correspond to $\phi(X)=\mathrm{age}^{3}$ ($r=27$) and $\phi(X)=\mathrm{age}^{5}$ ($r=28$). The sample size is $n=1{,}500$, and results are based on 500 Monte Carlo replications for each value of $\eta$. The top row reports raw rejection frequencies for the Anatolyev--S{ø}lvsten, plug-in, and cross-fit procedures. The bottom row reports the differences between the cross-fit rejection frequency and those of the other two procedures. Rejection frequencies are computed over replications in which the corresponding variance estimate is nonnegative. Across the 4,500 replications for each alternative direction, the Anatolyev--S{ø}lvsten variance estimate is negative in 29 replications under the cubic alternative and 137 replications under the fifth-order alternative. The cross-fit variance estimate is never negative across the 9,000 replications.}
\label{fig:phase1b} 
\end{figure}

The gain depends on how the drift is distributed. It is larger under the tail-concentrated fifth-order alternative than under the cubic one throughout, and it persists as $r/n$ grows. Appendix~\ref{subsec:Size-corrected-power} reports size-corrected frequencies across the four sample sizes, where the cubic gap falls from $5.0$ percentage points at $r/n=1.8\%$ to about zero at $5.4\%$ while the fifth-order gap stays between $4.0$ and $7.0$ points. This is consistent with the mechanism of Section~\ref{sec:gen}, in which the gain is governed by $B_{n}/V$ and a concentrated systematic component makes that ratio large. The comparison at the largest $r/n$ should be read with care, since the asymptotic approximation behind the cross-fit procedure becomes less accurate there.

\subsubsection{Numerical stability}

The simulations also reveal a substantial difference in numerical stability. Across the two baseline designs, the proposed cross-fit variance estimate is never negative in 9,000 replications. The Anatolyev--S{ø}lvsten variance estimate is negative in 29 replications under the cubic alternative and 137 replications under the fifth-order alternative, with most occurrences arising at larger values of $\eta$. The same holds in the smaller samples of Appendix~\ref{subsec:Size-corrected-power}: it is negative in as many as 50 of 500 replications at a given grid point, whereas the cross-fit estimate is never negative in any reported design.

These results provide finite-sample evidence that the proposed cross-fit estimator is numerically more stable in the designs considered here. They should not, however, be interpreted as implying that the cross-fit variance estimator is algebraically guaranteed to be nonnegative. Although $M_{ii}M_{jj}+M^{2}_{ij}>0$ under a full-rank design, the products entering the numerator of the cross-fit estimator can have either sign.

\newpage{}

\appendix

\section{Proofs for Section \ref{sec:gen}}

\subsection{Proof of Lemma \ref{lem:exp}}

By \eqref{eq:uhat} and \eqref{eq:split}, $\hat{u}_{i}=\check{u}_{i}+H^{\prime}_{i}\hat{\delta}$ and $\ell^{\prime}_{i}\hat{u}=(\ell^{\prime}_{i}\check{u})+h_{i}\hat{\delta}$. Multiplying the two factors gives
\[
\hat{u}_{i}(\ell^{\prime}_{i}\hat{u})=\check{u}_{i}(\ell^{\prime}_{i}\check{u})+\{\check{u}_{i}h_{i}+(\ell^{\prime}_{i}\check{u})H^{\prime}_{i}\}\hat{\delta}+(\hat{\delta}^{\prime}H_{i})(h_{i}\hat{\delta}).
\]
The coefficient in braces is a $1\times p$ row vector whose transpose is $g_{1i}$. In the last term both factors are scalars, so it equals $\hat{\delta}^{\prime}(H_{i}h_{i})\hat{\delta}$; since $\hat{\delta}^{\prime}A\hat{\delta}=\hat{\delta}^{\prime}(A+A^{\prime})\hat{\delta}/2$ for any $p\times p$ matrix $A$, and $(H_{i}h_{i})^{\prime}=h^{\prime}_{i}H^{\prime}_{i}$, it equals $\hat{\delta}^{\prime}g_{2i}\hat{\delta}$. This proves the first assertion. 

Substituting this identity into \eqref{eq:Vcf} and collecting terms by their degree in $\hat{\delta}$ gives \eqref{eq:expansion}, with $\mathcal{G}_{1}$ and $\mathcal{G}_{2}$ as stated and with
\[
\mathcal{R}_{n}=\sum^{n}_{i=1}\sum_{j\neq i}c_{ij}\bigl\{(g^{\prime}_{1i}\hat{\delta})(\hat{\delta}^{\prime}g_{2j}\hat{\delta})+(g^{\prime}_{1j}\hat{\delta})(\hat{\delta}^{\prime}g_{2i}\hat{\delta})+(\hat{\delta}^{\prime}g_{2i}\hat{\delta})(\hat{\delta}^{\prime}g_{2j}\hat{\delta})\bigr\}.
\]

Each term of the last display is homogeneous in $\hat{\delta}$: the first two are $\|\hat{\delta}\|^{3}$ times the corresponding term at $\hat{\delta}/\|\hat{\delta}\|$, and the third is $\|\hat{\delta}\|^{4}$ times it. Since $\hat{\delta}/\|\hat{\delta}\|$ is a unit vector, taking absolute values and then the supremum over unit vectors gives the stated bound.

Finally, if $\ell^{\prime}_{i}H=0$ for every $i$ then $h_{i}=0$, so $g_{2i}=0$ and $g_{1i}=(\ell^{\prime}_{i}\check{u})H_{i}$. Every term in $\mathcal{R}_{n}$ carries a factor $g_{2i}$ or $g_{2j}$, so $\mathcal{R}_{n}=0$. 

\subsection{Proof of Theorem \ref{thm:null}}

Under $H_{0}$ we have $\Delta_{i}=0$, so $\check{u}_{i}=e_{i}+\omega_{i}$ in the notation of Lemma \ref{lem:exp}, and \eqref{eq:expansion} gives
\[
\hat{V}_{\mathrm{cf}}=\sum^{n}_{i=1}\sum_{j\neq i}c_{ij}g_{0i}g_{0j}+\mathcal{G}^{\prime}_{1}\hat{\delta}+\hat{\delta}^{\prime}\mathcal{G}_{2}\hat{\delta}+\mathcal{R}_{n}.
\]
The last three terms are $o_{p}(V)$: by the Cauchy--Schwarz inequality $|\mathcal{G}^{\prime}_{1}\hat{\delta}|\le\|\hat{\delta}\|\,\|\mathcal{G}_{1}\|$ and $|\hat{\delta}^{\prime}\mathcal{G}_{2}\hat{\delta}|\le\|\hat{\delta}\|^{2}\|\mathcal{G}_{2}\|$, and \eqref{eq:condest} bounds the first two. For the third, Lemma \ref{lem:exp} gives $|\mathcal{R}_{n}|\le\|\hat{\delta}\|^{3}\mathcal{G}_{3}+\|\hat{\delta}\|^{4}\mathcal{G}_{4}$, which is $o_{p}(V)$ by \eqref{eq:condest}. It therefore suffices to show that $\sum^{n}_{i=1}\sum_{j\neq i}c_{ij}g_{0i}g_{0j}=V+o_{p}(V)$.

Write $g_{0i}=e_{i}(\ell^{\prime}_{i}e)+R_{i}$ with
\[
R_{i}=\omega_{i}(\ell^{\prime}_{i}e)+e_{i}(\ell^{\prime}_{i}\omega)+\omega_{i}(\ell^{\prime}_{i}\omega).
\]
By Assumption \ref{as:gen}, $\max_{i}|e_{i}(\ell^{\prime}_{i}e)|=O_{p}(b^{2}_{n})$ and
\begin{eqnarray*}
\max_{i}|R_{i}| & \le & \max_{i}|\omega_{i}|\max_{i}|\ell^{\prime}_{i}e|+\max_{i}|e_{i}|\max_{i}|\ell^{\prime}_{i}\omega|+\max_{i}|\omega_{i}|\max_{i}|\ell^{\prime}_{i}\omega|\\
 & = & o_{p}(b^{-2}_{n}),
\end{eqnarray*}
so that $\max_{i,j}\{|e_{i}(\ell^{\prime}_{i}e)R_{j}|+|R_{i}e_{j}(\ell^{\prime}_{j}e)|+|R_{i}R_{j}|\}=o_{p}(1)$. By \eqref{eq:V},
\begin{align*}
\sum^{n}_{i=1}\sum_{j\neq i}c_{ij}g_{0i}g_{0j}-V & =\sum^{n}_{i=1}\sum_{j\neq i}\{c_{ij}e_{i}(\ell^{\prime}_{i}e)e_{j}(\ell^{\prime}_{j}e)-2a^{2}_{ij}\mathbb{E}[e^{2}_{i}]\mathbb{E}[e^{2}_{j}]\}\\
 & \quad+\sum^{n}_{i=1}\sum_{j\neq i}c_{ij}(e_{i}(\ell^{\prime}_{i}e)R_{j}+R_{i}e_{j}(\ell^{\prime}_{j}e)+R_{i}R_{j}).
\end{align*}
The first term is $o_{p}(V)$ by \eqref{eq:cond}. For the second, Assumption \ref{as:gen} gives $C_{n}=O(A_{n})$, so
\begin{eqnarray*}
\left|\sum^{n}_{i=1}\sum_{j\neq i}c_{ij}(e_{i}(\ell^{\prime}_{i}e)R_{j}+R_{i}e_{j}(\ell^{\prime}_{j}e)+R_{i}R_{j})\right| & \le & C_{n}\max_{i,j}\{|e_{i}(\ell^{\prime}_{i}e)R_{j}|+|R_{i}e_{j}(\ell^{\prime}_{j}e)|+|R_{i}R_{j}|\}\\
 & = & O(A_{n})o_{p}(1)=o_{p}(A_{n}),
\end{eqnarray*}
which is $o_{p}(V)$ because $V\ge2c^{2}A_{n}$. The conclusion follows.

\subsection{Proof of Lemma \ref{lem:bias}}

Under $H_{1}$ we have $\check{u}_{i}=\pi_{i}+e_{i}$ with $\pi_{i}=\Delta_{i}+\omega_{i}$, so the leading term of Lemma \ref{lem:exp} is $g_{0i}=\check{u}_{i}(\ell^{\prime}_{i}\check{u})=\check{A}_{i}$ and $\sum^{n}_{i=1}\sum_{j\neq i}c_{ij}g_{0i}g_{0j}=\check{V}_{\mathrm{cf}}$.

For the plug-in estimator, $\check{u}^{2}_{i}=\pi^{2}_{i}+2\pi_{i}e_{i}+e^{2}_{i}$ gives $\mathbb{E}[\check{u}^{2}_{i}]=\pi^{2}_{i}+\sigma^{2}_{i}$, and for $i\neq j$ the independence of $e_{i}$ and $e_{j}$ gives $\mathbb{E}[\check{u}^{2}_{i}\check{u}^{2}_{j}]=\mathbb{E}[\check{u}^{2}_{i}]\mathbb{E}[\check{u}^{2}_{j}]$. Subtracting $V=2\sum^{n}_{i=1}\sum_{j\neq i}a^{2}_{ij}\sigma^{2}_{i}\sigma^{2}_{j}$ yields \eqref{eq:biasplug}. 

For the cross-fit estimator, since $\ell^{\prime}_{i}\Delta=0$, we have $\ell^{\prime}_{i}\check{u}=\lambda_{i}+\ell^{\prime}_{i}e$ with $\lambda_{i}=\ell^{\prime}_{i}\omega$, so 
\[
\check{A}_{i}=\pi_{i}\lambda_{i}+\pi_{i}(\ell^{\prime}_{i}e)+\lambda_{i}e_{i}+e_{i}(\ell^{\prime}_{i}e).
\]
Multiplying $\check{A}_{i}$ by $\check{A}_{j}$ for $i\neq j$ gives sixteen products, which we take in order of their degree in $e$, counting $\ell^{\prime}_{i}e$ as degree one. The one of degree zero is $\pi_{i}\lambda_{i}\pi_{j}\lambda_{j}$. The four of degree one carry a single mean-zero factor and vanish. Of the six of degree two, $\lambda_{i}\lambda_{j}\mathbb{E}[e_{i}e_{j}]=0$, and the rest follow from
\[
\mathbb{E}[(\ell^{\prime}_{i}e)(\ell^{\prime}_{j}e)]=(L\Sigma L^{\prime})_{ij},\quad\mathbb{E}[(\ell^{\prime}_{i}e)e_{j}]=L_{ij}\sigma^{2}_{j},\quad\text{and }\mathbb{E}[e_{i}(\ell^{\prime}_{i}e)]=L_{ii}\sigma^{2}_{i}.
\]
A product of degree three survives only if its three error factors can share one index. In $\lambda_{i}e_{i}\cdot e_{j}(\ell^{\prime}_{j}e)$ two of them are already fixed at $i$ and $j$, so it vanishes, and the same holds for $e_{i}(\ell^{\prime}_{i}e)\cdot\lambda_{j}e_{j}$; in $\pi_{i}(\ell^{\prime}_{i}e)\cdot e_{j}(\ell^{\prime}_{j}e)$ only $e_{j}$ is fixed and the two free indices can both be set to $j$, giving $\mathbb{E}[(\ell^{\prime}_{i}e)e_{j}(\ell^{\prime}_{j}e)]=L_{ij}L_{jj}\mu_{3j}$, and symmetrically $\mathbb{E}[e_{i}(\ell^{\prime}_{i}e)(\ell^{\prime}_{j}e)]=L_{ji}L_{ii}\mu_{3i}$. The single product of degree four is
\[
\mathbb{E}[e_{i}(\ell^{\prime}_{i}e)e_{j}(\ell^{\prime}_{j}e)]=(L_{ii}L_{jj}+L_{ij}L_{ji})\sigma^{2}_{i}\sigma^{2}_{j},
\]
where only second moments enter because $i\neq j$. Collecting the nine surviving products, $\mathbb{E}[\check{A}_{i}\check{A}_{j}]$ is the sum of the eight terms in braces in \eqref{eq:biascf} and $(L_{ii}L_{jj}+L_{ij}L_{ji})\sigma^{2}_{i}\sigma^{2}_{j}$. Multiplying by $c_{ij}$ and summing over $i\neq j$ gives $\mathbb{E}[\check{V}_{\mathrm{cf}}]$, in which the last term contributes exactly $V$ by the definition of $c_{ij}$; subtracting $V$ leaves \eqref{eq:biascf}.

\subsection{Proof of Theorem \ref{thm:alt}}

Decompose 
\[
\hat{V}_{\mathrm{cf}}-V=\bigl(\hat{V}_{\mathrm{cf}}-\check{V}_{\mathrm{cf}}\bigr)+\bigl(\check{V}_{\mathrm{cf}}-\mathbb{E}[\check{V}_{\mathrm{cf}}]\bigr)+B^{\mathrm{cf}}_{n}.
\]
The first term is $o_{p}(V)$ by the third condition of \eqref{eq:condalt}, the second is $o_{p}(V)$ by the first, and the third is $o(V)$ by the last. Hence $\hat{V}_{\mathrm{cf}}/V\overset{p}{\to}1$.

Similarly 
\[
\hat{V}-(V+B_{n})=\bigl(\hat{V}-\check{V}\bigr)+\bigl(\check{V}-\mathbb{E}[\check{V}]\bigr),
\]
both terms being $o_{p}(V+B_{n})$ by the fourth and second conditions of \eqref{eq:condalt}. Hence $\hat{V}/(V+B_{n})\overset{p}{\to}1$.

For the last assertion, every summand of \eqref{eq:biasplug} is nonnegative, so discarding the quartic term and using $\sigma^{2}_{i}\ge c$ together with the symmetry of the weights,
\[
B_{n}\ge2\sum^{n}_{i=1}\sum_{j\neq i}a^{2}_{ij}\bigl(\pi^{2}_{i}\sigma^{2}_{j}+\sigma^{2}_{i}\pi^{2}_{j}\bigr)\ge4c\sum^{n}_{i=1}\sum_{j\neq i}a^{2}_{ij}\pi^{2}_{i}.
\]
Since $V\le2C^{2}A_{n}$, 
\[
\frac{B_{n}}{V}\ge\frac{2c}{C^{2}}\cdot\frac{1}{A_{n}}\sum^{n}_{i=1}\sum_{j\neq i}a^{2}_{ij}\pi^{2}_{i},
\]
which is bounded away from zero by assumption. The cross-fit bias \eqref{eq:biascf} admits no matching lower bound: its terms carry $\pi_{i}\lambda_{i}$ and $\mu_{3i}$, whose signs are unrestricted.

\subsection{Proof of Corollary \ref{cor:pow}}

By Theorem \ref{thm:alt}, $\hat{V}_{\mathrm{cf}}/V\overset{p}{\to}1$ and $\hat{V}/(V+B_{n})\overset{p}{\to}1$. Writing 
\[
\frac{\hat{V}}{\hat{V}_{\mathrm{cf}}}=\frac{\hat{V}}{V+B_{n}}\cdot\frac{V+B_{n}}{V}\cdot\frac{V}{\hat{V}_{\mathrm{cf}}}
\]
and rearranging gives $\hat{V}/\hat{V}_{\mathrm{cf}}=(1+B_{n}/V)(1+o_{p}(1))$. The two statistics share the numerator $Q_{n}$, so $T^{\mathrm{cf}}_{n}/T_{n}=(\hat{V}/\hat{V}_{\mathrm{cf}})^{1/2}$, and taking square roots gives the first assertion. If $\liminf_{n}B_{n}/V\ge\gamma>0$ then $P(\hat{V}\ge\hat{V}_{\mathrm{cf}})\to1$, and on that event $|T^{\mathrm{cf}}_{n}|=|Q_{n}|/\sqrt{\hat{V}_{\mathrm{cf}}}\ge|Q_{n}|/\sqrt{\hat{V}}=|T_{n}|$, whatever the sign of $Q_{n}$. A two-sided test rejecting when $|T_{n}|$ exceeds a common critical value therefore rejects no less often when $\hat{V}_{\mathrm{cf}}$ is used.

\subsection{Proof of Corollary \ref{cor:sat}}

By Theorem \ref{thm:alt}, $\hat{V}=(V+B_{n})\{1+o_{p}(1)\}$, and by assumption $V+B_{n}=S_{n}\{1+o(1)\}$, so $\hat{V}=S_{n}\{1+o_{p}(1)\}$. With $Q_{n}=N_{n}+o_{p}(S^{1/2}_{n})$,
\[
T_{n}=\frac{Q_{n}}{\hat{V}^{1/2}}=\frac{N_{n}+o_{p}(S^{1/2}_{n})}{S^{1/2}_{n}\{1+o_{p}(1)\}}=\bigl\{\tau_{n}+o_{p}(1)\bigr\}\bigl\{1+o_{p}(1)\bigr\},
\]
which is the first assertion. For the second, Corollary \ref{cor:pow} gives $T^{\mathrm{cf}}_{n}/T_{n}=(1+B_{n}/V)^{1/2}\{1+o_{p}(1)\}$. Every term of \eqref{eq:biasplug} is nonnegative, so $B_{n}\ge S_{n}$, and $V+B_{n}=S_{n}\{1+o(1)\}$ forces $V/S_{n}\to0$, hence $B_{n}/V\to\infty$ and the ratio diverges.

Finally, suppose $\limsup_{n}|\tau_{n}|<c$. The first assertion gives $|T_{n}|\le\bigl\{|\tau_{n}|+o_{p}(1)\bigr\}\bigl\{1+o_{p}(1)\bigr\}$, so $P(|T_{n}|>c)\to0$. If in addition $\liminf_{n}|\tau_{n}|>0$, then $|T_{n}|$ is bounded away from zero with probability approaching one, and $|T^{\mathrm{cf}}_{n}|=|T_{n}|\,(T^{\mathrm{cf}}_{n}/T_{n})\overset{p}{\to}\infty$, so $P(|T^{\mathrm{cf}}_{n}|>c)\to1$.

\subsection{Sufficient conditions for Theorem \ref{thm:alt}} \label{app:lln}

The first two conditions of \eqref{eq:condalt} and the requirement $B^{\mathrm{cf}}_{n}=o(V)$ are statements about weighted sums over pairs of observations. This appendix reduces them to a few aggregate quantities that each application can bound directly, which keeps the verifications below short. Recall $\pi_{i}=\Delta_{i}+\omega_{i}$, $\lambda_{i}=\ell^{\prime}_{i}\omega$, and write
\[
\xi_{n}=\max_{i}\max(|\pi_{i}|,|\lambda_{i}|),\quad\mathcal{C}_{2}=\sum^{n}_{i=1}\sum_{j\neq i}c^{2}_{ij},\quad\bar{c}_{j}=\sum_{i\neq j}|c_{ij}|,
\]
\[
L^{2}_{F}=\sum^{n}_{i=1}\sum_{j\neq i}L^{2}_{ij},\qquad S^{2}_{F}=\sum^{n}_{i=1}\sum_{j\neq i}(L\Sigma L^{\prime})^{2}_{ij}.
\]
Let $\tilde{C}$ be the matrix with off-diagonal entries $c_{ij}$ and zero diagonal, and $D_{v}=\operatorname{diag}(v)$. Note that $\sum_{j}\bar{c}_{j}=C_{n}$ using the notation of Section \ref{sec:gen}.
\begin{lem} \label{lem:lln}
Suppose Assumption \ref{as:gen} holds, that $\ell^{\prime}_{i}\Delta=0$, $\|L\|\le1$, $\max_{i}\mathbb{E}[e^{6}_{i}]\le C$, and that \eqref{eq:cond} holds. If $\xi_{n}=O(1)$ and
\begin{equation}
\mathcal{C}_{2}=o(A^{2}_{n}),\qquad\max_{j}\bar{c}_{j}=o(A_{n}),\qquad\mathcal{C}^{1/2}_{2}(L_{F}+S_{F})=o(A_{n}),\label{eq:condagg}
\end{equation}
then the first two conditions of \eqref{eq:condalt} hold and $B^{\mathrm{cf}}_{n}=o(V)$.
\end{lem}

\begin{proof}
The proof has five steps. Steps 1--3 give the first condition of \eqref{eq:condalt}, Step 4 the second, and Step 5 the bound $B^{\mathrm{cf}}_{n}=o(V)$. We use repeatedly that $\operatorname{Var}(e^{\prime}Be)\le C\|B\|^{2}_{F}$ for symmetric $B$ (for general $B$, apply it to $(B+B^{\prime})/2$, whose Frobenius norm is no larger) and independent mean-zero errors with four bounded moments, and that $\|B_{1}B_{2}\|_{F}\le\min(\|B_{1}\|\,\|B_{2}\|_{F},\|B_{1}\|_{F}\|B_{2}\|)$. Since $\|L\|\le1$ we have $(LL^{\prime})_{jj}\le1$ for every $j$.

\subsubsection*{Step 1. Decomposition into $T_{1}$ and $T_{2}$}

Put $\nu_{i}=e_{i}(\ell^{\prime}_{i}e)-L_{ii}\sigma^{2}_{i}$, $\alpha_{i}=\mathbb{E}[\check{A}_{i}]=\pi_{i}\lambda_{i}+L_{ii}\sigma^{2}_{i}$ and $\eta_{i}=\check{A}_{i}-\alpha_{i}=\pi_{i}(\ell^{\prime}_{i}e)+\lambda_{i}e_{i}+\nu_{i}$, and let $\bar{\alpha}_{j}=\sum_{i\neq j}c_{ij}\alpha_{i}$. Because $c_{ij}=c_{ji}$,
\[
\sum^{n}_{i=1}\sum_{j\neq i}c_{ij}\bigl\{\check{A}_{i}\check{A}_{j}-\mathbb{E}[\check{A}_{i}\check{A}_{j}]\bigr\}=2\sum_{j}\bar{\alpha}_{j}\eta_{j}+\sum^{n}_{i=1}\sum_{j\neq i}c_{ij}\bigl\{\eta_{i}\eta_{j}-\mathbb{E}[\eta_{i}\eta_{j}]\bigr\}=:T_{1}+T_{2}.
\]

\subsubsection*{Step 2. The term $T_{1}$}

We show that $T_{1}=o_{p}(V)$. Since $V\ge2c^{2}A_{n}$, it is sufficient by Chebyshev's inequality to show $\operatorname{Var}(T_{1})=o(A^{2}_{n})$, which we now derive from the hypothesis $\max_{j}\bar{c}_{j}=o(A_{n})$. Substituting $\eta_{j}=\pi_{j}(\ell^{\prime}_{j}e)+\lambda_{j}e_{j}+\nu_{j}$ into $T_{1}=2\sum_{j}\bar{\alpha}_{j}\eta_{j}$ yields
\[
\tfrac{1}{2}T_{1}=\underbrace{\sum_{j}\bar{\alpha}_{j}\pi_{j}(\ell^{\prime}_{j}e)}_{(L^{\prime}g)^{\prime}e}+\underbrace{\sum_{j}\bar{\alpha}_{j}\lambda_{j}e_{j}}_{(D_{\bar{\alpha}}\lambda)^{\prime}e}+\underbrace{\sum_{j}\bar{\alpha}_{j}\nu_{j}}_{e^{\prime}D_{\bar{\alpha}}Le-\mathbb{E}[e^{\prime}D_{\bar{\alpha}}Le]},
\]
where $g_{j}=\bar{\alpha}_{j}\pi_{j}$ and $D_{\bar{\alpha}}=\operatorname{diag}(\bar{\alpha}_{1},\ldots,\bar{\alpha}_{n})$. The first identity holds because $\ell^{\prime}_{j}e$ is the $j$th entry of $Le$, and the third because the constants inside the $\nu_{j}$ sum to $\sum_{j}\bar{\alpha}_{j}L_{jj}\sigma^{2}_{j}=\mathbb{E}[e^{\prime}D_{\bar{\alpha}}Le]$. Collecting the first two gives
\[
\tfrac{1}{2}T_{1}=h^{\prime}e+\{e^{\prime}Ne-\mathbb{E}[e^{\prime}Ne]\},\quad h\coloneqq L^{\prime}g+D_{\bar{\alpha}}\lambda,\quad N\coloneqq D_{\bar{\alpha}}L.
\]
The linear form has variance at most $C\|h\|^{2}$. Since $e^{\prime}Ne=e^{\prime}N_{s}e$ with $N_{s}=(N+N^{\prime})/2$ symmetric, the quadratic form has variance at most $C\|N_{s}\|^{2}_{F}$. Since $\operatorname{Var}(X+Y)\le2\operatorname{Var}X+2\operatorname{Var}Y$, it follows that $\operatorname{Var}(T_{1})\le C\{\|h\|^{2}+\|N_{s}\|^{2}_{F}\}$. Now 
\[
\|h\|^{2}\le2\|g\|^{2}+2\sum_{j}\bar{\alpha}^{2}_{j}\lambda^{2}_{j}\le C\sum_{j}\bar{\alpha}^{2}_{j}(1+\pi^{2}_{j}+\lambda^{2}_{j}),
\]
where the first inequality follows from $\|L^{\prime}g\|\le\|L\|\|g\|$ and $\|L\|\le1$, and the second from $\|g\|^{2}=\sum_{j}\bar{\alpha}^{2}_{j}\pi^{2}_{j}$. Moreover 
\[
\|N_{s}\|^{2}_{F}\le\|N\|^{2}_{F}=\sum_{j}\bar{\alpha}^{2}_{j}(LL^{\prime})_{jj}\le\sum_{j}\bar{\alpha}^{2}_{j},
\]
where the first inequality is the triangle inequality for $\|\cdot\|_{F}$ together with $\|N^{\prime}\|_{F}=\|N\|_{F}$, and the last is $(LL^{\prime})_{jj}\le1$. Since $\xi_{n}=O(1)$ and $\sigma^{2}_{i}\le C$ and $|L_{ii}|\le\|L\|\le1$ we have $|\alpha_{i}|\le C$, hence $|\bar{\alpha}_{j}|\le C\bar{c}_{j}$ and
\[
\sum_{j}\bar{\alpha}^{2}_{j}(1+\pi^{2}_{j}+\lambda^{2}_{j})\le C\sum_{j}\bar{c}^{2}_{j}\le C\Bigl(\max_{j}\bar{c}_{j}\Bigr)C_{n}\le C\Bigl(\max_{j}\bar{c}_{j}\Bigr)A_{n}=o(A^{2}_{n}),
\]
where the second inequality follows from $\sum_{j}\bar{c}_{j}=C_{n}$, and the last follows from $C_{n}\le(2/c_{L})A_{n}$; the same bound covers $\|N_{s}\|^{2}_{F}$, so $\operatorname{Var}(T_{1})=o(A^{2}_{n})$.

\subsubsection*{Step 3. The term $T_{2}$}

We show that $T_{2}=o_{p}(V)$. Expanding $\eta_{i}=\pi_{i}(\ell^{\prime}_{i}e)+\lambda_{i}e_{i}+\nu_{i}$ gives
\begin{align*}
\eta_{i}\eta_{j} & =\pi_{i}\pi_{j}(\ell^{\prime}_{i}e)(\ell^{\prime}_{j}e)+\lambda_{i}\lambda_{j}e_{i}e_{j}+\pi_{i}\lambda_{j}(\ell^{\prime}_{i}e)e_{j}+\lambda_{i}\pi_{j}e_{i}(\ell^{\prime}_{j}e)\\
 & \quad+\pi_{i}(\ell^{\prime}_{i}e)\nu_{j}+\pi_{j}(\ell^{\prime}_{j}e)\nu_{i}+\lambda_{i}e_{i}\nu_{j}+\lambda_{j}e_{j}\nu_{i}+\nu_{i}\nu_{j},
\end{align*}
which splits $T_{2}$ into nine groups, four quadratic in $e$, four cubic and one quartic. The quartic group is covered by the hypothesis \eqref{eq:cond}; each of the other eight is shown to have variance $o(A^{2}_{n})$, which with $V\ge2c^{2}A_{n}$ and Chebyshev's inequality makes it $o_{p}(V)$.

The quartic group $\sum^{n}_{i=1}\sum_{j\neq i}c_{ij}\{\nu_{i}\nu_{j}-\mathbb{E}[\nu_{i}\nu_{j}]\}$ involves neither $\pi$ nor $\lambda$. Substituting $e_{i}(\ell^{\prime}_{i}e)=\nu_{i}+L_{ii}\sigma^{2}_{i}$ and using $c_{ij}=c_{ji}$,
\[
\sum^{n}_{i=1}\sum_{j\neq i}c_{ij}e_{i}(\ell^{\prime}_{i}e)e_{j}(\ell^{\prime}_{j}e)=\sum^{n}_{i=1}\sum_{j\neq i}c_{ij}\{\nu_{i}\nu_{j}-\mathbb{E}[\nu_{i}\nu_{j}]\}+2\sum_{j}\bar{w}_{j}\nu_{j}+\sum^{n}_{i=1}\sum_{j\neq i}c_{ij}\bigl(L_{ii}L_{jj}+L_{ij}L_{ji}\bigr)\sigma^{2}_{i}\sigma^{2}_{j}.
\]
Here $\bar{w}_{j}=\sum_{i\neq j}c_{ij}L_{ii}\sigma^{2}_{i}$, and the centering uses $\mathbb{E}[\nu_{i}\nu_{j}]=L_{ij}L_{ji}\sigma^{2}_{i}\sigma^{2}_{j}$ for $i\neq j$. By the definition of $c_{ij}$ the third sum equals $2\sum^{n}_{i=1}\sum_{j\neq i}a^{2}_{ij}\sigma^{2}_{i}\sigma^{2}_{j}=V$, so subtracting $V$ from both sides gives
\[
\sum^{n}_{i=1}\sum_{j\neq i}\bigl\{ c_{ij}e_{i}(\ell^{\prime}_{i}e)e_{j}(\ell^{\prime}_{j}e)-2a^{2}_{ij}\sigma^{2}_{i}\sigma^{2}_{j}\bigr\}=\sum^{n}_{i=1}\sum_{j\neq i}c_{ij}\{\nu_{i}\nu_{j}-\mathbb{E}[\nu_{i}\nu_{j}]\}+2\sum_{j}\bar{w}_{j}\nu_{j},
\]
whose left-hand side is that of \eqref{eq:cond}. As in the treatment of $T_{1}$, $\sum_{j}\bar{w}_{j}\nu_{j}=e^{\prime}D_{\bar{w}}Le-\mathbb{E}[e^{\prime}D_{\bar{w}}Le]$, and $|L_{ii}\sigma^{2}_{i}|\le C$ gives $|\bar{w}_{j}|\le C\bar{c}_{j}$, so its variance is at most $C\sum_{j}\bar{w}^{2}_{j}\le C\bigl(\max_{j}\bar{c}_{j}\bigr)A_{n}=o(A^{2}_{n})$. The quartic group is therefore $o_{p}(V)$ by \eqref{eq:cond}.

Four groups are quadratic in $e$. The two carrying one $\pi$ and one $\lambda$ are equal by the symmetry of $c$, so the four give three centered quadratic forms,
\begin{align*}
\sum^{n}_{i=1}\sum_{j\neq i}c_{ij}\pi_{i}\pi_{j}\{(\ell^{\prime}_{i}e)(\ell^{\prime}_{j}e)-\mathbb{E}[(\ell^{\prime}_{i}e)(\ell^{\prime}_{j}e)]\} & =e^{\prime}B_{1}e-\mathbb{E}[e^{\prime}B_{1}e], & B_{1} & \coloneqq L^{\prime}D_{\pi}\tilde{C}D_{\pi}L,\\
\sum^{n}_{i=1}\sum_{j\neq i}c_{ij}\lambda_{i}\lambda_{j}\{e_{i}e_{j}-\mathbb{E}[e_{i}e_{j}]\} & =e^{\prime}B_{2}e-\mathbb{E}[e^{\prime}B_{2}e], & B_{2} & \coloneqq D_{\lambda}\tilde{C}D_{\lambda},\\
2\sum^{n}_{i=1}\sum_{j\neq i}c_{ij}\pi_{i}\lambda_{j}\{(\ell^{\prime}_{i}e)e_{j}-\mathbb{E}[(\ell^{\prime}_{i}e)e_{j}]\} & =2\bigl\{ e^{\prime}B_{3}e-\mathbb{E}[e^{\prime}B_{3}e]\bigr\}, & B_{3} & \coloneqq L^{\prime}D_{\pi}\tilde{C}D_{\lambda}.
\end{align*}
By $\|L\|\le1$, their Frobenius norms are at most $\|D_{\pi}\tilde{C}D_{\pi}\|_{F}$, $\|D_{\lambda}\tilde{C}D_{\lambda}\|_{F}$ and $\|D_{\pi}\tilde{C}D_{\lambda}\|_{F}$. Each of these three matrices has zero diagonal and $(i,j)$ entry bounded in modulus by $\xi^{2}_{n}|c_{ij}|$, so each norm is at most $\xi^{2}_{n}\mathcal{C}^{1/2}_{2}$ and each group has variance $O(\xi^{4}_{n}\mathcal{C}_{2})=o(A^{2}_{n})$.

The remaining four groups are cubic in $e$. They involve the products $\pi_{i}(\ell^{\prime}_{i}e)\nu_{j}$, $\pi_{j}(\ell^{\prime}_{j}e)\nu_{i}$, $\lambda_{i}e_{i}\nu_{j}$ and $\lambda_{j}e_{j}\nu_{i}$. Relabelling $i$ and $j$ and using $c_{ij}=c_{ji}$ turns the second group into the first and the fourth into the third, so only two distinct sums arise,
\[
\sum^{n}_{i=1}\sum_{j\neq i}c_{ij}\pi_{i}(\ell^{\prime}_{i}e)\nu_{j}\quad\text{and}\quad\sum^{n}_{i=1}\sum_{j\neq i}c_{ij}\lambda_{i}e_{i}\nu_{j},
\]
each appearing twice. By Lemma \ref{lem:cubic} below, the variance of each sum is at most $C\xi^{2}_{n}\mathcal{C}_{2}$: take $G=L$ and $v=\pi$ for the first, $G=I$ and $v=\lambda$ for the second. Since $\xi_{n}=O(1)$ and $\mathcal{C}_{2}=o(A^{2}_{n})$, both are $o(A^{2}_{n})$. In the second, $U=\tilde{C}D_{\lambda}$ has zero diagonal, so the pieces bounded by $\sum_{k}U^{2}_{kk}$ vanish. This exhausts the nine terms in $T_{2}$, so $T_{2}=o_{p}(V)$, and with $T_{1}$ the first condition of \eqref{eq:condalt} follows.

\subsubsection*{Step 4. The plug-in estimator}

The second condition of \eqref{eq:condalt} is the same statement for the plug-in estimator. The argument just given yields it at $L=I$, $\lambda=\pi$ and weights $2a^{2}_{ij}$ in place of $c_{ij}$. Since $\|L\|\le1$ gives $L_{ii}L_{jj}+L_{ij}L_{ji}\le2$ and hence $c_{ij}\ge a^{2}_{ij}$, the plug-in analogues of $\mathcal{C}_{2}$ and of $\max_{j}\bar{c}_{j}$ are bounded up to constants by the quantities that \eqref{eq:condagg} controls. At $L=I$ the two Frobenius quantities vanish, $L_{F}=S_{F}=0$. The target is $V+B_{n}$, which is no smaller than $V$, so a bound of $o(A^{2}_{n})$ still suffices. The one group that does not transfer is the quartic one, since \eqref{eq:cond} is a hypothesis about the cross-fit weights. But at $L=I$ the variables $\nu_{i}=e^{2}_{i}-\sigma^{2}_{i}$ are independent across $i$, so
\[
\operatorname{Var}\Bigl(\sum^{n}_{i=1}\sum_{j\neq i}2a^{2}_{ij}\{\nu_{i}\nu_{j}-\mathbb{E}[\nu_{i}\nu_{j}]\}\Bigr)=2\sum^{n}_{i=1}\sum_{j\neq i}(2a^{2}_{ij})^{2}\mathbb{E}[\nu^{2}_{i}]\mathbb{E}[\nu^{2}_{j}]\le C\mathcal{C}_{2}=o(A^{2}_{n}),
\]
which needs only four moments. 

\subsubsection*{Step 5. The cross-fit bias}

It remains to show the third conclusion, $B^{\mathrm{cf}}_{n}=o(V)$. We bound the eight terms of \eqref{eq:biascf}. Assumption \ref{as:gen} gives $\max_{i}|\lambda_{i}|=o_{p}(1)$, so the first three terms are at most $C\xi_{n}C_{n}(\max_{i}|\lambda_{i}|)(1+\xi_{n}(\max_{i}|\lambda_{i}|))=o(A_{n})$. By the Cauchy--Schwarz inequality the fourth is at most 
\[
(\sum^{n}_{i=1}\sum_{j\neq i}c^{2}_{ij}\pi^{2}_{i}\pi^{2}_{j})^{1/2}S_{F}\le\xi^{2}_{n}\mathcal{C}^{1/2}_{2}S_{F},
\]
the fifth and sixth are at most $C\xi^{2}_{n}\mathcal{C}^{1/2}_{2}L_{F}$, and the last two are at most $C\xi_{n}\mathcal{C}^{1/2}_{2}L_{F}$, the third moments satisfying $|\mu_{3i}|\le C$. All are $o(A_{n})$ by \eqref{eq:condagg}, and $V\ge2c^{2}A_{n}$, so $B^{\mathrm{cf}}_{n}=o(V)$. 
\end{proof}

\begin{lem} \label{lem:cubic}
Suppose Assumption \ref{as:gen} holds, that $\|L\|\le1$ and $\max_{i}\mathbb{E}[e^{6}_{i}]\le C$, and put $\nu_{i}=e_{i}(\ell^{\prime}_{i}e)-L_{ii}\sigma^{2}_{i}$. Let $G$ be a matrix with rows $g^{\prime}_{i}$ and $\|G\|\le1$, and let $v$ be a vector with $\max_{i}|v_{i}|\le\xi_{n}$. Then
\[
\operatorname{Var}\Bigl(\sum^{n}_{i=1}\sum_{j\neq i}c_{ij}v_{i}(g^{\prime}_{i}e)\nu_{j}\Bigr)\le C\xi^{2}_{n}\mathcal{C}_{2}.
\]
\end{lem}

\begin{proof}
The argument has three parts. First, the sum is written as a cubic form in $e$ plus a linear form. Second, the cubic form is split by index coincidence into seven pieces. Third, the variance of each piece is bounded by $\|\gamma\|^{2}_{F}$, by $\sum_{k}U^{2}_{kk}$ or directly by $\xi^{2}_{n}\mathcal{C}_{2}$, the first two being themselves at most $\xi^{2}_{n}\mathcal{C}_{2}$. 

First, substituting $g^{\prime}_{i}e=\sum_{k}G_{ik}e_{k}$ and then $\nu_{j}=e_{j}\sum_{m}L_{jm}e_{m}-L_{jj}\sigma^{2}_{j}$ gives 
\begin{align}
\sum^{n}_{i=1}\sum_{j\neq i}c_{ij}v_{i}(g^{\prime}_{i}e)\nu_{j} & =\sum_{j}\sum_{k}\Bigl(\sum_{i\neq j}c_{ij}v_{i}G_{ik}\Bigr)e_{k}\nu_{j}\nonumber \\
 & =\sum_{j}\sum_{k}U_{jk}e_{k}\Bigl(e_{j}\sum_{m}L_{jm}e_{m}-L_{jj}\sigma^{2}_{j}\Bigr)\nonumber \\
 & =\sum_{k,j,m}\gamma_{kjm}e_{k}e_{j}e_{m}-q^{\prime}e,\label{eq:cubicid}
\end{align}
where $U=\tilde{C}D_{v}G$. The second equality uses that its $(j,k)$ entry is $\sum_{i\neq j}c_{ij}v_{i}G_{ik}$, where we use the symmetry of $c$ and the zero diagonal of $\tilde{C}$. The third collects the cubic terms with $\gamma_{kjm}=L_{jm}U_{jk}$ and the linear ones with $q_{k}=\sum_{j}\gamma_{kjj}\sigma^{2}_{j}$; write $\|\gamma\|^{2}_{F}=\sum_{k,j,m}\gamma^{2}_{kjm}$.

Second, the cubic form $\sum_{k,j,m}\gamma_{kjm}e_{k}e_{j}e_{m}-\sum_{k,j,m}\mathbb{E}[\gamma_{kjm}e_{k}e_{j}e_{m}]$ can be decomposed into
\begin{equation}
P_{3}+Q_{1}+Q_{2}+Q_{3}+\sum_{k}\gamma_{kkk}(e^{3}_{k}-\mu_{3k})+(\ell^{(1)}+\ell^{(2)}+\ell^{(3)})^{\prime}e,\label{eq:cubicsplit}
\end{equation}
where $P_{3}$ collects the terms with $k$, $j$, $m$ distinct, the $Q_{r}$ are the three two-index-coincidence sums 
\begin{align*}
Q_{1} & =\sum^{n}_{k=1}\sum_{m\neq k}\gamma_{kkm}(e^{2}_{k}-\sigma^{2}_{k})e_{m},\\
Q_{2} & =\sum^{n}_{k=1}\sum_{j\neq k}\gamma_{kjk}(e^{2}_{k}-\sigma^{2}_{k})e_{j},\\
Q_{3} & =\sum^{n}_{k=1}\sum_{j\neq k}\gamma_{kjj}(e^{2}_{j}-\sigma^{2}_{j})e_{k},
\end{align*}
one for each of the three ways two of the indices of $\gamma_{kjm}$ can coincide. In each of them the error whose index is repeated appears squared and centered, and the remaining error appears linearly. The three linear coefficients are $\ell^{(1)}_{m}=\sum_{k\neq m}\gamma_{kkm}\sigma^{2}_{k}$, $\ell^{(2)}_{j}=\sum_{k\neq j}\gamma_{kjk}\sigma^{2}_{k}$ and $\ell^{(3)}_{k}=\sum_{j\neq k}\gamma_{kjj}\sigma^{2}_{j}$.

Combining (\ref{eq:cubicid}), (\ref{eq:cubicsplit}) and $\mathbb{E}[q^{\prime}e]=0$ gives
\[
\sum^{n}_{i=1}\sum_{j\neq i}c_{ij}v_{i}\{(g^{\prime}_{i}e)\nu_{j}-\mathbb{E}[(g^{\prime}_{i}e)\nu_{j}]\}=P_{3}+Q_{1}+Q_{2}+Q_{3}+\sum_{k}\gamma_{kkk}(e^{3}_{k}-\mu_{3k})+(\ell^{(1)}+\ell^{(2)}+\ell^{(3)})^{\prime}e-q^{\prime}e.
\]
Note that the two vectors $\ell^{(3)}_{k}=\sum_{j\neq k}\gamma_{kjj}\sigma^{2}_{j}$ and $q_{k}=\sum_{j}\gamma_{kjj}\sigma^{2}_{j}$ differ only in the summand $j=k$, so $\ell^{(3)}-q=-(\gamma_{kkk}\sigma^{2}_{k})_{k}$. Then the last display is 
\begin{equation}
P_{3}+Q_{1}+Q_{2}+Q_{3}+\sum_{k}\gamma_{kkk}(e^{3}_{k}-\mu_{3k})-\sum_{k}\gamma_{kkk}\sigma^{2}_{k}e_{k}+(\ell^{(1)}+\ell^{(2)})^{\prime}e.\label{eq:cubic_centered_collecion}
\end{equation}

Finally, each piece of (\ref{eq:cubic_centered_collecion}) has variance at most $C\xi^{2}_{n}\mathcal{C}_{2}$. The bounds below use two quantities, both at most $\xi^{2}_{n}\mathcal{C}_{2}$. The first is $\|\gamma\|^{2}_{F}$:
\[
\|\gamma\|^{2}_{F}=\sum_{j}(LL^{\prime})_{jj}\|U_{j\cdot}\|^{2}\le\|U\|^{2}_{F}\le\|\tilde{C}D_{v}\|^{2}_{F}\|G\|^{2}\le\sum^{n}_{i=1}\sum_{j\neq i}c^{2}_{ij}v^{2}_{i}\le\xi^{2}_{n}\mathcal{C}_{2},
\]
where the equality uses $\gamma_{kjm}=L_{jm}U_{jk}$ and $\sum_{m}L^{2}_{jm}=(LL^{\prime})_{jj}$, the first inequality $(LL^{\prime})_{jj}\le\|L\|^{2}\le1$, the second $U=\tilde{C}D_{v}G$ together with $\|B_{1}B_{2}\|_{F}\le\|B_{1}\|_{F}\|B_{2}\|$, the third $\|G\|\le1$, and the last $v^{2}_{i}\le\xi^{2}_{n}$. The second quantity is $\sum_{k}U^{2}_{kk}$, for which $U_{kk}=\sum_{i\neq k}c_{ki}v_{i}G_{ik}$ and $(G^{\prime}G)_{kk}\le\|G\|^{2}\le1$ give $\sum_{k}U^{2}_{kk}\le\xi^{2}_{n}\mathcal{C}_{2}$ by the Cauchy--Schwarz inequality.

Now take the pieces in turn. \\
$\operatorname{Var}(P_{3})$: Expanding it pairs two coefficients $\gamma_{kjm}$ and $\gamma_{k^{\prime}j^{\prime}m^{\prime}}$, and the expectation vanishes unless $(k^{\prime},j^{\prime},m^{\prime})$ is a permutation of $(k,j,m)$, since the errors are independent with mean zero and the three indices of a triple are distinct. Each of the six permutations contributes at most $\|\gamma\|^{2}_{F}$ by the Cauchy--Schwarz inequality, so $\operatorname{Var}(P_{3})\le C\|\gamma\|^{2}_{F}$.
\\
$\operatorname{Var}(Q_{r})$: The same bound holds for all three, where the paired coefficients carry the same unordered pair of indices. For example,
\[
\operatorname{Var}(Q_{1})=\sum^{n}_{k=1}\sum_{m\neq k}\gamma^{2}_{kkm}(\mathbb{E}[e^{4}_{k}]-\sigma^{4}_{k})\sigma^{2}_{m}+\sum^{n}_{k=1}\sum_{m\neq k}\gamma_{kkm}\gamma_{mmk}\mu_{3k}\mu_{3m},
\]
in which the second sum pairs $\gamma_{kkm}$ with $\gamma_{mmk}$: the multisets differ, but the index pair $\{k,m\}$ is the same. Since $|\mu_{3k}|\le(\mathbb{E}[e^{4}_{k}])^{3/4}\le C$ and $2|\gamma_{kkm}\gamma_{mmk}|\le\gamma^{2}_{kkm}+\gamma^{2}_{mmk}$, that sum is at most $C\|\gamma\|^{2}_{F}$; it vanishes altogether when the errors are symmetric. \\
$\operatorname{Var}(\sum_{k}\gamma_{kkk}(e^{3}_{k}-\mu_{3k}))$: The diagonal piece has variance 
\[
\sum_{k}\gamma^{2}_{kkk}\operatorname{Var}(e^{3}_{k})\le C\sum_{k}U^{2}_{kk},
\]
and this is the only place the sixth moments are needed. \\ $\operatorname{Var}(\sum_{k}\gamma_{kkk}\sigma^{2}_{k}e_{k})$ and $\operatorname{Var}((\ell^{(1)}+\ell^{(2)})^{\prime}e)$: A linear form $c^{\prime}e$ has variance at most $C\|c\|^{2}$, so it is enough to bound $\sum_{k}\gamma^{2}_{kkk}\sigma^{4}_{k}$, $\|\ell^{(1)}\|^{2}$ and $\|\ell^{(2)}\|^{2}$. The first is bounded by 
\[
\sum_{k}L^{2}_{kk}U^{2}_{kk}\sigma^{4}_{k}\le C\sum_{k}U^{2}_{kk},
\]
where the inequality follows from $|L_{kk}|\le\|L\|\le1$ and $\sigma^{2}_{k}\le C$. For $\ell^{(1)}$, put $d_{k}=U_{kk}$, so that $\ell^{(1)}_{m}=\sum_{k\neq m}L_{km}U_{kk}\sigma^{2}_{k}=(L^{\prime}\Sigma d)_{m}-L_{mm}U_{mm}\sigma^{2}_{m}$. Then
\[
\|\ell^{(1)}\|^{2}\le2\|L^{\prime}\Sigma d\|^{2}+2\sum_{m}L^{2}_{mm}U^{2}_{mm}\sigma^{4}_{m}\le C\sum_{k}U^{2}_{kk},
\]
where the first inequality is $\|a-b\|^{2}\le2\|a\|^{2}+2\|b\|^{2}$ and the second follows from $\|L\|\le1$, $|L_{mm}|\le\|L\|\le1$ and $\sigma^{2}_{k}\le C$. For $\ell^{(2)}$, the entries are $\ell^{(2)}_{j}=\sum_{i\neq j}c_{ij}v_{i}(L\Sigma G^{\prime})_{ji}-\gamma_{jjj}\sigma^{2}_{j}$. Splitting the two parts as before and applying the Cauchy--Schwarz inequality to the first, $\sum_{i}(L\Sigma G^{\prime})^{2}_{ji}=(L\Sigma G^{\prime}G\Sigma L^{\prime})_{jj}\le C^{2}$ gives $\|\ell^{(2)}\|^{2}\le C\xi^{2}_{n}\mathcal{C}_{2}$. The seven pieces are dependent, so $\operatorname{Var}(\sum^{7}_{r=1}X_{r})\le7\sum_{r}\operatorname{Var}(X_{r})$ bounds the variance of the sum by $C\xi^{2}_{n}\mathcal{C}_{2}$.
\end{proof}

\section{Assumptions and proof of Proposition \ref{prop:spec} } \label{app:spec}

Let $\theta_{0}(x)=\mathbb{E}[y\mid x]$. 
\begin{assumption}
\label{as:spec} The following conditions hold, in the notation of Section \ref{sub:spec}.
\begin{description}
\item [{(i)}] $\{(y_{i},x_{i})\}^{n}_{i=1}$ are i.i.d. and $\mathbb{E}[e_{i}\mid x_{i}]=0$ with $e_{i}=y_{i}-\theta_{0}(x_{i})$. Under $H_{0}$, $\theta_{0}(x)=\gamma^{\prime}_{0}x$ for some $\gamma_{0}$. The conditional variances satisfy $0<c\le\sigma^{2}_{i}:=\mathbb{E}[e^{2}_{i}|x_{i}]\le C<\infty$ uniformly in $i$, and $\max_{i}\mathbb{E}[|e_{i}|^{6}|x_{i}]\le C$. The regressors have six unconditional moments, $\mathbb{E}\|x_{i}\|^{6}<\infty$. 
\item [{(ii)}] The basis dimension satisfies $K\to\infty$ and $K/n\to0$. Let $B=(b(x_{1}),\ldots,b(x_{n}))^{\prime}$, $P=B(B^{\prime}B)^{-1}B^{\prime}$, and $M=I-P$. The basis is normalized so that $\mathbb{E}[b(x)b(x)^{\prime}]=I_{K}$, and there is a sequence $\{\zeta_{K}\}$ with $\sup_{x}\|b(x)\|\le\zeta_{K}$ and $\zeta^{2}_{K}K/n\to0$. In addition, the basis spans the linear terms: there exists a $K\times G$ matrix $\alpha_K$ with $b(x)^{\prime}\alpha_K=x^{\prime}$ almost surely. In particular $X=B\alpha_K$, so $MX=0$.
\item [{(iii)}] There is a deterministic sequence $1\le b_{n}=o(n)$ such that, conditionally on $X=(x_{1},\ldots,x_{n})'$,
\[
\max_{i}|e_{i}|=O_{p}(b_{n}),\quad\max_{i}|m^{\prime}_{i}e|=O_{p}(b_{n}),\quad\max_{i}|\omega_{i}|+\max_{i}|m^{\prime}_{i}\omega|=o_{p}(b^{-3}_{n}),\quad\|\hat{\delta}\|=O_{p}(n^{-1/2}),
\]
where $m_{i}$ is the $i$th column of $M$.
\item [{(iv)}] Under the alternative the systematic component is bounded and non-degenerate:
\[
\max_{i}|\pi_{i}|=O_{p}(1),\qquad\liminf_{n}A^{-1}_{n}\sum^{n}_{i=1}\sum_{j\neq i}a^{2}_{ij}\pi^{2}_{i}>0,
\]
where $\pi_{i}=\Delta_{i}+\omega_{i}=\theta_{0}(x_{i})-\gamma^{\prime}_{*}x_{i}$.
\end{description}
\end{assumption}

The normalization in (ii) is without loss of generality: $P$, $M$ and $\hat{V}_{\mathrm{cf}}$ are invariant to nonsingular linear transformations of $b$, so one may always pass to the transformed basis. The condition asks that the normalized basis satisfy the growth condition on $\zeta_{K}$, as in \citet[Assumption 3]{Newey1997}. In (iii) one may always take $b_{n}=n^{1/6}$, since $\mathbb{E}[\max_{i}|e_{i}|^{6}\mid X]\le\sum_{i}\mathbb{E}[|e_{i}|^{6}\mid x_{i}]\le Cn$ by (i), and $m^{\prime}_{i}e$ has bounded sixth moments by Rosenthal's inequality; the restriction $b_{n}=o(n)$ therefore costs nothing, and the force of (iii) is the bound it places on $\omega$. With $b_{n}=n^{1/6}$ that bound needs to be $\max_{i}|\omega_{i}|+\max_{i}|m^{\prime}_{i}\omega|=o_{p}(n^{-1/2})$. For splines and power series, the approximation condition in \citet[Assumption 3]{Newey1997} holds with rate $K^{-s/G}$ when the conditional expectation has $s$ continuous derivatives, so the requirement is $\sqrt{n}K^{-s/G}\to0$. The first condition of (iv) asks that the deviation of the conditional expectation from its best linear predictor be bounded, which is vacuous under the null, where the deviation is zero, and it enters through $\xi_{n}=O_{p}(1)$ in Lemma \ref{lem:lln} and through the moment bounds on $\check{A}_{i}$ and $\check{u}_{i}$ in the estimation-error step under the alternative. The second makes the plug-in estimator inconsistent as required in Theorem \ref{thm:alt}.
\begin{proof}
For the specification testing problem, 
\[
a_{ij}=K^{-1/2}P_{ij},\qquad c_{ij}=\frac{2}{K}\frac{P^{2}_{ij}}{M_{ii}M_{jj}+M^{2}_{ij}},\qquad\ell_{i}=m_{i},
\]
where $M=I-P$ and $m_{i}$ denotes the $i$-th column of $M$. We verify the conditions of Theorems \ref{thm:null} and \ref{thm:alt}. All arguments are conditional on $X=(x_{1},\ldots,x_{n})^{\prime}$. To simplify the presentation we give the argument for $G=1$, that is, for scalar $x$; the vector case only replaces the products $x_{i}x_{j}$ by $x_{i}x^{\prime}_{j}$.

\subsubsection*{Step 1. The maintained assumption}

We begin by verifying Assumption \ref{as:gen}. By (ii), $\mathbb{E}\|n^{-1}B^{\prime}B-I\|^{2}\le\zeta^{2}_{K}K/n\to0$, so $\lambda_{\min}(n^{-1}B^{\prime}B)\overset{p}{\to}1$ \citep[proof of Theorem 1]{Newey1997}, and
\[
P_{ii}\le\lambda^{-1}_{\min}(n^{-1}B^{\prime}B)\,n^{-1}\max_{i}\|b(x_{i})\|^{2}\le\lambda^{-1}_{\min}(n^{-1}B^{\prime}B)\,\zeta^{2}_{K}/n=o_{p}(1).
\]
Two consequences are used below. Since $P$ is symmetric and idempotent, $\sum_{i,j}P^{2}_{ij}=\operatorname{tr}(P)=K$, while $\sum_{i}P^{2}_{ii}\le(\max_{i}P_{ii})K$, so that
\[
\sum^{n}_{i=1}\sum_{j\neq i}P^{2}_{ij}=K-\sum_{i}P^{2}_{ii}=K+o_{p}(K);
\]
and since $|P_{ij}|\le\sqrt{P_{ii}P_{jj}}$ for a projection, also
\[
\sum^{n}_{i=1}\sum_{j\neq i}P^{4}_{ij}\le\Bigl(\max_{i}P_{ii}\Bigr)^{2}\sum^{n}_{i=1}\sum_{j\neq i}P^{2}_{ij}=o_{p}(K)=o_{p}(K^{2}).
\]
Consequently $M_{ii}=1-P_{ii}$ is bounded away from zero uniformly in $i$, say $\min_{i}M_{ii}\ge c_{M}>0$, with probability approaching one. Since $M_{ij}=-P_{ij}$ for $i\neq j$, $M_{ii}M_{jj}+M^{2}_{ij}\ge c^{2}_{M}$ uniformly over $i\neq j$ with probability approaching one. This is the leverage condition of Assumption \ref{as:gen} with $c_{L}=c^{2}_{M}$. First, 
\[
A_{n}=\sum^{n}_{i=1}\sum_{j\neq i}a^{2}_{ij}=K^{-1}\sum^{n}_{i=1}\sum_{j\neq i}P^{2}_{ij}=1+o_{p}(1).
\]
Moreover, 
\[
2c^{2}A_{n}\le V=2\sum^{n}_{i=1}\sum_{j\neq i}a^{2}_{ij}\sigma^{2}_{i}\sigma^{2}_{j}\le2C^{2}A_{n},
\]
because $\sigma^{2}_{i}$ is bounded above and away from zero. Also,
\[
C_{n}=\sum^{n}_{i=1}\sum_{j\neq i}|c_{ij}|\le\frac{C}{K}\sum^{n}_{i=1}\sum_{j\neq i}P^{2}_{ij}=O_{p}(1)=O_{p}(A_{n}).
\]
The maximal bounds for $e_{i}$, $m^{\prime}_{i}e$, $\omega_{i}$, and $m^{\prime}_{i}\omega$ are imposed in Assumption~\ref{as:spec}(iii). This establishes Assumption \ref{as:gen}.

\subsubsection*{Step 2. Orthogonality of the drift}

We next verify the orthogonality condition in Theorem \ref{thm:alt}. Under the alternative, by the decomposition in Section \ref{sub:spec}, $\Delta=B\alpha_{*}-X\gamma_{*}$, which lies in $\operatorname{col}(B)$ by Assumption~\ref{as:spec}(ii). As $M=I-P$ projects onto the orthogonal complement of $\operatorname{col}(B)$, we have $M\Delta=0$, hence $m^{\prime}_{i}\Delta=0$ for every $i$, and the systematic component is
\[
\pi_{i}=\Delta_{i}+\omega_{i}=\theta_{0}(x_{i})-\gamma^{\prime}_{*}x_{i},\qquad\lambda_{i}=m^{\prime}_{i}\omega.
\]

\subsubsection*{Step 3. The null hypothesis}

First we check condition \eqref{eq:cond}. Under $H_{0}$, $\Delta=0$ and $\hat{u}_{i}=e_{i}+\rho_{i}$. Letting $B_{i}=e_{i}(m^{\prime}_{i}e)$, it holds $\mathbb{E}[B_{i}B_{j}|X]=\mathbb{E}[e_{i}(m^{\prime}_{i}e)e_{j}(m^{\prime}_{j}e)|X]$ for $i\neq j$. Since 
\[
m^{\prime}_{i}e=M_{ii}e_{i}+\sum_{k\neq i}M_{ki}e_{k},\qquad m^{\prime}_{j}e=M_{jj}e_{j}+\sum_{\ell\neq j}M_{\ell j}e_{\ell},
\]
independence and mean-zero errors imply $\mathbb{E}[B_{i}B_{j}|X]=\sigma^{2}_{i}\sigma^{2}_{j}(M_{ii}M_{jj}+M^{2}_{ij})$, and thus 
\[
\mathbb{E}\left[\left.\sum^{n}_{i=1}\sum_{j\neq i}c_{ij}B_{i}B_{j}\right|X\right]=\frac{2}{K}\sum^{n}_{i=1}\sum_{j\neq i}P^{2}_{ij}\sigma^{2}_{i}\sigma^{2}_{j}=V.
\]
Let $S_{n}=\sum^{n}_{i=1}\sum_{j\neq i}c_{ij}\{B_{i}B_{j}-\mathbb{E}[B_{i}B_{j}|X]\}$. Since $V\ge2c^{2}A_{n}=2c^{2}+o_{p}(1)$, it suffices to show $S_{n}=o_{p}(1)$. Writing $\tilde{P}^{2}_{ij}=P^{2}_{ij}/(M_{ii}M_{jj}+M^{2}_{ij})$, we have $c_{ij}=\frac{2}{K}\tilde{P}^{2}_{ij}$ and $2a^{2}_{ij}=\frac{2}{K}P^{2}_{ij}$, so that
\[
\sum^{n}_{i=1}\sum_{j\neq i}\bigl\{ c_{ij}B_{i}B_{j}-2a^{2}_{ij}\sigma^{2}_{i}\sigma^{2}_{j}\bigr\}=\frac{2}{K}\sum^{n}_{i=1}\sum_{j\neq i}\bigl\{\tilde{P}^{2}_{ij}e_{i}(m^{\prime}_{i}e)e_{j}(m^{\prime}_{j}e)-P^{2}_{ij}\sigma^{2}_{i}\sigma^{2}_{j}\bigr\}.
\]
Apart from the factor $2/K$, this is the expression of \citet[Lemma 2]{MikushevaSun2022} with $P$ the sieve projection here, $M=I-P$ and $U^{(1)}_{i}=U^{(2)}_{i}=U^{(3)}_{i}=e_{i}$; their weight is $\tilde{P}^{2}_{ij}$ and their normalization is $K^{-1}$. Their conditions hold here: $K\to\infty$ by (ii), $P_{ii}\le c_{P}<1$ for some $c_{P}$ because $\max_{i}P_{ii}=o_{p}(1)$ as shown above, and the errors are independent with mean zero and bounded sixth conditional moments by (i). The lemma therefore gives $S_{n}=o_{p}(1)$.

To see where the sixth moments of Assumption \ref{as:spec}(i) enter, substitute $e_{i}(m^{\prime}_{i}e)=\sum_{k}M_{ik}e_{i}e_{k}$, which makes $S_{n}$ a centered quartic form in $e$,
\[
S_{n}=\sum^{n}_{i=1}\sum_{j\neq i}\sum_{k,l}\Gamma_{ijkl}\bigl\{ e_{i}e_{j}e_{k}e_{l}-\mathbb{E}[e_{i}e_{j}e_{k}e_{l}|X]\bigr\},\qquad\Gamma_{ijkl}=c_{ij}M_{ik}M_{jl}.
\]
Since $i\neq j$, no index occurs more than three times in a quartet, and hence no more than six times in $\mathbb{E}[S^{2}_{n}|X]$; six is attained by the terms with $i^{\prime}=i$ and $j^{\prime}=j$, so no moment above order six enters. Since $V\ge2c^{2}A_{n}$, the conclusion $S_{n}=o_{p}(1)$ of the lemma gives 
\[
\sum^{n}_{i=1}\sum_{j\neq i}\left\{ c_{ij}e_{i}(m_{i}'e)e_{j}(m_{j}'e)-2a^{2}_{ij}\sigma^{2}_{i}\sigma^{2}_{j}\right\} =o_{p}(V).
\]
This is \eqref{eq:cond}.

Next we check condition \eqref{eq:condest}. Under $H_{0}$ we have $\theta_{0}(x)=\gamma^{\prime}_{0}x$ by (i), so $\mathbb{E}[b(x)y]=\mathbb{E}[b(x)x^{\prime}]\gamma_{0}$; the representation $b(x)^{\prime}\alpha_{K}=x^{\prime}$ of (ii) turns the right-hand side into $\mathbb{E}[b(x)b(x)^{\prime}]\alpha_{K}\gamma_{0}$, so that $\alpha_{*}=\alpha_{K}\gamma_{0}$ and $\alpha^{\prime}_{*}b(x)=\gamma^{\prime}_{0}x$. With $\hat{\gamma}\overset{p}{\to}\gamma_{0}$ this gives $\gamma_{*}=\gamma_{0}$, hence $\Delta_{i}=\omega_{i}=0$, so $\check{u}_{i}=e_{i}$ in the notation of Lemma \ref{lem:exp} and $\ell^{\prime}_{i}H=m^{\prime}_{i}X=0$ by (ii). Then $g_{2i}=0$ and $\mathcal{R}_{n}=0$, while $g_{0i}=e_{i}(m^{\prime}_{i}e)$ and $g_{1i}=(m^{\prime}_{i}e)x_{i}$, so that, $c_{ij}$ being symmetric,
\[
\mathcal{G}_{1}=2\sum^{n}_{i=1}\sum_{j\neq i}c_{ij}e_{i}(m^{\prime}_{i}e)x_{j}(m^{\prime}_{j}e),\qquad\mathcal{G}_{2}=\sum^{n}_{i=1}\sum_{j\neq i}c_{ij}x_{i}x_{j}(m^{\prime}_{i}e)(m^{\prime}_{j}e).
\]
Since $\|\hat{\delta}\|=O_{p}(n^{-1/2})$ and $V\ge2c^{2}A_{n}=2c^{2}+o_{p}(1)$, it suffices to show $|\mathcal{G}_{1}|=o_{p}(n^{1/2})$ and $|\mathcal{G}_{2}|=o_{p}(n)$.

To see $|\mathcal{G}_{1}|=o_{p}(n^{1/2})$, let $\bar{c}_{i}=\sum_{j\neq i}|c_{ij}|$, so that $\sum_{i}\bar{c}_{i}=C_{n}$. Applying the Cauchy--Schwarz inequality over the pairs $i\neq j$ gives
\begin{equation}
\Bigl|\sum^{n}_{i=1}\sum_{j\neq i}c_{ij}\alpha_{i}\beta_{j}\Bigr|=\Bigl|\sum^{n}_{i=1}\sum_{j\neq i}(|c_{ij}|^{1/2}\alpha_{i})(|c_{ij}|^{1/2}\beta_{j})\operatorname{sgn}(c_{ij})\Bigr|\le\Bigl(\sum^{n}_{i=1}\bar{c}_{i}\alpha^{2}_{i}\Bigr)^{1/2}\Bigl(\sum^{n}_{j=1}\bar{c}_{j}\beta^{2}_{j}\Bigr)^{1/2}\label{eq:wcs}
\end{equation}
for any $\alpha$ and $\beta$. Take $\alpha_{i}=e_{i}(m^{\prime}_{i}e)$ and $\beta_{j}=x_{j}(m^{\prime}_{j}e)$. Since $m^{\prime}_{i}e=M_{ii}e_{i}+\sum_{k\neq i}M_{ik}e_{k}$ and the errors are independent with mean zero,
\[
\mathbb{E}[\alpha^{2}_{i}|X]=M^{2}_{ii}\mathbb{E}[e^{4}_{i}|X]+\sigma^{2}_{i}\sum_{k\neq i}M^{2}_{ik}\sigma^{2}_{k}\le C,\qquad\mathbb{E}[(m^{\prime}_{j}e)^{2}|X]=\sum_{k}M^{2}_{jk}\sigma^{2}_{k}\le CM_{jj}\le C,
\]
where we used $\sum_{k}M^{2}_{ik}=M_{ii}\le1$ and Assumption \ref{as:spec} (i). Hence $\mathbb{E}[\sum_{i}\bar{c}_{i}\alpha^{2}_{i}|X]\le CC_{n}$ and $\mathbb{E}[\sum_{j}\bar{c}_{j}\beta^{2}_{j}|X]\le C\max_{j}x^{2}_{j}\,C_{n}$. By (i) the regressors have six moments, so $\max_{j}x^{2}_{j}=O_{p}(n^{1/3})$, and $C_{n}=O_{p}(1)$ was shown above; Markov's inequality therefore gives
\[
\sum^{n}_{i=1}\bar{c}_{i}\alpha^{2}_{i}=O_{p}(1),\qquad\sum^{n}_{j=1}\bar{c}_{j}\beta^{2}_{j}=O_{p}(n^{1/3}).
\]
Consequently $\|\mathcal{G}_{1}\|=O_{p}(n^{1/6})=o_{p}(n^{1/2})$. Using the same argument with $\alpha=\beta$ gives $\|\mathcal{G}_{2}\|=O_{p}(n^{1/3})=o_{p}(n)$. Since $g_{2i}=0$ the remaining two forms vanish, $\mathcal{G}_{3}=\mathcal{G}_{4}=0$, so all four conditions in \eqref{eq:condest} hold.

\subsubsection*{Step 4. The alternative}

We turn to Theorem \ref{thm:alt}. We first check the hypotheses of Lemma \ref{lem:lln}, which give the first two conditions of \eqref{eq:condalt} and $B^{\mathrm{cf}}_{n}=o(V)$, then the two estimation-error conditions, then the last condition of the theorem. Here $L=M$ is symmetric and idempotent, so $\|L\|\le1$; the sixth moments are Assumption~\ref{as:spec}(i); \eqref{eq:cond} was verified above; and $\xi_{n}=O_{p}(1)$ because Assumption~\ref{as:spec}(iv) bounds $\max_{i}|\pi_{i}|$ while $\max_{i}|\lambda_{i}|=\max_{i}|m^{\prime}_{i}\omega|=o_{p}(1)$ by (iii).

For \eqref{eq:condagg}, recall $a^{2}_{ij}=P^{2}_{ij}/K$ and $c_{ij}=(2/K)\tilde{P}^{2}_{ij}$ with $M_{ii}M_{jj}+M^{2}_{ij}\ge c^{2}_{M}$. Since $\sum^{n}_{i=1}\sum_{j\neq i}P^{4}_{ij}=o_{p}(K)$ and $A_{n}=1+o_{p}(1)$,
\[
\mathcal{C}_{2}\le\frac{C}{K^{2}}\sum^{n}_{i=1}\sum_{j\neq i}P^{4}_{ij}=o_{p}(K^{-1})=o_{p}(A^{2}_{n}),
\]
and, because $\sum_{i}P^{2}_{ij}=P_{jj}\le\max_{i}P_{ii}=o_{p}(1)$,
\[
\bar{c}_{j}\le\frac{C}{K}\sum_{i}P^{2}_{ij}=o_{p}(K^{-1}),\qquad\text{so}\qquad\max_{j}\bar{c}_{j}=o_{p}(A_{n}).
\]
For the two Frobenius quantities, $M_{ij}=-P_{ij}$ for $i\neq j$ gives $L^{2}_{F}=\sum^{n}_{i=1}\sum_{j\neq i}P^{2}_{ij}\le K$. Writing $M\Sigma M=\Sigma-P\Sigma-\Sigma P+P\Sigma P$, the off-diagonal entries of $M\Sigma M$ are $-(\sigma^{2}_{i}+\sigma^{2}_{j})P_{ij}+(P\Sigma P)_{ij}$, and 
\[
\|P\Sigma P\|^{2}_{F}=\operatorname{tr}(P\Sigma PP\Sigma P)\le C^{2}\operatorname{tr}(P)=C^{2}K,
\]
so that $S_{F}\le C\sqrt{\operatorname{tr}(P)}=C\sqrt{K}$. Hence
\[
\mathcal{C}^{1/2}_{2}(L_{F}+S_{F})=o_{p}(K^{-1/2})O_{p}(K^{1/2})=o_{p}(1)=o_{p}(A_{n}),
\]
and Lemma \ref{lem:lln} delivers the first two conditions of \eqref{eq:condalt} together with $B^{\mathrm{cf}}_{n}=o(V)$.

It remains to bound the estimation error. For the cross-fit estimator, $m^{\prime}_{i}X=0$ by (ii), so Lemma \ref{lem:exp} again gives $g_{2i}=0$ and $\mathcal{R}_{n}=0$, while now $g_{0i}=\check{A}_{i}$ and $g_{1i}=(m^{\prime}_{i}\check{u})x_{i}$ with $m^{\prime}_{i}\check{u}=m^{\prime}_{i}e+\lambda_{i}$. Relative to the null the only changes are that $m^{\prime}_{i}e$ acquires the shift $\lambda_{i}$ and that $e_{i}(m^{\prime}_{i}e)$ becomes $\check{A}_{i}$. Both factors keep the uniform second-moment bounds that \eqref{eq:wcs} needs: $\mathbb{E}[\check{A}^{2}_{i}|X]\le C(1+\pi^{2}_{i})(1+\lambda^{2}_{i})\le C$ by Assumption \ref{as:spec}(iv), and $\mathbb{E}[(m^{\prime}_{i}\check{u})^{2}|X]\le2\lambda^{2}_{i}+2\mathbb{E}[(m^{\prime}_{i}e)^{2}|X]\le C$ because $\max_{i}|\lambda_{i}|=o_{p}(1)$. The computation used under the null therefore applies unchanged and gives $\|\mathcal{G}_{1}\|=O_{p}(n^{1/6})$ and $\|\mathcal{G}_{2}\|=O_{p}(n^{1/3})$. With $\|\hat{\delta}\|=O_{p}(n^{-1/2})$ and $V\ge2c^{2}A_{n}$, this yields $\hat{V}_{\mathrm{cf}}-\check{V}_{\mathrm{cf}}=o_{p}(V)$.

For the plug-in estimator, \eqref{eq:Vhat} is \eqref{eq:Vcf} with $\ell_{i}$ the $i$th unit vector and $2a^{2}_{ij}$ in place of $c_{ij}$, so Lemma \ref{lem:exp} applies with $L=I$ and gives $g_{0i}=\check{u}^{2}_{i}$, $g_{1i}=2\check{u}_{i}x_{i}$ and $g_{2i}=x_{i}x^{\prime}_{i}$; here the degree-three and degree-four terms do not vanish. The first two forms need different moment bounds, since $g_{0i}$ and $g_{1i}$ no longer carry $m^{\prime}_{i}\check{u}$. Writing $\bar{a}_{i}=\sum_{j\neq i}2a^{2}_{ij}$, so that $\sum_{i}\bar{a}_{i}=2A_{n}=O_{p}(1)$, three sums bound them:
\[
\sum_{i}\bar{a}_{i}\check{u}^{4}_{i}=O_{p}(1),\qquad\sum_{i}\bar{a}_{i}\check{u}^{2}_{i}x^{2}_{i}=O_{p}(n^{1/3}),\qquad\sum_{i}\bar{a}_{i}x^{4}_{i}=O_{p}(n^{2/3}),
\]
the first two because $\mathbb{E}[\check{u}^{4}_{i}|X]\le C(1+\pi^{4}_{i})\le C$ by (i) and Assumption \ref{as:spec}(iv), the last two because $\max_{i}x^{2}_{i}=O_{p}(n^{1/3})$. By \eqref{eq:wcs} the first two give $\|\mathcal{G}_{1}\|=O_{p}(n^{1/6})$, while the first and third bound the $g_{0i}g_{2j}$ part of $\mathcal{G}_{2}$ and the second its $g_{1i}g^{\prime}_{1j}$ part, so $\|\mathcal{G}_{2}\|=O_{p}(n^{1/3})$. Hence $\|\hat{\delta}\|\|\mathcal{G}_{1}\|+\|\hat{\delta}\|^{2}\|\mathcal{G}_{2}\|=O_{p}(n^{-1/3})$. For the other two, $\max_{i}\|x_{i}\|=O_{p}(n^{1/6})$ and $\max_{i}|\check{u}_{i}|=O_{p}(b_{n})$ by (i) and (iii), so $\mathcal{G}_{3}=O_{p}(b_{n}n^{1/2})$ and $\mathcal{G}_{4}=O_{p}(n^{2/3})$, whence $\|\hat{\delta}\|^{3}\mathcal{G}_{3}+\|\hat{\delta}\|^{4}\mathcal{G}_{4}=O_{p}(b_{n}n^{-1})+O_{p}(n^{-4/3})=o_{p}(1)$. Since $V+B_{n}\ge V\ge2c^{2}A_{n}$, this gives $\hat{V}-\check{V}=o_{p}(V+B_{n})$.

Finally, the last condition of Theorem \ref{thm:alt} is the non-degeneracy condition of Assumption \ref{as:spec}(iv), so $\liminf_{n}B_{n}/V>0$ and the plug-in estimator is inconsistent for $V$. All conditions of Theorems \ref{thm:null} and \ref{thm:alt} therefore hold, which proves Proposition \ref{prop:spec}.
\end{proof}

\section{Assumptions and proof of Proposition \ref{prop:over} }

\label{app:over}
\begin{assumption} \label{as:over}
The following conditions hold, in the notation of Section \ref{sub:over}.
\begin{description}
\item [{(i)}] The errors $\{(u_{i},v_{i})\}^{n}_{i=1}$ are independent across $i$ with $\mathbb{E}[u_{i}]=0$ and $\mathbb{E}[v_{i}]=0$, and satisfy $0<c\le\mathbb{E}[u^{2}_{i}]\le C<\infty$, $\max_{i}\mathbb{E}|u_{i}|^{6}\le C$ and $\max_{i}\mathbb{E}\|v_{i}\|^{6}\le C$. The reduced form is bounded, $\max_{i}\|\Upsilon_{i}\|\le C$. The error of the framework satisfies the same bounds: with $e_{i}=(y_{i}-\beta^{\prime}_{*}x_{i})-\bar{u}_{i}$,
\[
0<c\le\mathbb{E}[e^{2}_{i}]\le C<\infty,\qquad\max_{i}\mathbb{E}|e_{i}|^{6}\le C.
\]
\item [{(ii)}] The instruments $Z$ are nonrandom with $\operatorname{rank}(Z)=K$, the number of instruments satisfies $K\to\infty$ and $K/n\to\alpha\in[0,1)$, and there is a constant $c_{P}$ with $P_{ii}\le c_{P}<1$ for all $i$, and the reduced form is spanned by them, $\Upsilon=Z\Psi$.
\item [{(iii)}] The estimator is consistent for the pseudo-true value, $\hat{\beta}\overset{p}{\to}\beta_{*}$, and $\max_{i}|\omega_{i}|=o_{p}(b^{-3}_{n})$ for a deterministic sequence $1\le b_n=o(n)$ satisfying $\max_{i}|e_{i}|=O_{p}(b_{n})$ and $\max_{i}|m^{\prime}_{i}e|=O_{p}(b_{n})$.
\item [{(iv)}] Under the alternative the misspecification is bounded and non-degenerate, and the weights are not concentrated on a few pairs:
\[
\max_{i}|\bar{u}_{i}|\le C,\qquad\liminf_{n}A^{-1}_{n}\sum^{n}_{i=1}\sum_{j\neq i}P^{2}_{ij}\bar{u}^{2}_{i}>0,\qquad\sum^{n}_{i=1}\sum_{j\neq i}P^{4}_{ij}=o(K).
\]
\end{description}
\end{assumption}

Under the null $\beta_{*}=\beta$ and $e_{i}=u_{i}$, so the bounds (i) places on $e_{i}$ are vacuous; under the alternative $e_{i}=u_{i}+v^{\prime}_{i}(\beta-\beta_{*})$ and the bounds are a genuine restriction, since neither $\|\beta-\beta_{*}\|$ nor the lower bound on $\mathbb{E}[e^{2}_{i}]$ follows from the moments of $u_{i}$ and $v_{i}$ alone.

The first condition of (iv) is vacuous under the null, where $\bar{u}=0$. The second makes the plug-in estimator inconsistent as required in Theorem \ref{thm:alt}. The third is implied by $\max_{i\neq j}|P_{ij}|=o(1)$. 
\begin{proof}
For the overidentification-testing problem, 
\[
a_{ij}=P_{ij},\qquad c_{ij}=\frac{2P^{2}_{ij}}{M_{ii}M_{jj}+M^{2}_{ij}},\qquad\ell_{i}=m_{i},
\]
where $M=I-P$ and $m_{i}$ denotes the $i$-th column of $M$. We verify the conditions of Theorems \ref{thm:null} and \ref{thm:alt}.

\subsubsection*{Step 1. The maintained assumption}

We begin by verifying Assumption \ref{as:gen}. By (ii), $P_{ii}\le c_{P}<1$ gives $M_{ii}=1-P_{ii}\ge1-c_{P}$, hence $M_{ii}M_{jj}+M^{2}_{ij}\ge(1-c_{P})^{2}$ uniformly over $i\neq j$. Because $P$ is idempotent, 
\[
\sum^{n}_{i=1}\sum_{j\neq i}P^{2}_{ij}=\mathrm{tr}(P^{2})-\sum_{i}P^{2}_{ii}=K-\sum_{i}P^{2}_{ii}.
\]
Moreover $\sum_{i}P^{2}_{ii}\le c_{P}\sum_{i}P_{ii}=c_{P}K$ by (ii), so that
\[
(1-c_{P})K\le\sum^{n}_{i=1}\sum_{j\neq i}P^{2}_{ij}\le K,\qquad A_{n}=\sum^{n}_{i=1}\sum_{j\neq i}a^{2}_{ij}=\sum^{n}_{i=1}\sum_{j\neq i}P^{2}_{ij}.
\]
Moreover, 
\[
2c^{2}A_{n}\le V=2\sum^{n}_{i=1}\sum_{j\neq i}a^{2}_{ij}\mathbb{E}[e^{2}_{i}]\mathbb{E}[e^{2}_{j}]\le2C^{2}A_{n},
\]
because the variances are bounded above and away from zero. The leverage condition of Assumption \ref{as:gen} holds with $c_{L}=(1-c_{P})^{2}$, and
\[
C_{n}=\sum^{n}_{i=1}\sum_{j\neq i}|c_{ij}|\le\frac{2}{(1-c_{P})^{2}}A_{n}.
\]
Thus $C_{n}=O(A_{n})$. The maximal bounds for $e_{i}$, $m^{\prime}_{i}e$, $\omega_{i}$, and $m^{\prime}_{i}\omega$ are imposed in Assumption \ref{as:over}(iii), the one for $m^{\prime}_{i}\omega$ through $m^{\prime}_{i}\omega=\omega_{i}$, which holds because $\omega=M\bar{u}$ and $M$ is idempotent. This establishes Assumption \ref{as:gen}.

\subsubsection*{Step 2. Orthogonality of the drift}

We next verify the orthogonality condition appearing in Theorem \ref{thm:alt}. By the decomposition in Section \ref{sub:over}, $\Delta=P\bar{u}$. Hence $\Delta\in\operatorname{col}(Z)$. Since $P$ is the projection matrix onto $\operatorname{col}(Z)$ and $M=I-P$, $M\Delta=(I-P)\Delta=0$. Therefore, $m^{\prime}_{i}\Delta=0$ for every $i$, and consequently $\max_{i}|m^{\prime}_{i}\Delta|=0=o_{p}(b^{-3}_{n})$.

\subsubsection*{Step 3. The null hypothesis}

First we check condition \eqref{eq:cond} for the error $e_{i}$ of the framework, as Theorems \ref{thm:null} and \ref{thm:alt} both require. Under $H_{0}$ the departure vanishes, $\bar{u}=0$, so $\Delta=0$, $\omega=0$ and $e_{i}=u_{i}$: the null is the special case $e_{i}=u_{i}$ of what follows. Writing $\tilde{P}^{2}_{ij}=P^{2}_{ij}/(M_{ii}M_{jj}+M^{2}_{ij})$, we have $c_{ij}=2\tilde{P}^{2}_{ij}$ and $2a^{2}_{ij}=2P^{2}_{ij}$, so that
\[
\sum^{n}_{i=1}\sum_{j\neq i}\bigl\{ c_{ij}e_{i}(m^{\prime}_{i}e)e_{j}(m^{\prime}_{j}e)-2a^{2}_{ij}\mathbb{E}[e^{2}_{i}]\mathbb{E}[e^{2}_{j}]\bigr\}=2K\cdot\frac{1}{K}\sum^{n}_{i=1}\sum_{j\neq i}\bigl\{\tilde{P}^{2}_{ij}e_{i}(m^{\prime}_{i}e)e_{j}(m^{\prime}_{j}e)-P^{2}_{ij}\mathbb{E}[e^{2}_{i}]\mathbb{E}[e^{2}_{j}]\bigr\}.
\]
The sum on the right is the expression of \citet[Lemma 2]{MikushevaSun2022} with $P$ the projection onto $\operatorname{col}(Z)$, $M=I-P$ and $U^{(1)}_{i}=U^{(2)}_{i}=U^{(3)}_{i}=e_{i}$. Their conditions hold by Assumption \ref{as:over}: $K\to\infty$ and $P_{ii}\le c_{P}<1$ by (ii), and by (i) the $e_{i}$ are independent across $i$ with mean zero and $\max_{i}\mathbb{E}|e_{i}|^{6}\le C$. This is where the bounds (i) places on $e_{i}$ itself are used. Hence the right-hand side is $2Ko_{p}(1)=o_{p}(K)$, which is $o_{p}(V)$ because $V\ge2c^{2}A_{n}\ge2c^{2}(1-c_{P})K$. This is \eqref{eq:cond}.

Next we check condition \eqref{eq:condest}. Recall from (ii) that $\Upsilon=Z\Psi$, so $M\Upsilon=0$ and $h_{i}=\ell^{\prime}_{i}H=m^{\prime}_{i}x=m^{\prime}_{i}v$. We again give the argument for $G=1$. Under $H_{0}$
\[
g_{0i}=u_{i}(m^{\prime}_{i}u),\qquad g_{1i}=u_{i}(m^{\prime}_{i}v)+(m^{\prime}_{i}u)x_{i},\qquad g_{2i}=x_{i}(m^{\prime}_{i}v).
\]
Since $m^{\prime}_{i}u=\sum_{k}M_{ik}u_{k}$ with $\sum_{k}M^{2}_{ik}=M_{ii}\le1$, independence and the fourth moments of (i) give $\mathbb{E}[(m^{\prime}_{i}u)^{2}]\le C$ and $\mathbb{E}[(m^{\prime}_{i}u)^{4}]\le C$, and likewise for $m^{\prime}_{i}v$. With $x_{i}=\Upsilon_{i}+v_{i}$ and $\max_{i}\|\Upsilon_{i}\|\le C$ the Cauchy--Schwarz inequality then gives $\max_{k}\max_{i}\mathbb{E}\|g_{ki}\|^{2}\le C$.

With $\bar{c}_{i}=\sum_{j\neq i}|c_{ij}|$, so that $\sum_{i}\bar{c}_{i}=C_{n}$, and $c_{ij}=c_{ji}$ here, we can apply \eqref{eq:wcs}. Since $\mathbb{E}\bigl[\sum^{n}_{i=1}\bar{c}_{i}\|g_{ki}\|^{2}\bigr]\le CC_{n}$, Markov's inequality gives $\sum^{n}_{i=1}\bar{c}_{i}\|g_{ki}\|^{2}=O_{p}(C_{n})$ for $k=0,1,2$. By the same symmetry, the two terms of $\mathcal{G}_{1}$ coincide and so do the last two of $\mathcal{G}_{2}$, so that
\[
\|\mathcal{G}_{1}\|\le2\sum^{n}_{i=1}\sum_{j\neq i}|c_{ij}|\,|g_{0i}|\,\|g_{1j}\|,\qquad\|\mathcal{G}_{2}\|\le\sum^{n}_{i=1}\sum_{j\neq i}|c_{ij}|\,\|g_{1i}\|\,\|g_{1j}\|+2\sum^{n}_{i=1}\sum_{j\neq i}|c_{ij}|\,|g_{0i}|\,\|g_{2j}\|.
\]
For $\mathcal{G}_{3}$ and $\mathcal{G}_{4}$, with $\|\delta\|=1$ we have $|g^{\prime}_{1i}\delta|\le\|g_{1i}\|$ and $|\delta^{\prime}g_{2i}\delta|\le\|g_{2i}\|$, so the summand of $\mathcal{G}_{3}$ is at most $|c_{ij}|(\|g_{1i}\|\,\|g_{2j}\|+\|g_{1j}\|\,\|g_{2i}\|)$ and that of $\mathcal{G}_{4}$ at most $|c_{ij}|\,\|g_{2i}\|\,\|g_{2j}\|$. Summing these and using the symmetry of $c$ once more,
\[
\mathcal{G}_{3}\le2\sum^{n}_{i=1}\sum_{j\neq i}|c_{ij}|\,\|g_{1i}\|\,\|g_{2j}\|,\qquad\mathcal{G}_{4}\le\sum^{n}_{i=1}\sum_{j\neq i}|c_{ij}|\,\|g_{2i}\|\,\|g_{2j}\|.
\]
By \eqref{eq:wcs} each of the five double sums on the right is $O_{p}(C_{n})$. Since $C_{n}\le2A_{n}/(1-c_{P})^{2}\le2K/(1-c_{P})^{2}$, all $\mathcal{G}_{k}$ are $O_{p}(K)$ for $k=1,2,3,4$.

Consistency of $\hat{\beta}$ gives $\|\hat{\delta}\|=o_{p}(1)$, and $V\ge2c^{2}A_{n}\ge2c^{2}(1-c_{P})K$, so
\[
\|\hat{\delta}\|\|\mathcal{G}_{1}\|
+\|\hat{\delta}\|^{2}\|\mathcal{G}_{2}\|
+\|\hat{\delta}\|^{3}\mathcal{G}_{3}
+\|\hat{\delta}\|^{4}\mathcal{G}_{4}
=o_{p}(K)=o_{p}(V),
\]
which is \eqref{eq:condest}. All conditions of Theorem \ref{thm:null} are therefore satisfied, which proves the null part of Proposition \ref{prop:over}.

\subsubsection*{Step 4. The alternative}

We turn to Theorem \ref{thm:alt}. We first check the hypotheses of Lemma \ref{lem:lln}, which give the first two conditions of \eqref{eq:condalt} and $B^{\mathrm{cf}}_{n}=o(V)$, then the two estimation-error conditions, then the last condition of the theorem. Because $\Delta=P\bar{u}$ and $\omega=M\bar{u}$, the systematic component is the misspecification itself and the surviving part is its projection off the instruments,
\[
\pi_{i}=\Delta_{i}+\omega_{i}=\bar{u}_{i},\qquad\lambda_{i}=m^{\prime}_{i}\omega=(M\bar{u})_{i}=\omega_{i},
\]
the last equality because $M$ is idempotent. 

We check the hypotheses of Lemma \ref{lem:lln}. Again $L=M$ is symmetric idempotent, so $\|L\|\le1$; the sixth moments are Assumption \ref{as:over}(i); \eqref{eq:cond} was verified above; and $\xi_{n}=O_{p}(1)$ by Assumption \ref{as:over}(iv) together with $\max_{i}|\lambda_{i}|=\max_{i}|\omega_{i}|=o_{p}(1)$ from (iii). For \eqref{eq:condagg}, $c_{ij}=2\tilde{P}^{2}_{ij}$ with $M_{ii}M_{jj}+M^{2}_{ij}\ge(1-c_{P})^{2}$, so $\mathcal{C}_{2}\le C\sum^{n}_{i=1}\sum_{j\neq i}P^{4}_{ij}=o_{p}(K)$ by Assumption \ref{as:over}(iv), which is $o_{p}(A^{2}_{n})$ because $A_{n}\ge(1-c_{P})K\to\infty$. Next $\bar{c}_{j}\le C\sum_{i}P^{2}_{ij}=CP_{jj}\le Cc_{P}$, so $\max_{j}\bar{c}_{j}=O(1)=o(A_{n})$. Finally $L^{2}_{F}=\sum^{n}_{i=1}\sum_{j\neq i}P^{2}_{ij}\le K$ and, as in the proof of Proposition \ref{prop:spec}, $S_{F}\le C\sqrt{K}$; hence
\[
\mathcal{C}^{1/2}_{2}(L_{F}+S_{F})=o_{p}(K^{1/2})O(K^{1/2})=o_{p}(K)=o_{p}(A_{n}).
\]
Lemma \ref{lem:lln} therefore gives the first two conditions of \eqref{eq:condalt} and $B^{\mathrm{cf}}_{n}=o(V)$. Note where the conditions of Assumption \ref{as:over}(iv) act: the bound on $\bar{u}$ controls the size of the systematic component, and the summability of $P^{4}_{ij}$ stops a single pair of observations from carrying a non-negligible share of the weight.

The estimation-error conditions are those of the null argument with $u_{i}$ replaced by $\check{u}_{i}=\bar{u}_{i}+e_{i}$ and $m^{\prime}_{i}u$ by $m^{\prime}_{i}\check{u}=\lambda_{i}+m^{\prime}_{i}e$. Since $\xi_{n}=O_{p}(1)$, the moment bounds $\max_{k}\max_{i}\mathbb{E}\|g_{ki}\|^{2}\le C$ used there continue to hold, so
\[
\|\hat{\delta}\|\|\mathcal{G}_{1}\|
+\|\hat{\delta}\|^{2}\|\mathcal{G}_{2}\|
+\|\hat{\delta}\|^{3}\mathcal{G}_{3}
+\|\hat{\delta}\|^{4}\mathcal{G}_{4}
=o_{p}(K)=o_{p}(V);
\]
this gives
$\hat{V}_{\mathrm{cf}}-\check{V}_{\mathrm{cf}}=o_{p}(V)$. For the plug-in estimator, \eqref{eq:Vhat} is \eqref{eq:Vcf} with $\ell_{i}$ the $i$th unit vector and $2a^{2}_{ij}=2P^{2}_{ij}$ in place of $c_{ij}$, so Lemma \ref{lem:exp} applies with $L=I$ and gives $g_{0i}=\check{u}^{2}_{i}$, $g_{1i}=2\check{u}_{i}x_{i}$, $g_{2i}=x_{i}x^{\prime}_{i}$, the weights summing to $2A_{n}\le2K$. The same moment bound $\max_{k}\max_{i}\mathbb{E}\|g_{ki}\|^{2}\le C$ holds for these, using $\max_{i}\|\Upsilon_{i}\|\le C$ and the sixth moments of $v_{i}$, so
\[
\|\hat{\delta}\|\|\mathcal{G}_{1}\|
+\|\hat{\delta}\|^{2}\|\mathcal{G}_{2}\|
+\|\hat{\delta}\|^{3}\mathcal{G}_{3}
+\|\hat{\delta}\|^{4}\mathcal{G}_{4}
=o_{p}(K)=o_{p}(V),
\]
and hence
$\hat{V}-\check{V}=o_{p}(V)=o_{p}(V+B_{n})$.

Finally, Assumption \ref{as:over}(iv) requires $\liminf_{n}A^{-1}_{n}\sum^{n}_{i=1}\sum_{j\neq i}P^{2}_{ij}\bar{u}^{2}_{i}>0$, which makes the plug-in estimator inconsistent as required in Theorem \ref{thm:alt}. All conditions of Theorems \ref{thm:null} and \ref{thm:alt} therefore hold, which proves Proposition \ref{prop:over}.
\end{proof}

\section{Assumptions and proof of Proposition \ref{prop:F} } \label{app:F}

\begin{assumption} \label{as:F}
The following conditions hold, in the notation of Section \ref{sub:F}.
\begin{description}
\item [{(i)}] $\{(y_{i},w_{i},z_{i})\}^{n}_{i=1}$ are i.i.d. and 
\[
y_{i}=w^{\prime}_{i}\beta_{1}+z^{\prime}_{i}\beta_{2}+\varepsilon_{i},\qquad\mathbb{E}[\varepsilon_{i}|x_{i}]=0,
\]
where $x_{i}=(w^{\prime}_{i},z^{\prime}_{i})^{\prime}$. Also,
\[
0<c\le\sigma^{2}_{i}:=\mathbb{E}[\varepsilon^{2}_{i}|x_{i}]\le C<\infty,\qquad\max_{1\le i\le n}\mathbb{E}[\varepsilon^{6}_{i}|x_{i}]\le C,
\]
uniformly in $i$, for some constants $0<c<C<\infty$. In addition, $\mathbb{E}\|w_{i}\|^{2}<\infty$ and $\mathbb{E}[w_{i}w^{\prime}_{i}]$ is nonsingular.
\item [{(ii)}] The number of restrictions satisfies $r\to\infty$, and there is a constant $c_{M}>0$ such that $\min_{i}M_{ii}\ge c_{M}$ with probability approaching one.
\item [{(iii)}] There exists a deterministic sequence $1\le b_n=o(n)$ such that $\max_{i}|\varepsilon_{i}|=O_{p}(b_{n})$ and $\max_{i}|m^{\prime}_{i}\varepsilon|=O_{p}(b_{n})$, where $m_{i}$ denotes the $i$-th column of $M$. 
\item [{(iv)}] Under the alternative the drift is bounded and non-degenerate, and the weights are not concentrated on a few pairs:
\[
\max_{i}|\Delta_{i}|=O_{p}(1),\qquad\liminf_{n}A^{-1}_{n}\sum^{n}_{i=1}\sum_{j\neq i}Q^{2}_{ij}\Delta^{2}_{i}>0,\qquad\sum^{n}_{i=1}\sum_{j\neq i}Q^{4}_{ij}=o_{p}(r).
\]
\end{description}
\end{assumption}

The condition of \citet{AnatolyevSolvsten2023} is the fourth-moment bound $\max_{i}\mathbb{E}[\varepsilon^{4}_{i}|x_{i}]+\sigma^{-2}_{i}=O_{p}(1)$; \citet[Lemma 2]{MikushevaSun2022} requires the sixth moment, and we state the bounds with deterministic constants so that they supply Assumption \ref{as:gen} directly, as in Assumptions \ref{as:spec} and \ref{as:over} above. For the same reason (ii) fixes the constant $c_{M}$. Their $\max_{i}M^{-1}_{ii}=O_{p}(1)$ gives a bound that may depend on how small a failure probability one insists on, whereas here $c_{M}$ has to be a single constant: it becomes the $c_{L}=c^{2}_{M}$ of Assumption \ref{as:gen}, and the proof carries it on through $c_{P}=1-c_{M}$ and $A_{n}\ge(1-c_{P})r$. The two conditions on $w_{i}$ are not part of their Assumption 2.

We place no direct restriction on $r/n$: as the proof below shows, the condition in (ii) already delivers $\limsup_{n}r/n<1$, the leverage bound $(P_{X})_{ii}\le c_{P}<1$ that \citet[Lemma 2]{MikushevaSun2022} requires of the projection to which it is applied, together with the bound $Q_{ii}\le c_{P}$ used for $\sum_{i}Q^{2}_{ii}\le c_{P}r$, and two-sided bounds on $A_{n}$. \citet{AnatolyevSolvsten2023} allow $r$ to remain fixed; we retain $r\to\infty$, which enters through \citet[Assumption 1]{MikushevaSun2022} and lets us avoid the non-standard limit that a fixed number of restrictions produces.

The first condition of (iv) is vacuous under the null, where $\Delta=0$. The second makes the plug-in estimator inconsistent as required in Theorem \ref{thm:alt}. Note that the first condition is stated for $\Delta=M_{W}Z(\beta_{2}-\beta_{20})$ and does not follow from a bound on $z^{\prime}_{i}(\beta_{2}-\beta_{20})$ alone, since $(M_{W}\Delta_{0})_{i}=\Delta_{0i}-w^{\prime}_{i}(W^{\prime}W)^{-1}W^{\prime}\Delta_{0}$ and (i) bounds $\max_{i}\|w_{i}\|$ only by $o_{p}(n^{1/2})$. The third is implied by $\max_{i\neq j}|Q_{ij}|=o_{p}(1)$.
\begin{proof}
For the many restrictions testing problem, 
\[
a_{ij}=Q_{ij},\qquad c_{ij}=2\frac{Q^{2}_{ij}}{M_{ii}M_{jj}+M^{2}_{ij}},\qquad\ell_{i}=m_{i},
\]
where $m_{i}$ denotes the $i$-th column of $M$. We verify the conditions of Theorems \ref{thm:null} and \ref{thm:alt}. All expectations below are conditional on $X=[W,Z]$.

\subsubsection*{Step 1. The maintained assumption}

We begin by verifying Assumption \ref{as:gen}. By Assumption \ref{as:F}(ii), $\min_{i}M_{ii}\ge c_{M}>0$ with probability approaching one, so that $M_{ii}M_{jj}+M^{2}_{ij}\ge c^{2}_{M}$ uniformly over $i\neq j$, the leverage condition of Assumption~\ref{as:gen} with $c_{L}=c^{2}_{M}$. The same bound on $\min_{i}M_{ii}$ carries the remaining leverage and order calculations of this step. It first bounds the leverages of $Q$, which supplies the weights $a_{ij}=Q_{ij}$: by the Frisch--Waugh--Lovell decomposition, $P_{X}=P_{W}+Q$ with $P_{W}=W(W^{\prime}W)^{-1}W^{\prime}$ positive semidefinite, so that $(P_{W})_{ii}\ge0$ and
\[
Q_{ii}\le Q_{ii}+(P_{W})_{ii}=(P_{X})_{ii}=1-M_{ii}\le1-c_{M}=:c_{P}<1.
\]
Averaging $1-M_{ii}\le1-c_{M}$ over $i$ and using $\sum_{i}(1-M_{ii})=d$ gives $d/n\le1-c_{M}$; since $r\le d$ this makes $\limsup_{n}r/n\le1-c_{M}<1$. Finally $Q$ is symmetric idempotent of rank $r$, so $\sum^{n}_{i=1}\sum_{j\neq i}Q^{2}_{ij}=\operatorname{tr}(Q^{2})-\sum_{i}Q^{2}_{ii}=r-\sum_{i}Q^{2}_{ii}$. Moreover $\sum_{i}Q^{2}_{ii}\le c_{P}\sum_{i}Q_{ii}=c_{P}r$, so that
\[
(1-c_{P})r\le A_{n}=\sum^{n}_{i=1}\sum_{j\neq i}a^{2}_{ij}=\sum^{n}_{i=1}\sum_{j\neq i}Q^{2}_{ij}\le r.
\]
Since $0<c\le\sigma^{2}_{i}\le C<\infty$, we obtain 
\[
2c^{2}A_{n}\le V=2\sum^{n}_{i=1}\sum_{j\neq i}a^{2}_{ij}\sigma^{2}_{i}\sigma^{2}_{j}\le2C^{2}A_{n}.
\]
Furthermore, 
\[
C_{n}=\sum^{n}_{i=1}\sum_{j\neq i}|c_{ij}|\le\frac{2}{c^{2}_{M}}\sum^{n}_{i=1}\sum_{j\neq i}Q^{2}_{ij}=\frac{2}{c^{2}_{M}}A_{n}.
\]
Assumption \ref{as:F}(iii) supplies the maximal bounds for $\varepsilon_{i}$ and $m^{\prime}_{i}\varepsilon$. The approximation error vanishes here, $\omega=0$, so the conditions Assumption \ref{as:gen} places on $\omega$ hold trivially. Write $\hat{\delta}=-(n^{-1}W^{\prime}W)^{-1}(n^{-1}W^{\prime}\varepsilon)$ and bound the two factors. By (i) the matrix $\mathbb{E}[w_{i}w^{\prime}_{i}]$ is nonsingular, so $\lambda_{\min}(n^{-1}W^{\prime}W)$ is bounded away from zero with probability approaching one; and since $\mathbb{E}[\varepsilon_{i}\mid x_{i}]=0$ kills the cross terms, $\mathbb{E}\|n^{-1}W^{\prime}\varepsilon\|^{2}=n^{-2}\sum_{i}\mathbb{E}\|w_{i}\varepsilon_{i}\|^{2}\le Cn^{-1}\mathbb{E}\|w_{i}\|^{2}$, which is $O(n^{-1})$ by (i). Hence $\|\hat{\delta}\|=O_{p}(n^{-1/2})$ under both hypotheses. Moreover $m^{\prime}_{i}\rho=m^{\prime}_{i}W\hat{\delta}=0$ identically because $MW=0$. This establishes Assumption \ref{as:gen}.

\subsubsection*{Step 2. Orthogonality of the drift}

We next verify the orthogonality condition in Theorem \ref{thm:alt}. Under the alternative, $\Delta=M_{W}Z(\beta_{2}-\beta_{20})$. Since $Z(\beta_{2}-\beta_{20})\in\operatorname{col}(Z)\subseteq\operatorname{col}(X)$ and $MW=0$ gives $MP_{W}=0$, it follows that $M\Delta=0$. Therefore $m^{\prime}_{i}\Delta=0$ for every $i$, and hence $\max_{i}|m^{\prime}_{i}\Delta|=0=o_{p}(b^{-3}_{n})$.

\subsubsection*{Step 3. The null hypothesis}

First we check condition \eqref{eq:cond}. Under $H_{0}$ we have $\beta_{2}=\beta_{20}$, so $\Delta=0$, and $\omega=0$, so $\check{u}_{i}=\varepsilon_{i}$ in the notation of Lemma \ref{lem:exp}. With $\tilde{P}^{2}_{ij}=Q^{2}_{ij}/(M_{ii}M_{jj}+M^{2}_{ij})$ we have $c_{ij}=2\tilde{P}^{2}_{ij}$ and $2a^{2}_{ij}=2Q^{2}_{ij}$, so
\begin{equation}
\sum^{n}_{i=1}\sum_{j\neq i}\bigl\{ c_{ij}\varepsilon_{i}(m^{\prime}_{i}\varepsilon)\varepsilon_{j}(m^{\prime}_{j}\varepsilon)-2a^{2}_{ij}\sigma^{2}_{i}\sigma^{2}_{j}\bigr\}=2\sum^{n}_{i=1}\sum_{j\neq i}\bigl\{\tilde{P}^{2}_{ij}\varepsilon_{i}(m^{\prime}_{i}\varepsilon)\varepsilon_{j}(m^{\prime}_{j}\varepsilon)-Q^{2}_{ij}\sigma^{2}_{i}\sigma^{2}_{j}\bigr\}.\label{eq:MS_target}
\end{equation}
\citet[Lemma 2]{MikushevaSun2022} is stated for a projection and its own complement, and $Q$ is not the complement of $M$: here $M=I-P_{X}$ with $P_{X}=P_{W}+Q$, so the lemma does not apply to the pair $(Q,M)$ directly. It does apply to the pair $(P_{X},M)$, and we reduce to that case by splitting the weight. Rearranging the decomposition of Step 1 as $Q=P_{X}-P_{W}$ gives
\[
Q^{2}_{ij}=(P_{X})^{2}_{ij}-2(P_{X})_{ij}(P_{W})_{ij}+(P_{W})^{2}_{ij}.
\]
Substituting this into the right-hand side of \eqref{eq:MS_target} yields three pieces
\[
2\sum^{n}_{i=1}\sum_{j\neq i}\Bigl\{\frac{v_{ij}}{M_{ii}M_{jj}+M^{2}_{ij}}\varepsilon_{i}(m^{\prime}_{i}\varepsilon)\varepsilon_{j}(m^{\prime}_{j}\varepsilon)-v_{ij}\sigma^{2}_{i}\sigma^{2}_{j}\Bigr\},\quad\text{for }v_{ij}\in\left\{ (P_{X})^{2}_{ij},-2(P_{X})_{ij}(P_{W})_{ij},(P_{W})^{2}_{ij}\right\} .
\]
For $v_{ij}=(P_{X})^{2}_{ij}$, it is $2d$ times the expression of \citet[Lemma 2]{MikushevaSun2022} with $P=P_{X}$ and $U^{(1)}_{i}=U^{(2)}_{i}=U^{(3)}_{i}=\varepsilon_{i}$. Their conditions hold: $P_{X}$ is a projection of rank $d\to\infty$, its complement is $M$, the leverage bound $(P_{X})_{ii}=1-M_{ii}\le1-c_{M}<1$ holds by (ii), and the errors are independent with mean zero and six bounded conditional moments by (i). Hence it is $2d\,o_{p}(1)=o_{p}(d)=o_{p}(r)$.

For $v_{ij}=-2(P_{X})_{ij}(P_{W})_{ij}$ and $(P_{W})^{2}_{ij}$, the pieces are negligible because $w_{i}$ has fixed dimension, so that $\operatorname{tr}(P_{W})=d-r$ stays bounded. For a symmetric idempotent matrix $A$ one has $\sum_{i}\sum_{j}A^{2}_{ij}=\operatorname{tr}(A)$, so by the Cauchy--Schwarz inequality
\begin{align*}
\sum^{n}_{i=1}\sum_{j\neq i}\bigl|(P_{X})_{ij}(P_{W})_{ij}\bigr| & \le\bigl\{\operatorname{tr}(P_{X})\operatorname{tr}(P_{W})\bigr\}^{1/2}=\bigl\{ d(d-r)\bigr\}^{1/2}=O(r^{1/2}),\\
\sum^{n}_{i=1}\sum_{j\neq i}(P_{W})^{2}_{ij} & \le d-r=O(1).
\end{align*}
Moreover, since $\sum_{k}M^{2}_{ik}=M_{ii}\le1$, it follows
\[
\mathbb{E}[\varepsilon^{2}_{i}(m^{\prime}_{i}\varepsilon)^{2}|X]=M^{2}_{ii}\mathbb{E}[\varepsilon^{4}_{i}|x_{i}]+\sigma^{2}_{i}\sum_{k\neq i}M^{2}_{ik}\sigma^{2}_{k}\le C.
\]
Thus applying the Cauchy--Schwarz inequality gives $\mathbb{E}[|\varepsilon_{i}(m^{\prime}_{i}\varepsilon)\varepsilon_{j}(m^{\prime}_{j}\varepsilon)|\,|X]\le C$ uniformly in $i\neq j$. Hence, for each of the two,
\[
\mathbb{E}\Bigl[\Bigl|2\sum^{n}_{i=1}\sum_{j\neq i}\Bigl\{\frac{v_{ij}}{M_{ii}M_{jj}+M^{2}_{ij}}\varepsilon_{i}(m^{\prime}_{i}\varepsilon)\varepsilon_{j}(m^{\prime}_{j}\varepsilon)-v_{ij}\sigma^{2}_{i}\sigma^{2}_{j}\Bigr\}\Bigr|\Bigm|X\Bigr]\le C\sum^{n}_{i=1}\sum_{j\neq i}|v_{ij}|,
\]
the first term inside the sum by $M_{ii}M_{jj}+M^{2}_{ij}\ge c^{2}_{M}$ and the moment bound just established, the second by $\sigma^{2}_{i}\le C$. The right-hand side is $O(r^{1/2})$ for $v_{ij}=-2(P_{X})_{ij}(P_{W})_{ij}$ and $O(1)$ for $v_{ij}=(P_{W})^{2}_{ij}$, and it is deterministic, so both pieces are $O_{p}(r^{1/2})=o_{p}(r)$ by Markov's inequality. Hence the left-hand side is $o_{p}(r)$, which is $o_{p}(V)$ because $V\ge2c^{2}A_{n}\ge2c^{2}(1-c_{P})r$. This is \eqref{eq:cond}.

Next we check condition \eqref{eq:condest}. The annihilation $\ell^{\prime}_{i}H=m^{\prime}_{i}W=0$ gives $g_{2i}=0$, so $\mathcal{G}_{3}=\mathcal{G}_{4}=\mathcal{R}_{n}=0$ and only $\mathcal{G}_{1}$ and $\mathcal{G}_{2}$ remain. Under $H_{0}$, $g_{0i}=\varepsilon_{i}(m^{\prime}_{i}\varepsilon)$ and $g_{1i}=(m^{\prime}_{i}\varepsilon)w_{i}$, exactly as in the proof of Proposition~\ref{prop:spec} with $w_{i}$ in place of $x_{i}$. The same bound \eqref{eq:wcs} applies: with $\bar{c}_{i}=\sum_{j\neq i}|c_{ij}|$ and $\alpha_{i}=\varepsilon_{i}(m^{\prime}_{i}\varepsilon)$, $\beta_{j}=w_{j}(m^{\prime}_{j}\varepsilon)$,
\[
\mathbb{E}\Bigl[\sum^{n}_{i=1}\bar{c}_{i}\alpha^{2}_{i}\Bigm|X\Bigr]\le CC_{n},\qquad\mathbb{E}\Bigl[\sum^{n}_{j=1}\bar{c}_{j}\|\beta_{j}\|^{2}\Bigm|X\Bigr]\le C\max_{j}\|w_{j}\|^{2}\,C_{n},
\]
so that $\|\mathcal{G}_{1}\|$ and $\|\mathcal{G}_{2}\|$ are $O_{p}(C_{n}\max_{j}\|w_{j}\|)$ and $O_{p}(C_{n}\max_{j}\|w_{j}\|^{2})$ respectively. Since $w_{i}$ is fixed-dimensional with bounded second moments, $\max_{j}\|w_{j}\|=o_{p}(n^{1/2})$, and $C_{n}=O_{p}(A_{n})=O_{p}(r)$, so with $\|\hat{\delta}\|=O_{p}(n^{-1/2})$
\[
\|\hat{\delta}\|\,\|\mathcal{G}_{1}\|=o_{p}(r),\qquad\|\hat{\delta}\|^{2}\|\mathcal{G}_{2}\|=o_{p}(r),
\]
both of which are $o_{p}(V)$. This is \eqref{eq:condest}. All conditions of Theorem \ref{thm:null} are therefore satisfied, which proves the null part of Proposition \ref{prop:F}.

\subsubsection*{Step 4. The alternative}

We turn to Theorem \ref{thm:alt}. We first check the hypotheses of Lemma \ref{lem:lln}, which give the first two conditions of \eqref{eq:condalt} and $B^{\mathrm{cf}}_{n}=o(V)$, then the two estimation-error conditions, then the last condition of the theorem. Here $\omega=0$, so the systematic component is the drift itself and the surviving part vanishes identically,
\[
\pi_{i}=\Delta_{i}=(M_{W}Z(\beta_{2}-\beta_{20}))_{i},\qquad\lambda_{i}=m^{\prime}_{i}\omega=0.
\]
By Step 2, $M\Delta=0$. This is the application in which the cross-fit bias reduces furthest: by Lemma \ref{lem:bias} only the terms carrying $(M\Sigma M)_{ij}$ and $M_{ij}$ survive.

We check the hypotheses of Lemma \ref{lem:lln}. As before $L=M$ is symmetric idempotent, so $\|L\|\le1$; the sixth moments are Assumption \ref{as:F}(i); \eqref{eq:cond} was verified above; and $\xi_{n}=\max_{i}|\Delta_{i}|=O_{p}(1)$ by Assumption~\ref{as:F}(iv), $\lambda$ being zero. For \eqref{eq:condagg}, $c_{ij}=2Q^{2}_{ij}/(M_{ii}M_{jj}+M^{2}_{ij})$ with $M_{ii}M_{jj}+M^{2}_{ij}\ge c^{2}_{M}$, so Assumption~\ref{as:F}(iv) gives $\mathcal{C}_{2}\le C\sum^{n}_{i=1}\sum_{j\neq i}Q^{4}_{ij}=o_{p}(r)$, which is $o_{p}(A^{2}_{n})$ because $A_{n}\ge(1-c_{P})r\to\infty$. Next $\bar{c}_{j}\le C\sum_{i}Q^{2}_{ij}=CQ_{jj}\le Cc_{P}$, so $\max_{j}\bar{c}_{j}=O(1)=o(A_{n})$. For the Frobenius quantities, $\operatorname{tr}(P_{X})=d$ and $M_{ij}=-(P_{X})_{ij}$ for $i\neq j$; then $L^{2}_{F}\le d$ and, exactly as in the proof of Proposition \ref{prop:spec} with $P_{X}$ in place of $P$, $S_{F}\le C\sqrt{d}$. Since $d=r+\dim(w_{i})$ and $r\to\infty$, $d/r\to1$, whence
\[
\mathcal{C}^{1/2}_{2}(L_{F}+S_{F})=o_{p}(r^{1/2})O(d^{1/2})=o_{p}(r)=o_{p}(A_{n}).
\]
Lemma \ref{lem:lln} therefore gives the first two conditions of \eqref{eq:condalt} and $B^{\mathrm{cf}}_{n}=o(V)$.

The estimation error is the easiest of the three applications, because the rate $\|\hat{\delta}\|=O_{p}(n^{-1/2})$ established above holds under the alternative as well as under the null. Since $MW=0$ we have $\ell^{\prime}_{i}H=m^{\prime}_{i}W=0$, so Lemma \ref{lem:exp} gives $g_{2i}=0$ and $\mathcal{R}_{n}=0$, while $g_{0i}=\check{A}_{i}$ and $g_{1i}=(m^{\prime}_{i}\varepsilon)w_{i}$; indeed $m^{\prime}_{i}\hat{u}=m^{\prime}_{i}\varepsilon$ exactly, under the alternative as under the null, because $M\Delta=0$ and $MW=0$. Write $\bar{c}_{i}=\sum_{j\neq i}|c_{ij}|\le C\sum_{j}Q^{2}_{ij}=CQ_{ii}\le Cc_{P}$, so that $\sum_{i}\bar{c}_{i}=C_{n}=O_{p}(r)$. By \eqref{eq:wcs},
\[
\|\mathcal{G}_{1}\|\le C\Bigl(\sum_{i}\bar{c}_{i}\check{A}^{2}_{i}\Bigr)^{1/2}\Bigl(\sum_{j}\bar{c}_{j}(m^{\prime}_{j}\varepsilon)^{2}\|w_{j}\|^{2}\Bigr)^{1/2},\qquad\|\mathcal{G}_{2}\|\le C\sum_{j}\bar{c}_{j}(m^{\prime}_{j}\varepsilon)^{2}\|w_{j}\|^{2}.
\]
Now $\mathbb{E}[\check{A}^{2}_{i}|X]\le C(1+\Delta^{2}_{i})\le C$ by Assumption~\ref{as:F}(iv) and $\mathbb{E}[(m^{\prime}_{i}\varepsilon)^{2}|X]\le C$ because $\sum_{k}M^{2}_{ik}=M_{ii}\le1$, so the first sum is $O_{p}(C_{n})=O_{p}(r)$ and the second is at most $C\max_{j}\bar{c}_{j}\sum_{j}\|w_{j}\|^{2}=O_{p}(n)$, using $\mathbb{E}\|w_{i}\|^{2}<\infty$. Hence $\|\hat{\delta}\|\|\mathcal{G}_{1}\|=O_{p}(n^{-1/2})O_{p}(\sqrt{rn})=O_{p}(\sqrt{r})$ and $\|\hat{\delta}\|^{2}\|\mathcal{G}_{2}\|=O_{p}(1)$; since $V\ge2c^{2}A_{n}\ge2c^{2}(1-c_{P})r\to\infty$, both are $o_{p}(V)$ and $\hat{V}_{\mathrm{cf}}-\check{V}_{\mathrm{cf}}=o_{p}(V)$.

For the plug-in estimator, \eqref{eq:Vhat} is \eqref{eq:Vcf} with $\ell_{i}$ the $i$th unit vector and $2a^{2}_{ij}=2Q^{2}_{ij}$ in place of $c_{ij}$, so Lemma \ref{lem:exp} applies with $L=I$ and gives $g_{0i}=\check{u}^{2}_{i}$, $g_{1i}=2\check{u}_{i}w_{i}$ and $g_{2i}=w_{i}w^{\prime}_{i}$, the weights having row sums $\bar{a}_{i}=2(Q_{ii}-Q^{2}_{ii})\le2c_{P}$ and $\sum_{i}\bar{a}_{i}=2A_{n}\le2r$. Three sums bound all four forms:
\[
\sum_{i}\bar{a}_{i}\check{u}^{4}_{i}=O_{p}(r),\qquad\sum_{i}\bar{a}_{i}\check{u}^{2}_{i}\|w_{i}\|^{2}=O_{p}(n),\qquad\sum_{i}\bar{a}_{i}\|w_{i}\|^{4}=o_{p}(n^{2}),
\]
the first because $\mathbb{E}[\check{u}^{4}_{i}|X]\le C$ by Assumption~\ref{as:F}(i) and (iv) while $\sum_{i}\bar{a}_{i}\le2r$, the second because $\bar{a}_{i}\le2c_{P}$ and $\mathbb{E}\|w_{i}\|^{2}<\infty$, and the third because $\sum_{i}\bar{a}_{i}\|w_{i}\|^{4}\le2c_{P}\max_{i}\|w_{i}\|^{2}\sum_{i}\|w_{i}\|^{2}$ with $\max_{i}\|w_{i}\|^{2}=o_{p}(n)$. The first two bound $\mathcal{G}_{1}$; the second bounds the $g_{1i}g^{\prime}_{1j}$ part of $\mathcal{G}_{2}$ and the first and third its $g_{0i}g_{2j}$ part; the second and third bound $\mathcal{G}_{3}$, and the third bounds $\mathcal{G}_{4}$. Applying \eqref{eq:wcs} to $\mathcal{G}_{1},\ldots,\mathcal{G}_{4}$ gives $\|\mathcal{G}_{1}\|=O_{p}(\sqrt{rn})$, $\|\mathcal{G}_{2}\|=O_{p}(n)+o_{p}(n\sqrt{r})$, $\mathcal{G}_{3}=o_{p}(n^{3/2})$ and $\mathcal{G}_{4}=o_{p}(n^{2})$, whence
\[
\|\hat{\delta}\|\|\mathcal{G}_{1}\|+\|\hat{\delta}\|^{2}\|\mathcal{G}_{2}\|+\|\hat{\delta}\|^{3}\mathcal{G}_{3}+\|\hat{\delta}\|^{4}\mathcal{G}_{4}=O_{p}(\sqrt{r})+o_{p}(\sqrt{r})+o_{p}(1)=o_{p}(r),
\]
so $\hat{V}-\check{V}=o_{p}(V)=o_{p}(V+B_{n})$.

Finally the last condition of Theorem \ref{thm:alt} is Assumption \ref{as:F}(iv). All conditions of Theorems \ref{thm:null} and \ref{thm:alt} therefore hold, which proves Proposition \ref{prop:F}.
\end{proof}

\section{Size-corrected power for Section \ref{sec:calibrated} } \label{subsec:Size-corrected-power}

We report size-corrected power for Section \ref{sec:calibrated} because the raw sizes at $\eta=0$ differ slightly across procedures, all between $3.6\%$ and $4.4\%$. Size correction removes any advantage the cross-fit procedure enjoys merely from being better sized, and isolates the drift-contamination mechanism. Throughout, rejection frequencies use only the replications with a non-negative variance estimate.

Figure~\ref{fig:size_corrected} shows the size-corrected curves at the baseline sample size. The cross-fit advantage survives the correction: for $\phi(X)=\mathrm{age}^{3}$ the peak cross-fit $-$ plug-in gap is $+5.0$pp at $\eta=0.15$, and cross-fit $-$ AS23 is $+4.8$pp there; for the heavier-tailed $\phi(X)=\mathrm{age}^{5}$ the two gaps reach $+6.6$pp and $+4.8$pp. That the gap persists after equalizing size shows it reflects the reduction in drift contamination rather than the mild conservativeness of the plug-in. By $\eta=0.30$ all three procedures approach saturation and the gaps close.

\begin{figure}[h]
\centering \begin{tikzpicture}
\begin{groupplot}[
    group style={
        group size=2 by 2,
        horizontal sep=1.4cm,
        vertical sep=0.6cm,
        xlabels at=edge bottom,
        xticklabels at=edge bottom,
    },
    width=0.5\textwidth,
    xmin=-0.01, xmax=0.31,
    grid=major, grid style={gray!20},
    legend style={font=\scriptsize, draw=gray!50, fill=white, fill opacity=0.95, text opacity=1},
    every axis plot/.append style={very thick},
    tick label style={font=\scriptsize},
    label style={font=\small},
    title style={font=\small},
]
\nextgroupplot[
    height=5.5cm,
    title={$\phi(X)=\mathrm{age^{3}}$},
    ylabel={size-corr.\ rej.\ freq.\ (\%)},
    ymin=0, ymax=1.02,
    ytick={0,0.2,0.4,0.6,0.8,1.0},
    yticklabels={0,20,40,60,80,100},
    legend pos=north west,
]
\addplot[blue!70!black, mark=triangle*] coordinates {(0.000,0.050) (0.025,0.072) (0.050,0.108) (0.075,0.148) (0.100,0.210) (0.125,0.348) (0.150,0.492) (0.200,0.748) (0.300,0.958)};
\addlegendentry{AS23}
\addplot[orange!90!black, mark=square*] coordinates {(0.000,0.050) (0.025,0.066) (0.050,0.110) (0.075,0.146) (0.100,0.206) (0.125,0.348) (0.150,0.490) (0.200,0.746) (0.300,0.954)};
\addlegendentry{plug-in}
\addplot[red!85!black, mark=*] coordinates {(0.000,0.050) (0.025,0.072) (0.050,0.116) (0.075,0.166) (0.100,0.224) (0.125,0.382) (0.150,0.540) (0.200,0.776) (0.300,0.968)};
\addlegendentry{cross-fit}
\draw[dotted, gray!70] (axis cs:-0.01,0.05) -- (axis cs:0.31,0.05);
\nextgroupplot[
    height=5.5cm,
    title={$\phi(X)=\mathrm{age^{5}}$},
    ymin=0, ymax=1.02,
    ytick={0,0.2,0.4,0.6,0.8,1.0},
    yticklabels={0,20,40,60,80,100},
]
\addplot[blue!70!black, mark=triangle*] coordinates {(0.000,0.050) (0.025,0.064) (0.050,0.156) (0.075,0.330) (0.100,0.484) (0.125,0.626) (0.150,0.702) (0.200,0.856) (0.300,0.976)};
\addplot[orange!90!black, mark=square*] coordinates {(0.000,0.050) (0.025,0.060) (0.050,0.152) (0.075,0.316) (0.100,0.470) (0.125,0.618) (0.150,0.680) (0.200,0.846) (0.300,0.982)};
\addplot[red!85!black, mark=*] coordinates {(0.000,0.050) (0.025,0.074) (0.050,0.186) (0.075,0.374) (0.100,0.524) (0.125,0.674) (0.150,0.746) (0.200,0.872) (0.300,0.990)};
\draw[dotted, gray!70] (axis cs:-0.01,0.05) -- (axis cs:0.31,0.05);
\nextgroupplot[
    height=4cm,
    xlabel={$\eta = \|\beta_2\|$},
    ylabel={\small power gap (pp)},
    ymin=-0.02, ymax=0.10,
    ytick={0.00, 0.04, 0.08}, yticklabels={0,4,8},
    legend pos=north west,
]
\addplot[orange!90!black, mark=square*] coordinates {(0.000,0.000) (0.025,0.006) (0.050,0.006) (0.075,0.020) (0.100,0.018) (0.125,0.034) (0.150,0.050) (0.200,0.030) (0.300,0.014)};
\addlegendentry{cross-fit $-$ plug-in}
\addplot[blue!70!black, mark=triangle*, dashed] coordinates {(0.000,0.000) (0.025,0.000) (0.050,0.008) (0.075,0.018) (0.100,0.014) (0.125,0.034) (0.150,0.048) (0.200,0.028) (0.300,0.010)};
\addlegendentry{cross-fit $-$ AS23}
\draw[dotted, gray!70] (axis cs:-0.01,0) -- (axis cs:0.31,0);
\nextgroupplot[
    height=4cm,
    xlabel={$\eta = \|\beta_2\|$},
    ymin=-0.02, ymax=0.10,
    ytick={0.00, 0.04, 0.08}, yticklabels={0,4,8},
]
\addplot[orange!90!black, mark=square*] coordinates {(0.000,0.000) (0.025,0.014) (0.050,0.034) (0.075,0.058) (0.100,0.054) (0.125,0.056) (0.150,0.066) (0.200,0.026) (0.300,0.008)};
\addplot[blue!70!black, mark=triangle*, dashed] coordinates {(0.000,0.000) (0.025,0.010) (0.050,0.030) (0.075,0.044) (0.100,0.040) (0.125,0.048) (0.150,0.044) (0.200,0.016) (0.300,0.014)};
\draw[dotted, gray!70] (axis cs:-0.01,0) -- (axis cs:0.31,0);
\end{groupplot}
\end{tikzpicture}\caption{Size-corrected power for the three many-restriction procedures (AS23, plug-in, cross-fit) under two heavy-tail HTE directions: $\phi(X)=\mathrm{age}^{3}$ (left, $r=27$) and $\phi(X)=\mathrm{age}^{5}$ (right, $r=28$). Bernoulli-LPM errors, $n=1500$, $500$ repetitions per $\eta$. Top row: size-corrected rejection frequency vs.$\eta=\|\beta_{2}\|$. Bottom row: the power gap of cross-fit over plug-in and over AS23. Rejection frequencies use only the repetitions with non-negative variance.}
\label{fig:size_corrected} 
\end{figure}
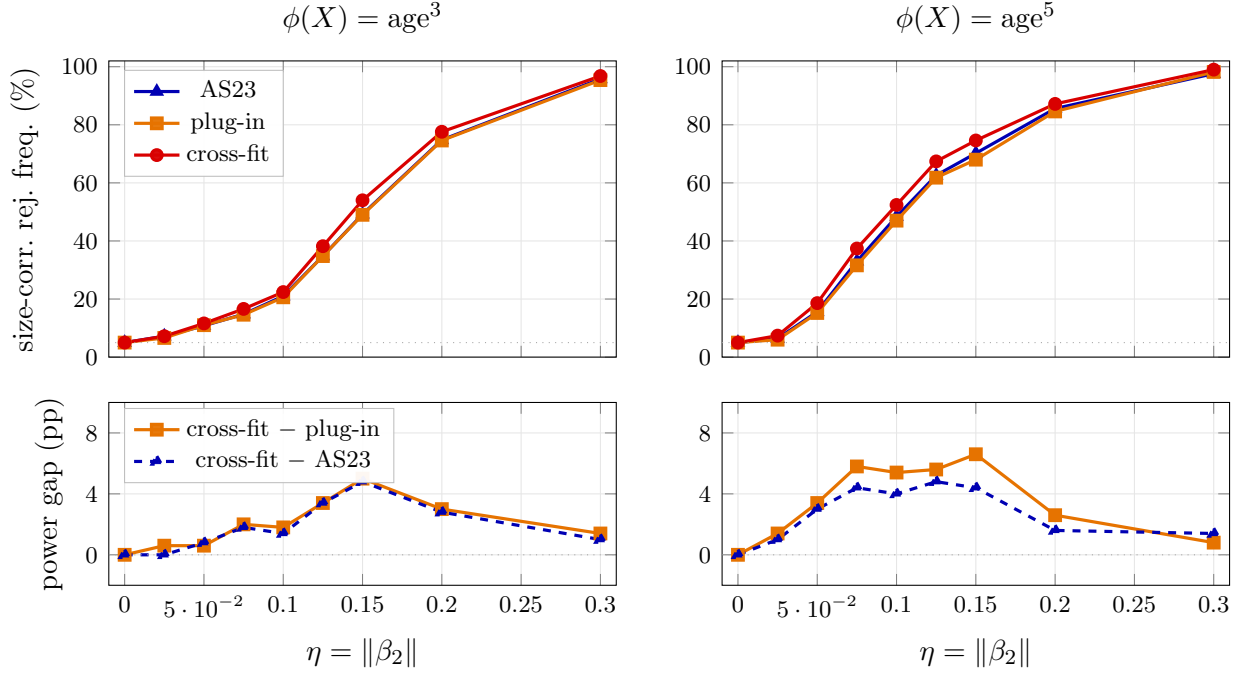

Table~\ref{tab:smalln} and Figure~\ref{fig:smalln} repeat the exercise across the four sample sizes. The table reports frequencies at the representative alternative $\eta=0.15$, together with the largest number of negative variance estimates observed at any point of the $\eta$-grid; the figure gives the differences over the whole grid. The $n=1{,}500$ column of the table is the baseline gap above. Rejection frequencies fall as $n$ decreases at the fixed signal scale $\eta=0.15$, since the aggregate signal in the statistic weakens. The gap over the plug-in procedure depends on the shape of the alternative: under the cubic alternative it falls from $5.0$ percentage points at $r/n=1.8\%$ to about zero at $5.4\%$, while under the fifth-order alternative it stays between $4.0$ and $7.0$ points.

\begin{table}[t]
\caption{Size-corrected rejection frequencies as $r/n$ increases}
\label{tab:smalln} \centering {\small{}%
\begin{tabular}{lrrrrrrrr}
\hline 
 &  & \multicolumn{3}{c}{{\small Rejection frequency at $\eta=0.15$}} & \multicolumn{2}{c}{{\small Difference}} & \multicolumn{2}{c}{{\small Max. negative estimates}}\tabularnewline
\hline 
{\small Design } & {\small$r/n$ } & {\small AS23 } & {\small Plug-in } & {\small Cross-fit } & {\small CF--plug-in } & {\small CF--AS23 } & {\small AS23 } & {\small Cross-fit }\tabularnewline
\hline 
{\small$\mathrm{age}^{3},\ n=1{,}500$ } & {\small 1.8\% } & {\small 0.492 } & {\small 0.490 } & {\small 0.540 } & {\small 0.050 } & {\small 0.048 } & {\small 19 } & {\small 0 }\tabularnewline
{\small$\mathrm{age}^{3},\ n=1{,}000$ } & {\small 2.7\% } & {\small 0.270 } & {\small 0.256 } & {\small 0.294 } & {\small 0.038 } & {\small 0.024 } & {\small 17 } & {\small 0 }\tabularnewline
{\small$\mathrm{age}^{3},\ n=750$ } & {\small 3.6\% } & {\small 0.207 } & {\small 0.220 } & {\small 0.236 } & {\small 0.016 } & {\small 0.029 } & {\small 18 } & {\small 0 }\tabularnewline
{\small$\mathrm{age}^{3},\ n=500$ } & {\small 5.4\% } & {\small 0.136 } & {\small 0.160 } & {\small 0.162 } & {\small 0.002 } & {\small 0.026 } & {\small 26 } & {\small 0 }\tabularnewline
\hline 
{\small$\mathrm{age}^{5},\ n=1{,}500$ } & {\small 1.8\% } & {\small 0.702 } & {\small 0.680 } & {\small 0.746 } & {\small 0.066 } & {\small 0.044 } & {\small 26 } & {\small 0 }\tabularnewline
{\small$\mathrm{age}^{5},\ n=1{,}000$ } & {\small 2.8\% } & {\small 0.510 } & {\small 0.496 } & {\small 0.566 } & {\small 0.070 } & {\small 0.056 } & {\small 33 } & {\small 0 }\tabularnewline
{\small$\mathrm{age}^{5},\ n=750$ } & {\small 3.7\% } & {\small 0.323 } & {\small 0.334 } & {\small 0.382 } & {\small 0.048 } & {\small 0.059 } & {\small 40 } & {\small 0 }\tabularnewline
{\small$\mathrm{age}^{5},\ n=500$ } & {\small 5.6\% } & {\small 0.189 } & {\small 0.178 } & {\small 0.218 } & {\small 0.040 } & {\small 0.029 } & {\small 50 } & {\small 0 }\tabularnewline
\hline 
\end{tabular}}{\small\par}

{\small\medskip{}
}{\small{}%
\begin{minipage}[c]{0.96\textwidth}%
{\footnotesize Notes: Rejection frequencies are size-corrected and evaluated at $\eta=0.15$. The final two columns report the maximum, over the reported $\eta$-grid, number of negative variance estimates among 500 replications.}%
\end{minipage}}{\small\par}
\end{table}

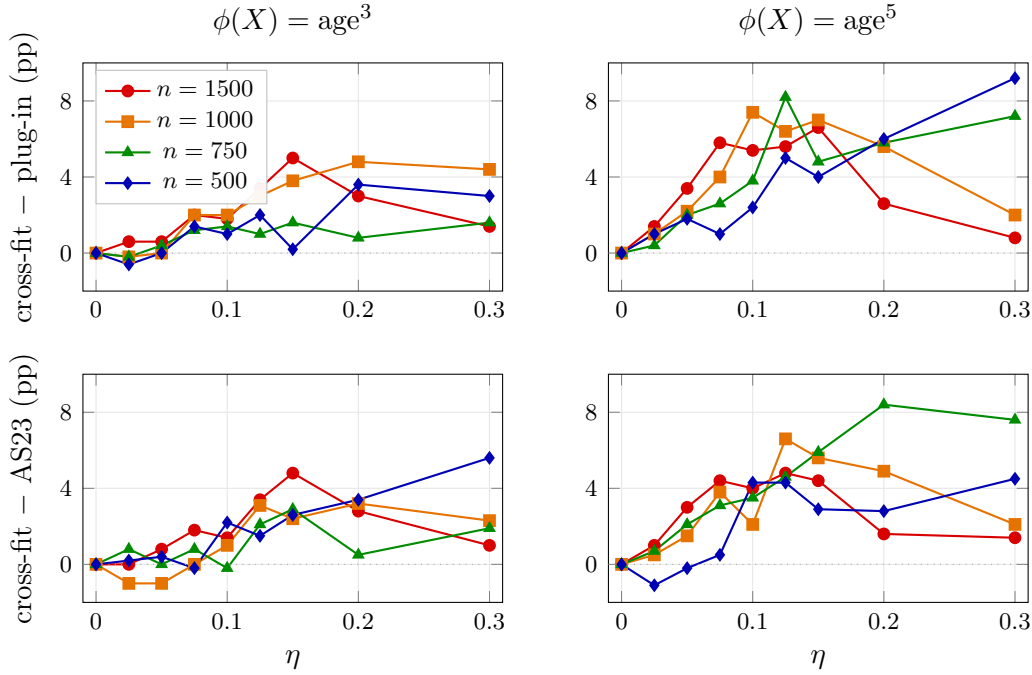
\begin{figure}[t]
\centering \begin{tikzpicture}
\begin{groupplot}[
    group style={
        group size=2 by 2,
        horizontal sep=1.4cm,
        vertical sep=1.1cm,
        xlabels at=edge bottom
    },
    width=0.43\textwidth,
    height=4.6cm,
    xmin=-0.01,
    xmax=0.31,
    ymin=-0.02,
    ymax=0.10,
    ytick={0.00,0.04,0.08},
    yticklabels={0,4,8},
    grid=major,
    grid style={gray!20},
    legend style={
        font=\scriptsize,
        draw=gray!50,
        fill=white,
        fill opacity=0.95,
        text opacity=1
    },
    every axis plot/.append style={thick},
    tick label style={font=\scriptsize},
    label style={font=\small},
    title style={font=\small},
    cycle list={
        {red!85!black,mark=*},
        {orange!90!black,mark=square*},
        {green!55!black,mark=triangle*},
        {blue!70!black,mark=diamond*}
    }
]

\nextgroupplot[
    title={\(\phi(X)=\mathrm{age}^{3}\)},
    ylabel={cross-fit \(-\) plug-in (pp)},
    legend pos=north west
]
\addplot coordinates {
(0.000,0.000)
(0.025,0.006)
(0.050,0.006)
(0.075,0.020)
(0.100,0.018)
(0.125,0.034)
(0.150,0.050)
(0.200,0.030)
(0.300,0.014)
};
\addlegendentry{\(n=1500\)}

\addplot coordinates {
(0.000,0.000)
(0.025,-0.002)
(0.050,0.000)
(0.075,0.020)
(0.100,0.020)
(0.125,0.030)
(0.150,0.038)
(0.200,0.048)
(0.300,0.044)
};
\addlegendentry{\(n=1000\)}

\addplot coordinates {
(0.000,0.000)
(0.025,-0.002)
(0.050,0.004)
(0.075,0.012)
(0.100,0.014)
(0.125,0.010)
(0.150,0.016)
(0.200,0.008)
(0.300,0.016)
};
\addlegendentry{\(n=750\)}

\addplot coordinates {
(0.000,0.000)
(0.025,-0.006)
(0.050,0.000)
(0.075,0.014)
(0.100,0.010)
(0.125,0.020)
(0.150,0.002)
(0.200,0.036)
(0.300,0.030)
};
\addlegendentry{\(n=500\)}

\draw[dotted,gray!70]
(axis cs:-0.01,0) -- (axis cs:0.31,0);

\nextgroupplot[
    title={\(\phi(X)=\mathrm{age}^{5}\)}
]
\addplot coordinates {
(0.000,0.000)
(0.025,0.014)
(0.050,0.034)
(0.075,0.058)
(0.100,0.054)
(0.125,0.056)
(0.150,0.066)
(0.200,0.026)
(0.300,0.008)
};

\addplot coordinates {
(0.000,0.000)
(0.025,0.010)
(0.050,0.022)
(0.075,0.040)
(0.100,0.074)
(0.125,0.064)
(0.150,0.070)
(0.200,0.056)
(0.300,0.020)
};

\addplot coordinates {
(0.000,0.000)
(0.025,0.004)
(0.050,0.020)
(0.075,0.026)
(0.100,0.038)
(0.125,0.082)
(0.150,0.048)
(0.200,0.058)
(0.300,0.072)
};

\addplot coordinates {
(0.000,0.000)
(0.025,0.010)
(0.050,0.018)
(0.075,0.010)
(0.100,0.024)
(0.125,0.050)
(0.150,0.040)
(0.200,0.060)
(0.300,0.092)
};

\draw[dotted,gray!70]
(axis cs:-0.01,0) -- (axis cs:0.31,0);

\nextgroupplot[
    xlabel={\(\eta\)},
    ylabel={cross-fit \(-\) AS23 (pp)}
]
\addplot coordinates {
(0.000,0.000)
(0.025,0.000)
(0.050,0.008)
(0.075,0.018)
(0.100,0.014)
(0.125,0.034)
(0.150,0.048)
(0.200,0.028)
(0.300,0.010)
};

\addplot coordinates {
(0.000,0.000)
(0.025,-0.010)
(0.050,-0.010)
(0.075,0.000)
(0.100,0.010)
(0.125,0.031)
(0.150,0.024)
(0.200,0.032)
(0.300,0.023)
};

\addplot coordinates {
(0.000,0.000)
(0.025,0.008)
(0.050,0.000)
(0.075,0.008)
(0.100,-0.002)
(0.125,0.021)
(0.150,0.029)
(0.200,0.005)
(0.300,0.019)
};

\addplot coordinates {
(0.000,0.000)
(0.025,0.002)
(0.050,0.004)
(0.075,-0.002)
(0.100,0.022)
(0.125,0.015)
(0.150,0.026)
(0.200,0.034)
(0.300,0.056)
};

\draw[dotted,gray!70]
(axis cs:-0.01,0) -- (axis cs:0.31,0);

\nextgroupplot[
    xlabel={\(\eta\)}
]
\addplot coordinates {
(0.000,0.000)
(0.025,0.010)
(0.050,0.030)
(0.075,0.044)
(0.100,0.040)
(0.125,0.048)
(0.150,0.044)
(0.200,0.016)
(0.300,0.014)
};

\addplot coordinates {
(0.000,0.000)
(0.025,0.005)
(0.050,0.015)
(0.075,0.038)
(0.100,0.021)
(0.125,0.066)
(0.150,0.056)
(0.200,0.049)
(0.300,0.021)
};

\addplot coordinates {
(0.000,0.000)
(0.025,0.007)
(0.050,0.021)
(0.075,0.031)
(0.100,0.035)
(0.125,0.046)
(0.150,0.059)
(0.200,0.084)
(0.300,0.076)
};

\addplot coordinates {
(0.000,0.000)
(0.025,-0.011)
(0.050,-0.002)
(0.075,0.005)
(0.100,0.043)
(0.125,0.043)
(0.150,0.029)
(0.200,0.028)
(0.300,0.045)
};

\draw[dotted,gray!70]
(axis cs:-0.01,0) -- (axis cs:0.31,0);

\end{groupplot}
\end{tikzpicture}

\caption{Size-corrected rejection-frequency differences as $r/n$ increases. The columns correspond to $\phi(X)=\mathrm{age}^{3}$ ($r=27$) and $\phi(X)=\mathrm{age}^{5}$ ($r=28$). The top row reports the cross-fit rejection frequency minus the plug-in rejection frequency, and the bottom row reports the cross-fit rejection frequency minus the Anatolyev--S{ø}lvsten rejection frequency. Results are based on 500 Monte Carlo replications for each pair $(\eta,n)$ and are computed over replications in which the corresponding variance estimates are nonnegative.}
\label{fig:smalln} 
\end{figure}

\newpage{}

\printbibliography

\end{document}